%% file: main.tex
\documentclass{llncs}
\input{packages}
\iftech
\else
\usepackage[firstpage]{draftwatermark}\SetWatermarkText{\hspace*{5.2in}\raisebox{5in}{\includegraphics[scale=0.1]{aec-badge-cav}}}\SetWatermarkAngle{0}
\fi

\begin{document}
  \input{title}
  \input{00-abstract}

  \input{01-introduction}
  \input{02-preliminaries}

  \input{03-multiobjective}

  \input{04-ma}
  \input{05-evaluation}

  \input{07-conclusion}
  \bibliographystyle{splncs}
  \bibliography{literature}
  
  \iftech
  \clearpage
  \appendix
  \input{E-extra-definitions}
  \input{G-proofsmomc}

  \input{A-proofsunbounded}
  \input{B-proofsreward}
  \input{C-proofsbounded}
  \input{D-singleObj}
  \input{F-evaluation-details}
  \fi
\end{document}

%% file: packages.tex
\usepackage{amsmath,amssymb,amsfonts}
\usepackage{mathtools}
\usepackage{thm-restate}
\usepackage[usenames,dvipsnames]{color}
\usepackage{xcolor}
\usepackage{xspace}
\usepackage{microtype}
\usepackage{todonotes}
\usepackage{colonequals}
\usepackage{wasysym}
 \usepackage{booktabs}

\usepackage{paralist}
\usepackage{multirow}
\usepackage{multicol}

\usepackage{refcount}
\usepackage{paralist}

\usepackage{subfigure}
\usepackage{url}
\usepackage{appendix}
\usepackage{graphicx}
\usepackage{float}
\usepackage{algorithm}
\usepackage{listings}
\usepackage{nicefrac}
\usepackage{tikz} 
\usepackage{pgfplots}
\usepackage{marvosym}
\usepackage{wrapfig}
\usepackage{hyperref}

\usepackage{footnote}
\usepackage{tablefootnote}
\usepackage{array}

\usetikzlibrary{arrows,decorations.pathmorphing,positioning,fit,trees,shapes,shadows,automata,calc,patterns} 

\usepackage{macros} % self-defined macros

%% file: title.tex
\title{Markov Automata with Multiple Objectives}

%\author{Authors omitted for Review}

\author{Tim Quatmann  \and  Sebastian Junges \and Joost-Pieter Katoen}

 \institute{RWTH Aachen University, Aachen, Germany}

\maketitle

%% file: 00-abstract.tex
\begin{abstract}
Markov automata combine non-determinism, probabilistic branching, and exponentially distributed delays.
This compositional variant of continuous-time Markov decision processes is used in reliability engineering, performance evaluation and stochastic scheduling.
Their verification so far focused on single objectives such as (timed) reachability, and expected costs.
In practice, often the objectives are mutually dependent and the aim is to reveal trade-offs.
We present algorithms to analyze several objectives simultaneously and approximate Pareto curves.
This includes, e.g., several (timed) reachability objectives, or various expected cost objectives. We also consider combinations thereof, such as on-time-within-budget objectives---which policies guarantee reaching a goal state within a deadline with at least probability $p$ while keeping the allowed average costs below a threshold?
We adopt existing approaches for classical Markov decision processes.
The main challenge is to treat policies exploiting state residence times, even for \emph{un}timed objectives.
Experimental results show the feasibility and scalability of our approach.
\end{abstract}

%% file: 01-introduction.tex
\section{Introduction}
\label{sec:introduction}

%% background: what are MA, why are they important, and which single objectives can be analysed? 
Markov automata~\cite{DBLP:conf/lics/EisentrautHZ10,DBLP:journals/iandc/DengH13} extend labeled transition systems with probabilistic branching and exponentially distributed delays.
They are a compositional variant of continuous-time Markov decision processes (CTMDPs), in a similar vein as Segala's probabilistic automata extend classical MDPs.
Transitions of a Markov automaton (MA) lead from states to probability distributions over states, and are either labeled with actions (allowing for interaction) or real numbers (rates of exponential distributions).
MAs are used in reliability engineering~\cite{DBLP:journals/tdsc/BoudaliCS10}, hardware design~\cite{DBLP:conf/cav/CosteHLS09}, data-flow computation~\cite{KatoenWu16}, dependability~\cite{DBLP:journals/cj/BozzanoCKNNR11} and performance evaluation~\cite{DBLP:conf/apn/EisentrautHK013}, as MAs are a natural semantic framework for modeling formalisms such as AADL, dynamic fault trees, stochastic Petri nets, stochastic activity networks, SADF etc.
The verification of MAs so far focused on single objectives such as reachability, timed reachability, expected costs, and long-run averages~\cite{HatefiH12,DBLP:journals/corr/GuckHHKT14,DBLP:conf/atva/GuckTHRS14,DBLP:conf/setta/HatefiBWFHB15,Wimmeretal17}.
These analyses cannot treat objectives that are mutually influencing each other, like quickly reaching a target is more costly.
The aim of this paper is to analyze \emph{multiple} objectives on MAs at once and to facilitate \emph{trade--off analysis} by approximating Pareto curves.

%% what does multi-objective model checking of MA bring us?
Consider the stochastic job scheduling problem of~\cite{DBLP:journals/jacm/BrunoDF81}: perform $n$ jobs with exponential service times on $k$ identical processors under a pre-emptive scheduling policy. 
Once a job finishes, all $k$ processors can be assigned any of the $m$ remaining jobs. 
When $n{-}m$ jobs are finished, this yields $\binom{m}{k}$ non-deterministic choices.
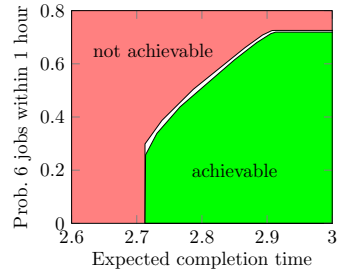
\begin{wrapfigure}[11]{r}{0.34\textwidth}
	\centering
	\scalebox{\picscale}{
		\input{pics/intro_plot}
	}
	\caption{Approx.\ Pareto curve for stochastic job scheduling.}
	\label{fig:intro:achievable_plot}
\end{wrapfigure}
The largest-expected-service-time-first-policy is optimal to minimize the expected time to complete all jobs~\cite{DBLP:journals/jacm/BrunoDF81}. 
It is unclear how to schedule when imposing \emph{extra} constraints, e.g., requiring a high probability to finish a batch of $c$ jobs within a tight deadline (to accelerate their post-processing), or having a low average waiting time.
These \emph{multiple objectives} involve non-trivial \emph{trade--offs}. % which can be analysed with our algorithms. 
Our algorithms analyze such trade--offs.
Fig.~\ref{fig:intro:achievable_plot}, e.g., shows the obtained result for 12 jobs and 3 processors. 
It approximates the set of points $(\pointi{1}, \pointi{2})$ for schedules achieving that (1) the expected time to complete all jobs is at most $\pointi{1}$ and (2) the probability to finish half of the jobs within an hour is at least $\pointi{2}$.

%% \paragraph{Contributions of this paper.}
This paper presents techniques to verify MAs with multiple objectives.
We consider multiple (un)timed reachability and expected reward objectives as well as their combinations.
%Mixtures of lower and upper bounded objectives are admitted.
Put shortly, we reduce all these problems to instances of multi-objective verification problems on classical MDPs.
For multi-objective queries involving (combinations of) untimed reachability and expected reward objectives, corresponding algorithms on the \emph{underlying} MDP can be used.
In this case, the MDP is simply obtained by ignoring the timing information, see Fig.~\ref{fig:models:umdp}.
The crux is in relating MA schedulers---that can exploit state sojourn times to optimize their decisions---to MDP schedulers. 
For multiple timed reachability objectives, \emph{digitization}~\cite{HatefiH12,DBLP:journals/corr/GuckHHKT14} is employed to obtain an MDP, see Fig.~\ref{fig:models:dma}.
The key is to mimic sojourn times  by self-loops with appropriate probabilities.  
This provides a sound arbitrary close approximation of the timed behavior and also allows to combine timed reachability objectives with other types of objectives. 
The main contribution is to show that digitization is sound for \emph{all} possible MA schedulers.
%Our more general results require a new proof strategy 
This requires a new proof strategy as the existing ones are tailored to optimizing a single objective.
\iftech
All proofs can be found in the appendix.
\else
All proofs can be found in an extended version~\cite{TR}.
\fi
Experiments on instances of four MA benchmarks show encouraging results.
Multiple untimed reachability and expected reward objectives can be efficiently treated for models with millions of states.
As for single objectives \cite{DBLP:journals/corr/GuckHHKT14}, timed reachability is more expensive.
Our implementation is competitive to \prism~for multi-objective MDPs~\cite{KNP11,ForejtKPatva12} and to \imca~\cite{DBLP:journals/corr/GuckHHKT14} for single-objective MAs.

\paragraph{Related work.}
Multi-objective decision making for MDPs with discounting and long-run objectives has been well investigated; for a recent survey, see~\cite{DBLP:journals/jair/RoijersVWD13}.
Etessami \emph{et al.}~\cite{EtessamiKVY08} consider verifying finite MDPs with multiple $\omega$-regular objectives.
Other multiple objectives include expected rewards under worst-case reachability~\cite{ForejtKNPQtacas11,DBLP:conf/stacs/BruyereFRR14}, quantiles and conditional probabilities~\cite{DBLP:conf/csl/BaierDK14}, mean pay-offs and stability~\cite{DBLP:journals/jcss/BrazdilCFK17}, long-run objectives~\cite{DBLP:journals/corr/abs-1104-3489,DBLP:conf/tacas/BassetKTW15}, total average discounted rewards under PCTL~\cite{DBLP:conf/ecai/Teichteil-Konigsbuch12}, and stochastic shortest path objectives~\cite{DBLP:conf/vmcai/RandourRS15}.
This has been extended to MDPs with unknown cost function~\cite{DBLP:conf/tacas/Junges0DTK16}, infinite-state MDPs~\cite{DBLP:conf/atva/DavidJLLLST14} arising from two-player timed games in a stochastic environment, and stochastic two-player games~\cite{DBLP:conf/mfcs/ChenFKSW13}.
To the best of our knowledge, this is the first work on multi-objective MDPs extended with \emph{random timing}.

\begin{figure}[t]
	\centering
	\subfigure[MA $\ma$.]{
		\scalebox{\picscale}{
			\input{pics/models_ma}
		}
		\label{fig:models:ma}
	}
	\subfigure[Underlying MDP $\umdp$.]{
		\scalebox{\picscale}{
			\input{pics/models_umdp}
		}
		\label{fig:models:umdp}
	}
	\subfigure[Digitization $\dma$.]{
		\scalebox{\picscale}{
			\input{pics/models_dma}
		}
		\label{fig:models:dma}
	}
	\vspace{-2mm}
	\caption{MA $\ma$ with underlying MDP $\umdp$ and digitization $\dma$.}
	\label{fig:models}
	\vspace{-3mm}
\end{figure}
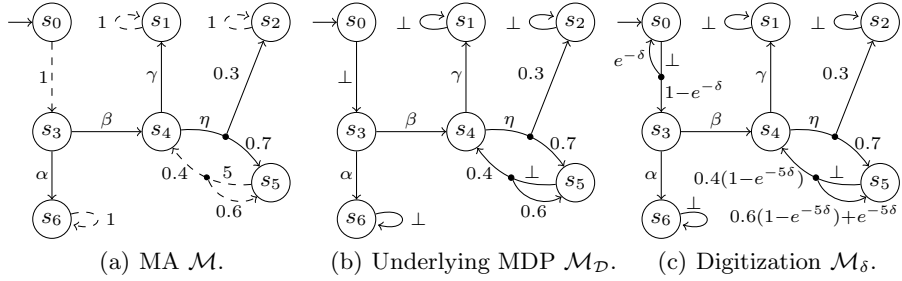

%% file: pics/intro_plot.tex
\pgfplotsset{compat = 1.3} % for label shifting to work
\begin{tikzpicture}[scale=\plotscale]
\draw[white, use as bounding box] (-0.58,-0.4) rectangle (5.1,3.9);  % for right aligned wrapfig
\begin{axis}[
%axis equal image,
    xmin=2.6,
xmax=3.0,
ymin=0.0,
ymax = 0.8,
enlargelimits=false,
width=0.52\textwidth,
xlabel = {Expected completion time},
ylabel shift = -2pt,
xlabel shift = -2pt,
ylabel = {Prob. 6 jobs within 1 hour},
axis on top,
axis background/.style={fill=badArea},
]

\addplot[fill=white, very thin] table [col sep=comma] {pics/intro_plot_overapproximation.csv} -- cycle;
\addplot[fill=goodArea, very thin] table [col sep=comma] {pics/intro_plot_underapproximation.csv} -- cycle;

%\addplot[lineColor, very thick] coordinates  {(0.43, 0.32) (0.58,0.32) (0.58, 0.43)   (0.43, 0.43)} -- cycle ;

%\plotPoint{(0.48,0.35)}{$\tilde{\point}$}
\node at (axis cs:2.85,0.2) {achievable};
\node at (axis cs:2.725,0.65) {not achievable};

\end{axis}
\end{tikzpicture}

%% file: pics/models_ma.tex
\begin{tikzpicture}[scale=1,baseline=(s0.south)]
    \draw[white, use as bounding box] (-0.4,-3.1) rectangle (3.3,0.2); 
    
    \node [state] (s0) at (0,0) {$\state_0$};
    \node [state] (s1) [on grid, right=15mm of s0] {$\state_1$};
    \node [state] (s2) [on grid, right=15mm of s1] {$\state_2$};
    \node [state] (s3) [on grid, below=15mm of s0] {$\state_3$};
    \node [state] (s4) [on grid, right=15mm of s3] {$\state_4$};
    \node [state] (s5) [on grid, below=22mm of s2] {$\state_5$};
    \node [state] (s6) [on grid, below=12mm of s3] {$\state_6$};
    
    \initstateLeft{s0}

  \draw (s0) edge[markovian] node[pos=0.5, left=-2pt] {\scriptsize$1$}  (s3);
  \draw (s1) edge[markovian, loop left] node[] {\scriptsize$1$}  (s1);
  \draw (s2) edge[markovian, loop left] node[] {\scriptsize$1$} (s2);
  \draw (s3) edge[probabilistic] node[left=-2pt] {\scriptsize$\act$} (s6);
  \draw (s3) edge[probabilistic] node[above=-2pt] {\scriptsize$\altact$} (s4);
  \draw (s4) edge[probabilistic] node[left=-2pt] {\scriptsize$\gamma$}  (s1);
  \draw (s4) edge[probabilistic, bend left] node[pos=0.25, above=-2pt] {\scriptsize$\eta$} node[action] (s4eta) {} node [pos=0.65, right] {\scriptsize$0.7$} (s5);
  \draw (s4eta) edge[probabilistic] node[pos=0.65, left=1pt] {\scriptsize$0.3$} (s2);
  \draw (s5) edge[markovian, bend left] node[pos=0.25, above=-1pt] {\scriptsize$5$}  node[action] (s5markovian){} node[pos=0.6, left=1pt] {\scriptsize$0.4$} (s4);
  \draw (s5markovian) edge[markovian, bend right=50] node[below=-1pt, align=left] {\scriptsize$0.6$} (s5);
  \draw (s6) edge[markovian, loop right] node[] {\scriptsize$1$}  (s6);
\end{tikzpicture}

%% file: pics/models_umdp.tex
\begin{tikzpicture}[scale=1,baseline=(s0.south)]
    \draw[white, use as bounding box] (-0.2,-3.1) rectangle (3.3,0.2); 

\node [state] (s0) at (0,0) {$\state_0$};
\node [state] (s1) [on grid, right=15mm of s0] {$\state_1$};
\node [state] (s2) [on grid, right=15mm of s1] {$\state_2$};
\node [state] (s3) [on grid, below=15mm of s0] {$\state_3$};
\node [state] (s4) [on grid, right=15mm of s3] {$\state_4$};
\node [state] (s5) [on grid, below=22mm of s2] {$\state_5$};
\node [state] (s6) [on grid, below=12mm of s3] {$\state_6$};

\initstateLeft{s0}

\draw (s0) edge[probabilistic] node[pos=0.5, left=-2pt] {\scriptsize$\markovianAct$}  (s3);
\draw (s1) edge[probabilistic, loop left] node[] {\scriptsize$\markovianAct$}  (s1);
\draw (s2) edge[probabilistic, loop left] node[] {\scriptsize$\markovianAct$} (s2);
\draw (s3) edge[probabilistic] node[left=-2pt] {\scriptsize$\act$} (s6);
\draw (s3) edge[probabilistic] node[above=-2pt] {\scriptsize$\altact$} (s4);
\draw (s4) edge[probabilistic] node[left=-2pt] {\scriptsize$\gamma$}  (s1);
\draw (s4) edge[probabilistic, bend left] node[pos=0.25, above=-2pt] {\scriptsize$\eta$} node[action] (s4eta) {} node [pos=0.65, right] {\scriptsize$0.7$} (s5);
\draw (s4eta) edge[probabilistic] node[pos=0.65, left=1pt] {\scriptsize$0.3$} (s2);
\draw (s5) edge[probabilistic, bend left] node[pos=0.25, above=-1pt] {\scriptsize$\markovianAct$}  node[action] (s5markovian){} node[pos=0.6, left=1pt] {\scriptsize$0.4$} (s4);
\draw (s5markovian) edge[probabilistic, bend right=50] node[below=-1pt, align=left] {\scriptsize$0.6$} (s5);
\draw (s6) edge[probabilistic, loop right] node[] {\scriptsize$\markovianAct$}  (s6);
\end{tikzpicture}

%% file: pics/models_dma.tex
\begin{tikzpicture}[scale=1,baseline=(s0.south)]
    \draw[white, use as bounding box] (-0.2,-3.1) rectangle (3,0.2); 

\node [state] (s0) at (0,0) {$\state_0$};
\node [state] (s1) [on grid, right=15mm of s0] {$\state_1$};
\node [state] (s2) [on grid, right=15mm of s1] {$\state_2$};
\node [state] (s3) [on grid, below=15mm of s0] {$\state_3$};
\node [state] (s4) [on grid, right=15mm of s3] {$\state_4$};
\node [state] (s5) [on grid, below=22mm of s2] {$\state_5$};
\node [state] (s6) [on grid, below=12mm of s3] {$\state_6$};

\initstateLeft{s0}

\draw (s0) edge[probabilistic] node[pos=0.25, right=-2pt] {\scriptsize$\markovianAct$} node[action] (s0markovian) {} node[pos=0.7, right=-2pt] {\scriptsize$1{-}e^{-\digConstant}$} (s3);
\draw (s0markovian) edge[probabilistic, bend left] node[left=-2pt] {\scriptsize$e^{-\digConstant}$} (s0);
\draw (s1) edge[probabilistic, loop left] node[] {\scriptsize$\markovianAct$} (s1);
\draw (s2) edge[probabilistic, loop left] node[] {\scriptsize$\markovianAct$} (s2);
\draw (s3) edge[probabilistic] node[left=-2pt] {\scriptsize$\act$} (s6);
\draw (s3) edge[probabilistic] node[above=-2pt] {\scriptsize$\altact$} (s4);
\draw (s4) edge[probabilistic] node[left=-2pt] {\scriptsize$\gamma$} (s1);
\draw (s4) edge[probabilistic, bend left] node[pos=0.25, above=-2pt] {\scriptsize$\eta$} node[action] (s4eta) {} node [pos=0.65, right] {\scriptsize$0.7$} (s5);
\draw (s4eta) edge[probabilistic] node[pos=0.65, left=1pt] {\scriptsize$0.3$} (s2);
\draw (s5) edge[probabilistic, bend left] node[pos=0.25, above=-1pt] {\scriptsize$\markovianAct$}  node[action] (s5markovian){} node[pos=0.5, left=1pt] {\scriptsize$0.4(1{-}e^{-5\digConstant})$} (s4);
\draw (s5markovian) edge[probabilistic, bend right=50] node[below=-1pt, align=left] {\scriptsize$0.6(1{-}e^{-5\digConstant}){+} e^{-5\digConstant}$ \qquad} (s5);
\draw (s6) edge[probabilistic, loop right] node[pos=0.15, above=-2pt] {\scriptsize$\markovianAct$} (s6);
\end{tikzpicture}

%% file: 02-preliminaries.tex
\section{Preliminaries}
\label{sec:Preliminaries}
%\unpolished{
%Most of the presented definitions have been adopted from~\cite{HatefiH12,DBLP:journals/corr/GuckHHKT14, DBLP:conf/atva/GuckTHRS14}.
%Further information can also be found in~\cite{Timmer2013,Neuhausser2010}.
%A gentle introduction to the analysis of probabilistic models (and model checking in general) is given in~\cite{BK08}.
%A comprehensive overview of the mathematical backgrounds is provided by~\cite{ash2000probability}.
%}

\paragraph{Notations.} 
The set of real numbers is denoted by $\RR$, and we write $\RRgz = \{x \in \RR \mid x > 0\}$ and  $\RRnn = \RRgz \cup \{0\}$.
 For a finite set $\States$, $\Dist{\States}$ denotes the set of probability distributions over $\States$. 
 $\dist \in \Dist{\States}$ is \emph{Dirac} if $\dist(\state) = 1$ for some $\state \in \States$.

\subsection{Models}	 
Markov automata generalize both Markov decision processes (MDPs) and continuous time Markov chains (CTMCs). They are extended with rewards (or, equivalently, costs) to allow modelling, e.g., energy consumption.
% the time until a failure, or the number of required retries to deliver a message via a lossy channel.
\label{sec:models:mas}
%A \emph{Markov automaton}~\cite{DBLP:conf/lics/EisentrautHZ10} can be interpreted as a transition system with two types of transitions:
%\begin{itemize}
%	\item \emph{Probabilistic transitions} which instantaneously pick a successor state according to a discrete probability distribution over the state space, and
%	\item \emph{Markovian transitions} that lead to a given successor state after an exponentially distributed delay.
%\end{itemize}
%Nondeterminism is incorporated into the formalism by allowing multiple probabilistic transitions (i.e., multiple probability distributions over successor states) originating from the same state.

\begin{definition}[Markov automaton]
	\label{def:MarkovAutomaton}
	A \emph{Markov automaton (MA)} is a tuple $\maDef$ where
		 $\States$ is a finite set of \emph{states} with \emph{initial state} $\sinit \in \States$, $\Actions$ is a finite set of \emph{actions} with $\markovianAct \in \Actions$ and $\Actions \cap \RRnn = \emptyset$,
%		 \jpk{I expect $\Actions \cap \RRgz$ here}
%		 \sj{That doesnt work due to a trick with paths (which can have time stamp 0), and this is also correct}
	\begin{itemize}
		\item ${\probTransRel} \subseteq \States \times (\Actions\cupdot \RRgz) \times \Dist{\States}$ is a set of \emph{transitions}
%		\item ${\markTransRel} \subseteq \States \times \RRgz \times \States$ is a set of \emph{Markovian transitions} 
such that for all $\state \in \States$ there is at most one transition  $(\state, \rate, \dist) \in {\probTransRel}$ with $\rate \in \RRgz$, and
% with we have $|\{ \mid \rate \in \RRgz, \dist \in \Dist{\States}  \}| \le 1$ , and
%there is at most one $\rate \in \RRgz$ with $(\state, \rate, \state') \in {\probTransRel}$, and
		\item $\rewFct[1], \dots, \rewFct[\numRewFunctions]$ with $\numRewFunctions \ge 0$ are \emph{reward functions} $\rewFct[i] \colon \States \cupdot (\States \times \Actions) \to \RRnn$.
	\end{itemize}
\end{definition}
\noindent In the remainder of the paper, let $\maDef$ denote an MA.
A transition $(\state, \actionOrRate, \dist) \in {\probTransRel}$, denoted by $\state \probTrans{\actionOrRate} \dist$, is called \emph{probabilistic} if $\actionOrRate \in \Actions$ and \emph{Markovian} if $\actionOrRate \in \RRgz$.
In the latter case, $\actionOrRate$ is the rate of an exponential distribution, modeling a time-delayed transition.
Probabilistic transitions fire instantaneously.
The successor state is determined  by $\dist$,  i.e., we move to $\state'$ with probability $\dist(\state')$.
Probabilistic (Markovian) states PS (MS) have an outgoing probabilistic (Markovian) transition, respectively: $\PS=\{ \state \in \States \mid \state \probTrans{\act} \dist, \act \in \Actions\}$ and $\MS=\{ \state \in \States \mid \state \probTrans{\rate} \dist, \rate \in \RRgz \}$.
The \emph{exit rate} $\rateAtState{\state}$ of $\state \in \MS$ is uniquely given by $\state \probTrans{\rateAtState{\state}} \dist$.
The \emph{transition probabilities} of $\ma$ are given by the function $\probP \colon \States \times \Actions \times \States \to [0,1]$ satisfying $\probP(\state, \act, \state') = \dist(\state')$  if either  $\state \probTrans{\act} \dist$  or \big($\act = \markovianAct$ and $\state \probTrans{\rateAtState{\state}} \dist$\big) and $\probP(\state,\act, \state')  = 0$ in all other cases.
The value $\probP(\state, \act, \state')$ corresponds to the probability to move from $\state$ with action $\act$ to $\state'$.
The  \emph{enabled actions} at state $\state$ are given by $\Act{\state} = \{ \act \in \Actions \mid \exists \state' \in \States \colon  \probP(\state, \act, \state') > 0 \}$.
\begin{example}
	Fig.~\ref{fig:models:ma} shows an MA $\ma$. We do not depict Dirac probability distributions. Markovian transitions are illustrated by dashed arrows.
\end{example}
We assume \emph{action-deterministic} MAs: $|\{ \dist \in \Dist{\States} \mid \state \probTrans{\act} \dist \}| \le 1$ holds for all $\state \in \States$ and $\act \in \Actions$. 
%
%Moreover, the \emph{maximal progress assumption} is applied:
%As the probability to take a Markovian transition without any delay is zero, probabilistic transitions take precedence over Markovian transitions.
%It follows that Markovian transitions from probabilistic states can be removed. Thus, we assume $\PS \cap \MS = \emptyset$.
Terminal states $\state \notin \PS \cup \MS$ are excluded by adding a Markovian self-loop. %transition leading from $\state$ back to $\state$.
As standard for MAs~\cite{DBLP:conf/lics/EisentrautHZ10,DBLP:journals/iandc/DengH13}, we impose the  \emph{maximal progress assumption}, i.e., probabilistic transitions take precedence over Markovian ones.
Thus, we remove transitions $\state \probTrans{\rate} \dist$ for $\state \in \PS$ and $\rate \in \RRgz$ 
% transition $\state \markTrans{\rate}\state$ where $\rate$ is the minimal rate occurring in $\ma$. 
which yields $\States = \PS \cupdot \MS$.
MAs with \emph{Zeno behavior}, where infinitely many actions can be taken within finite time with non-zero probability, are unrealistic and considered a modeling error. 

A reward function $\rewFct[i]$ defines \emph{state rewards} and \emph{action rewards}.
When sojourning in a state $\state$ for $\ttime$ time units% (due to the delay of Markovian transitions)
, the state reward $\rewState[i]{\state} \cdot \ttime$ is obtained.
Upon taking a transition $\state \probTrans{\actionOrRate} \dist$, we collect action reward $\rewAct[i]{\state}{\actionOrRate}$ (if $\actionOrRate \in \Actions$) or $\rewAct{\state}{\markovianAct}$ (if $\actionOrRate \in \RRgz$).
%Action rewards are collected upon taking a transition.
%A probabilistic transition $\state \probTrans{\act}\dist$ yields reward $\rewAct[i]{\state}{\act}$.
%Taking a Markovian transition  $\state \probTrans{\rate} \dist$ yields $\rewAct[i]{\state}{\markovianAct}$.
For presentation purposes, in the remainder of this section, rewards are omitted. Full definitions with rewards can be found in \tech{App:RewardDefs}.

\begin{definition}[Markov decision process \cite{Put94}]
	\label{def:MarkovDecisionProcess}
	A \emph{Markov decision process} (MDP)  is a tuple $\mdpDefNR$ with
	$\States, \sinit, \Actions$ as in Def.~\ref{def:MarkovAutomaton} and $\probP \colon \States \times \Actions \times \States \to [0,1]$ are the \emph{transition probabilities} satisfying 
$		\sum_{\state' \in \States}\probP(\state, \act, \state') \in \{0,1\}$
 for all $\state \in \States$ and $\act \in \Actions$.	  
\end{definition}
MDPs are MAs without Markovian states and thus without timing aspects, i.e., MDPs exhibit probabilistic branching and non-determinism. Zeno behavior is not a concern, as we do not consider timing aspects.  %An MDP $\mdpDef$ is an MA $\ma = (\States, \Actions, \probTransRel, \emptyset,\allowbreak \sinit,\allowbreak \{\rewMdpFct[1], \dots, \rewMdpFct[\numRewFunctions]\})$ where the set of probabilistic transitions satisfies
%\begin{displaymath}
%\state \probTrans{\act}\dist \iff \dist(\state') = \probP(\state, \act, \state') \text{ for all } \state' \in \States
%\end{displaymath}
%and the reward functions $\rewMdpFct[1], \dots, \rewMdpFct[\numRewFunctions]$ are extended with state rewards such that $\rewState[i]{\state} = 0$ for all $\state \in \States$.
%There is an outgoing transition at state $\state$ with action $\act$ iff it holds that $\sum_{\state' \in \States}\probP(\state, \act, \state') = 1$.
%Deadlock states $s \in \States \setminus \PS$ satisfy $\sum_{\state' \in \States}\probP(\state, \act, \state') = 0$ for all $\act \in \Actions$.
%Deadlock in MDPs are avoided by adding a probabilistic transition $\state \probTrans{\markovianAct} \dist$ with $\dist(\state)=1$. Timing, and thereby Zeno behavior, is not a concern in MDPs.
%\paragraph{Underlying MDP.}
The  \emph{underlying MDP} of an MA abstracts away from its timing:
%To this end, we fix some $\markovianAct \in \Actions$ as a designated action for Markovian transitions.
\begin{definition}[Underlying MDP]
	\label{def:UnderlyingMdp}
%	For MA $\maDefNR$ with transition probability function $\probP$, the underlying MDP is $\umdpDefNR$.
	The MDP $\umdpDefNR$ is the \emph{underlying MDP} of MA $\maDefNR$ with transition probabilities $\probP$.
%	For MA $\maDefNR$ the underlying MDP of $\ma$ is given by $\umdpDefNR$, where
%	$\probP(\state,\act, \state') = \dist(\state')$ if either  $\state \probTrans{\act} \dist$  or \big($\act = \markovianAct$ and $\state \probTrans{\rateAtState{\state}} \dist$\big) and $\probP(\state,\act, \state')  = 0$ in all other cases.
%	\begin{displaymath}
%		\probP(\state,\act, \state') =
%		\begin{cases}
%			\dist(\state')   & \text{if } \state \probTrans{\act} \dist \text{ or } (\act = \markovianAct \text{ and }  \state \probTrans{\rateAtState{\state}} \dist) \\
%%			\nicefrac{\rate}{\rateAtState{\state}}  & \text{if } \act = \markovianAct,  \state \markTrans{\rate} \state'\\
%			0  & \text{otherwise.}
%		\end{cases}% \text{ and $0$ otherwise.}
%	\end{displaymath}
\end{definition}
% Let us fix the underlying MDP $\umdpDefNR$ of $\ma$.
%Note that the function $\probP$ is a valid transition probability function due to the restrictions we made for MAs.
%In particular, the assumption that $\ma$ is action deterministic
%This is directly clear for $\state \in \PS$.
%For $\state \in \MS$, the probability to leave $\state$ via a given transition $\state \markTrans{\rate}\state'$ is given by $\nicefrac{\rate}{\rateAtState{\state}}$.
%\begin{displaymath}
%\int_{0}^{\infty} \frac{\rate}{\rateAtState{\state}} \cdot \rateAtState{\state} \cdot e^{-\rateAtState{\state} \ttime} \diff\ttime 
%= \frac{\rate}{\rateAtState{\state}}
% = \probP(\state, \markovianAct, \state').
%\end{displaymath}
%\paragraph{Digitization.}                           
The \emph{digitization} $\dma$ of $\ma$ w.r.t.\ some digitization constant $\digConstant \in \RRgz$ is an MDP which digitizes the time~\cite{HatefiH12,DBLP:journals/corr/GuckHHKT14}.
The main difference between $\umdp$ and $\dma$ is that the latter also introduces \emph{self-loops} which describe the probability to stay in a Markovian state for $\digConstant$ time units. 
More precisely, the outgoing transitions of states $\state \in \MS$ in $\dma$ represent that either (1) a Markovian transition in $\ma$ was taken within $\digConstant$ time units,
% (happens with probability $ \probP(\state, \markovianAct, \state') \cdot (1-e^{-\rateAtState{\state} \digConstant})$)
 or (2) no transition is taken within $\digConstant$ time units 
% (happens with probability $e^{-\rateAtState{\state} \digConstant}$)
 -- which is captured by taking the self-loop in $\dma$.
Counting the taken self-loops at $\state \in \MS$ allows to approximate the sojourn time in $\state$.
\begin{definition}[Digitization of an MA]
	\label{def:dma}
	For MA $\maDefNR$ with transition probabilities $\probP$ and \emph{digitization constant} $\digConstant \in \RRgz$, the \emph{digitization of $\ma$ \wrt $\digConstant$} is the MDP $\dmaDefNR$ where
	\begin{displaymath}
	\probPdig(\state, \act, \state') =
	\begin{cases}
	\probP(\state, \markovianAct, \state') \cdot (1-e^{-\rateAtState{\state} \digConstant})                                           & \text{if } \state \in \MS, \act = \markovianAct, \state \neq \state'\\
	\probP(\state, \markovianAct, \state') \cdot (1-e^{-\rateAtState{\state} \digConstant}) + e^{-\rateAtState{\state} \digConstant}  & \text{if } \state \in \MS, \act = \markovianAct, \state = \state'\\
		\probP(\state, \act, \state')                                                                                                     & \text{otherwise.}
	\end{cases}
	\end{displaymath}
\end{definition}

\begin{example}
	Fig.~\ref{fig:models} shows an MA $\ma$ with its underlying MDP $\umdp$ and a  digitization $\dma$ for unspecified $\digConstant \in \RRgz$.
\end{example}

\paragraph{Paths and schedulers.}
Paths represent runs of $\ma$ starting in the initial state. % $\sinit$. %, and is an alternating sequence of states and stamps.
%    Paths depict the visited states, the sojourn times, and the performed actions.
%	For $\stamp[] \in \RRnn \cupdot \Actions$ let $\timeOfStamp = \stamp$ iff $\stamp \in \RRnn$ and $0$ otherwise and $\actionOfStamp = \stamp$ iff $\stamp \in \Actions$ and $\markovianAct$ otherwise. 
%For $\stamp \in \Actions$ and $\ttime \in \RRnn$ we set $\ttime(\act) = 0$, $\ttime(\ttime)=\ttime$, $\act(\act) = \act$, and $\act(\ttime) = \markovianAct$.
Let $\timeOfStamp = 0$ and $\actionOfStamp = \stamp$, if $\stamp \in \Actions$, and $\timeOfStamp = \stamp$ and $\actionOfStamp = \markovianAct$, if $\stamp \in \RRnn$.
%	For $\stamp[] \in \RRnn \cupdot \Actions$ let $\timeOfStamp = \stamp$ iff $\stamp \in \RRnn$ (otherwise and $\actionOfStamp = \stamp$ iff $\stamp \in \Actions$ and $\markovianAct$ otherwise. 
%	\begin{displaymath}
%	\timeOfStamp = 
%	\begin{cases}
%	0  & \text{if } \stamp \in \Actions \\
%	\stamp    & \text{if } \stamp \in \RRnn
%	\end{cases}
%	\quad \text{ and } \quad 
%	\actionOfStamp = 
%	\begin{cases}
%	\stamp  & \text{if } \stamp \in \Actions \\
%	\markovianAct    & \text{if } \stamp \in \RRnn \ .
%	\end{cases}
%	\end{displaymath}
	\begin{definition}[Infinite path]
		\label{def:Paths}
		An \emph{infinite path} of MA $\ma$ with transition probabilities $\probP$ is an infinite sequence $\ppath = \pathIseqTimed$ of states $\state_0, \state_1, \dots \in \States$ and \emph{stamps} $\stamp[0], \stamp[1], \dots \in \Actions \cupdot \RRnn$ such that (1) $\sum_{i=0}^{\infty} \timeOfStamp[i] = \infty$, and for any $i \ge 0$ it holds that
		(2) $\probP(\state_i, \actionOfStamp[i], \state_{i+1}) > 0$,
		(3) $\state_i \in \PS$  implies  $\stamp[i] \in \Actions$, and
		(4) $\state_i \in \MS$ implies $\stamp[i] \in \RRnn$.
	\end{definition}
	An infix $\state_i \pathTransUniv{\stamp[i]} \state_{i+1}$ of a path $\ppath$ represents that we stay at $\state_i$ for $\timeOfStamp[i]$ time units and then perform action $\actionOfStamp[i]$ and move to state $\state_{i+1}$. 
	Condition~(1) excludes Zeno paths, % as their probability is zero (by definition). 
	condition~(2) ensures positive transition probabilities, and
%	For $\state_i \in \PS$ the sojourn time $\timeOfStamp[i]$ is always zero and for $\state_i \in \MS$ the action $\actionOfStamp[i]$ is always the dedicated action $\markovianAct$.
	conditions~(3) and~(4) assert that stamps $\stamp[i]$ match the transition type at $\state_i$.
	
	A \emph{finite path} is a finite prefix $\ppath' = \pathFseqTimed$ of an infinite path.
	The \emph{length} of $\ppath'$ is $\length{\ppath'} = n$, its \emph{last state} is $\last{\ppath'} = \state_n$, and the \emph{time duration} is $\timeOfPath{\ppath'} = \sum_{0 \le i < \length{\ppath'}} \timeOfStamp[i]$.
	We denote the sets of finite and infinite paths of $\ma$ by $\FPaths[\ma]$ and $\IPaths[\ma]$, respectively.
	The superscript $\ma$ is omitted if the model is clear from the context.
	For a finite or infinite path $\ppath = \pathIseqTimed$ the \emph{prefix} of $\ppath$ of length $n$ is denoted by $\pref{\ppath}{n}$.
	The $i$th state visited by $\ppath$ is given by $\ithStateOfPath{\ppath}{i} = \state_i$.
	The \emph{time-abstraction} $\taOf{\ppath}$ of $\ppath$ removes all sojourn times and is a path of the underlying MDP $\umdp$:
$
	\taOf{\ppath} = \state_0 \pathTransUniv{\actionOfStamp[0]} \state_1 \pathTransUniv{\actionOfStamp[1]} \dots
$\ .
	Paths of $\umdp$ are also referred to as the \emph{time-abstract paths of $\ma$}.

%\paragraph{Schedulers.}
%	A scheduler defines for each finite path $\ppath$ a probability distribution over the outgoing actions at $\last{\ppath}$, resolving the non-determinism of $\ma$.
%	We define the \emph{enabled actions} at a state $\state$ by
%	$\Act{\state} = 
%	\{ \act \in \Actions \mid \state \probTrans{\act} \dist  \}$ if $\state \in \PS$ and $\{\markovianAct\}$ otherwise. 	
	
	\begin{definition}[Generic scheduler]
		\label{def:GenericScheduler}
		A \emph{generic scheduler} for $\ma$ is a measurable function $
		\sched \colon \FPaths \times \Actions \to [0,1]
		$
		such that			 $\schedEval{\sched}{\ppath}{\cdot} \in \Dist{\Act{\last{\ppath}}}$  for each $\ppath \in \FPaths$.%\ 
	\end{definition}
A scheduler $\sched$ for $\ma$ resolves the non-determinism of $\ma$:
$\schedEval{\sched}{\ppath}{\act}$ is the probability to take transition $\last{\ppath} \probTrans{\act} \dist$  after observing the run $\ppath$.
%	To ensure that the resolving of non-determinism induces measurable probabilities, we restrict ourselves to generic schedulers that satisfy $\{\ppath \in \FPathsSampleSpace \mid \schedEval{\sched}{\ppath}{\act} \in B\} \in \FPathsEvents$ for every $\act \in \Actions$ and Borel set $B \in \mathcal{B}([0,1])$.
%	We restrict ourselves to generic schedulers that are \emph{measurable} (see~\cite{DBLP:conf/fossacs/NeuhausserSK09} for details on measurability). 
	The set of such schedulers is denoted by $\GMSched[\ma]$ ($\GMSched[]$ if $\ma$ is clear from the context).
	$\sched \in \GMSched$ is \emph{deterministic} if the distribution $\schedEval{\sched}{\ppath}{\cdot}$ is Dirac for any $\ppath$.
	\emph{Time-abstract schedulers} behave independently of the time-stamps of the given path, i.e.,
	 $\schedEval{\sched}{\ppath}{\act} = \schedEval{\sched}{\ppath'}{\act}$ for all actions $\act$ and paths $\ppath, \ppath'$ with $\taOf{\ppath} = \taOf{\ppath'}$.
	We write $\TASched[\ma]$ to denote the set of time-abstract schedulers of $\ma$.
	GM is the most general  scheduler class for MAs.
	For MDPs, the most general scheduler class is $\TASched$.

\subsection{Objectives}
	An objective $\obj[i]$ is a representation of a \emph{quantitative} property like the probability to reach an error state, or the expected energy consumption.
	To express  \emph{Boolean} properties (e.g., the  probability to reach an error state is below $\pointi{i}$),
	$\obj[i]$ is combined with a \emph{threshold} $\rel[i] \pointi{i}$ where $\rel[i]\, \in \{<, \le, >, \ge \}$  is a \emph{threshold relation} and $\pointi{i} \in \RR$ is a  \emph{threshold value}.
	Let $\ma, \sched \models \obj[i] \rel[i] \pointi{i}$ denote that the MA $\ma$ under scheduler $\sched \in \GMSched$ satisfies the property $\obj[i] \rel[i] \pointi{i}$.
%	We consider reachability objectives and expected reward objectives.
	
	\paragraph{Reachability objectives.} 
%A \emph{reachability objective} $\intervalBoundedReachObj[]$ expresses the probability that some state in a given set of goal states $\goalStates \subseteq \States$ is reached in a given time interval $I$.
$\interval \subseteq \RR$ is a \emph{time interval} if it is of the form 
$\interval = [\lowerTimeBound, \upperTimeBound]$ or $\interval = [\lowerTimeBound,\infty)$, where $0 \le \lowerTimeBound < \upperTimeBound$.
%Let $\goalStates \subseteq \States$ be a set of goal states and $\interval$ be a time interval.
The set of paths reaching a set of goal states $\goalStates \subseteq \States$ in time $\interval$ is defined as
%The set of paths reaching a set of goal states $\goalStates \subseteq \States$ within a time interval $\interval$ is defined as
 \begin{align*}
\eventually[\interval] \goalStates = \{& \ppath = \pathIseqTimed \in \IPaths[] \mid \exists n \ge 0 \colon \ithStateOfPath{\ppath}{n} \in \goalStates \text{ and}\\
& \interval \cap [\ttime, \ttime+\timeOfStamp[n]] \neq \emptyset \text{ for }  \ttime = \timeOfPath{\pref{\ppath}{n}}  \}.%\ .
\end{align*}
We write $\eventually \goalStates$ instead of $\eventually[{[0,\infty)}] \goalStates$.
%To quantify the probability of paths, 
A probability measure $\ProbModelScheduler{\ma}{\sched}$ on sets of infinite paths is defined, which generalizes both the standard probability measure on MDPs and on CTMCs. A formal definition is given in \tech{App:ProbMeasure}.

\begin{definition}[Reachability objective]
	\label{def:reachObjectives}
	A \emph{reachability objective} has the form $\intervalBoundedReachObj[]$ for time interval $\interval$ and goal states $\goalStates$. 
	The objective is \emph{timed} if $\interval \neq [0,\infty)$ and \emph{untimed} otherwise.
	For MA $\ma$ and scheduler $\sched \in \GMSched$, let
	$\ma, \sched \models \intervalBoundedReachObj[] \rel[i] \pointi{i}$  iff  $\intervalBoundedReachProbMa[] \rel[i]\pointi{i}$.
%	and its satisfaction \wrt a threshold $\rel[i] \pointi{i}$ is
%	
%	Let $\ma$ be an MA with scheduler $\sched \in \GMSched[]$ and subset of states $\goalStates$.
%	Further, let $\interval$ be a time interval.
%	\begin{displaymath}
%	\ma, \sched \models \intervalBoundedReachObj[] \rel[i] \pointi{i} \iff \intervalBoundedReachProbMa[] \rel[i]\pointi{i}\ .
%	\end{displaymath}
\end{definition}

\paragraph{Expected reward objectives.}
Expected rewards $\expReachRew[]{\ma}{\sched}{\rewFct[\rewFctIndex]}$ define the expected amount of reward collected (\wrt $\rewFct[\rewFctIndex]$) until a goal state in $\goalStates \subseteq \States$ is reached.
%Expected rewards $\expReachRew[]{\ma}{\sched}{\rewFct[\rewFctIndex]}$ define for a set of goal states $\goalStates \subseteq \States$ for which $\probObj{\eventually[{[0,\infty)}] \goalStates} = 1$ \sj{scheduler...} the rewards obtained via $\rewFct[\rewFctIndex]$ along each path $\pi$ which eventually reaches  $G$ scaled with the probability for $\pi$
This is a straightforward generalization of the notion on CTMCs and MDPs. A formal definition is found in \tech{App:ExpectedReward}.
\begin{definition}[Expected reward objective]
	An \emph{expected reward objective} has the form  $\expReachRewObj[]$ where $\rewFctIndex$ is the index of  reward function $\rewFct[j]$ and $\goalStates \subseteq \States$.
		For MA $\ma$ and scheduler $\sched \in \GMSched$, let
		$	\ma, \sched \models \expReachRewObj[] \rel[i] \pointi{i}$  iff  $\expReachRew[]{\ma}{\sched}{\rewFct[\rewFctIndex]} \rel[i] \pointi{i}$.
\end{definition}
Expected \emph{time} objectives $\expTimeObj[]$ are   expected reward objectives that consider  the reward function $\rewTimeFct[]$  with $\rewTimeFct[](\state) = 1$ if $\state \in \MS$ and all other rewards are zero.

%% file: 03-multiobjective.tex
\section{Multi-objective Model Checking}
\label{sec:multiobjective}
	Standard model checking considers objectives individually.
	This approach is not feasible when we are interested in multiple objectives that should be fulfilled by the same scheduler, e.g., a scheduler that maximizes the expected profit might violate certain safety constraints.
	\emph{Multi-objective} model checking aims to analyze multiple objectives at once and reveals possible trade-offs.
	
%  \tq{refer somewhere ``globally'' that all proofs can be found in App?}
	
\begin{definition}[Satisfaction of multiple objectives]
	\label{def:satOfMultiObj}
	Let $\ma$ be an MA and $\sched \in \GMSched$.
	For objectives $\obj[] = \objTuple$ with threshold relations ${\rel} = \relTuple \in \{<, \le, >, \ge\}^\numObjectives$ and threshold values $\point = \pointTuple \in \RR^\numObjectives$ let 
%	,  the relation $\models$ satisfies
	\begin{displaymath}
	\ma, \sched \models \obj \rel \point \iff \ma, \sched \models \obj[i] \rel[i] \pointi{i} \text{ for all } 1 \le i \le \numObjectives.
	\end{displaymath}
	Furthermore, let $\achievabilityQ{\ma}{\obj \rel \point} \iff \exists \sched \in \GMSched$ such that $\ma, \sched \models \obj \rel \point$.
\end{definition}
If $\ma, \sched \models \obj \rel \point$, the point $\point \in \RR^\numObjectives$ is \emph{achievable} in $\ma$ with scheduler $\sched$.
The \emph{set of achievable points} of $\ma$ w.r.t.\ $\obj$ and $\point$ is 
$
\{\point \in \RR^\numObjectives \mid \achievabilityQ{\ma}{\obj \rel \point}  \}.
$
This definition is compatible with the notions on MDPs as given in \cite{ForejtKPatva12,EtessamiKVY08}.

\begin{example}
	\label{ex:acheivablePoints}
	Fig.~\ref{fig:maAchievablePoints:untimed} and Fig.~\ref{fig:maAchievablePoints:timed} depict the set of achievable points of the MA $\ma$ from Fig.~\ref{fig:maAchievablePoints:ma} w.r.t.\ relations ${\rel}=(\ge,\ge)$ and objectives $(\reachObjUniv{\{s_2\}},\reachObjUniv{\{s_4\}})$ and $(\reachObjUniv{\{s_2\}},\boundedReachObjUniv{\{s_4\}}{{[0,2]}})$, respectively. 
	Using the set of achievable points, we can answer Pareto, numerical, and achievability queries as considered in \cite{ForejtKPatva12}, e.g., the Pareto front lies on the border of the set.
\end{example}
\paragraph{Schedulers.}
For single-objective model checking on MAs, it suffices to consider deterministic schedulers
\cite{DBLP:conf/fossacs/NeuhausserSK09}. 
%, total-time positional deterministic schedulers (deterministic schedulers that depend on the total time and the current state) suffice for timed reachability
For untimed reachability and expected rewards even \emph{time-abstract} deterministic schedulers suffice
\cite{DBLP:conf/fossacs/NeuhausserSK09}. 
%This result does not hold for multi-objective model checking. 
Multi-objective model checking on MDPs requires history-dependent, randomized schedulers \cite{EtessamiKVY08}. 
On MAs, schedulers may also employ \emph{timing} information to make optimal choices, even if only \emph{untimed} objectives are considered.
%The following example illustrates that deterministic schedulers might require \emph{timing} information to make optimal choices when multiple \emph{untimed} reachability objectives are to be achieved.
\begin{example}
	\label{ex:timingIsImportant}
	Consider the MA $\ma$ in Fig.~\ref{fig:maAchievablePoints:ma} with untimed objectives $\reachObjUniv{\{s_2\}} \geq 0.5$ and $ \reachObjUniv{\{s_4\}} \geq 0.5$.
	 A simple graph argument yields that both properties are only satisfied if action $\act$ is taken with probability exactly a half. 
	 Thus, on the underlying MDP, no deterministic scheduler satisfies both objectives. 
	 On the MA however, paths can be distinguished by their sojourn time in $s_0$. 
	 As the probability mass to stay in $s_0$ for at most $\ln(2)$ is exactly $0.5$, a timed scheduler $\sched$ with $\schedEval{\sched}{\state_0 \pathTransUniv{\ttime} \state_1}{\act}=1$ if $\ttime \leq \ln(2)$ and $0$ otherwise does satisfy both objectives. 
\end{example}
\begin{figure}[t]
	\centering
	\subfigure[MA $\ma$.]{
		\scalebox{\picscale}{
			\input{pics/ma_infinite_polytope}
		}
		\label{fig:maAchievablePoints:ma}
	}
	\subfigure[Untimed objectives.]{
		\scalebox{0.85}{
			\input{pics/achievablePoints_untimed}
		}
		\label{fig:maAchievablePoints:untimed}
	}
	\subfigure[Timed objectives.]{
		\scalebox{0.85}{
			\input{pics/achievablePoints_timed}
		}
		\label{fig:maAchievablePoints:timed}
	}
	\label{fig:ma_umdp}
	\vspace{-2mm}
\caption{Markov automaton and achievable points.}
\vspace{-3mm}
\end{figure}
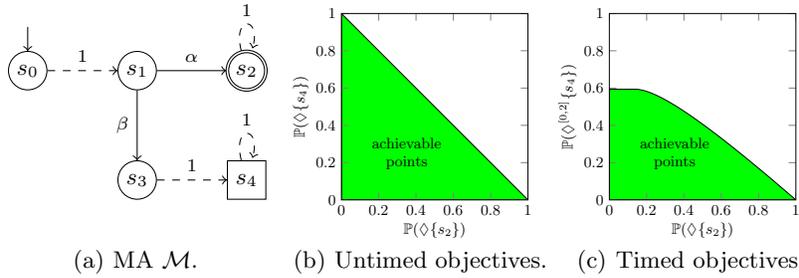
\begin{restatable}{theorem}{theoremDeterministicTADoesNotSuffice}
	\label{thm:deterministicTADoesNotSuffice}
	For some MA $\ma$ with $\achievabilityQ{\ma}{\obj \rel \point}$, no deterministic time-abstract scheduler $\sched$ satisfies  $\ma, \sched \models \obj \rel \point$. 
\end{restatable}
\paragraph{The geometric shape of the achievable points.}
Like for MDPs~\cite{EtessamiKVY08}, the set of achievable points of any combination of aforementioned objectives is convex.
\begin{restatable}{proposition}{propositionAchievablePointsIsConvex}
	\label{prop:achievablePointsIsConvex}
The set $\{\point \in \RR^\numObjectives \mid \achievabilityQ{\ma}{\obj \rel \point}  \}$ is convex.
\end{restatable}
\noindent For MDPs, the set of achievable points is a convex polytope where the vertices can be realized by deterministic schedulers that use memory  bounded by the number of objectives.
As there are finitely many such schedulers, the polytope  is finite~\cite{EtessamiKVY08}, \ie, it can be represented by a finite number of vertices.
This result does not carry over to MAs. 
For example, the achievable points of the MA from Fig.~\ref{fig:maAchievablePoints:ma} together with the objectives $(\reachObjUniv{\{s_2\}},\boundedReachObjUniv{\{s_4\}}{{[0,2]}})$ form the infinite polytope shown in Fig.~\ref{fig:maAchievablePoints:timed}. The insight here is that for any sojourn time $t \leq 2$ in $s_0$, the  timing information is relevant for optimal schedulers: 
The shorter the sojourn time in $s_0$, the higher the probability to reach $s_4$ within the time bound.
% A formal proof is found in App.~\ref{app:proofsMOMC:infinite}.
%
%
\begin{restatable}{theorem}{theoremInfinitePolytope}
	\label{Thm:InfinitePolytope}
	For some MA $\ma$ and objectives $\obj[]$, the polytope $\{\point \in \RR^\numObjectives \mid \achievabilityQ{\ma}{\obj \rel \point}  \}$ is not finite.
\end{restatable}
\noindent As infinite convex polytopes cannot be represented by a finite number of vertices,  any method extending the approach of \cite{ForejtKPatva12} -- which computes these vertices -- can only approximate the set of achievable points. 

\paragraph{Problem statement.}
%\sj{Finish}
For an MA and objectives with threshold relations, construct arbitrarily tight over- and under-approximations of the achievable points.

%% file: pics/ma_infinite_polytope.tex
  \begin{tikzpicture}
    \draw[white, use as bounding box] (-0.3,-2.3) rectangle (3.15,1); 
    \node [state] (s0) at (0,0) {$\state_0$};
    \node [state] (s1) [on grid, right=15mm of s0] {$\state_1$};
    \node [state, accepting] (s2) [on grid, right=15mm of s1] {$\state_2$};
    \node [state] (s3) [on grid, below=15mm of s1] {$\state_3$};
    \node [state, rectangle] (s4) [on grid, right=15mm of s3] {$\state_4$};
    
    \initstateAbove{s0}
    
    \draw (s0) edge[markovian] node[above] {\scriptsize$1$} (s1);
\draw (s1) edge[probabilistic] node[above] {\scriptsize$\act$} (s2);
\draw (s1) edge[probabilistic] node[left] {\scriptsize$\altact$} (s3);
    \draw (s3) edge[markovian] node[above] {\scriptsize$1$} (s4);
    \draw (s2) edge[markovian, loop above] node {\scriptsize$1$} (s2);
    \draw (s4) edge[markovian, loop above] node {\scriptsize$1$} (s4);
\end{tikzpicture}

%% file: pics/achievablePoints_untimed.tex
\pgfplotsset{compat = 1.3} % for label shifting to work
\begin{tikzpicture}[scale=\plotscale]
\draw[white, use as bounding box] (-0.6,-0.8) rectangle (3.9,3.75); 
\begin{axis}[
axis equal image,
    xmin=0,
xmax=1,
ymin=0.0,
ymax = 1,
enlargelimits=false,
width=0.52\textwidth,
xlabel = {$\reachObjUniv{\{s_2\}}$},
ylabel shift = -2pt,
xlabel shift = -2pt,
ylabel = {$\reachObjUniv{\{s_4\}}$},
axis on top,
%axis background/.style={fill=badArea},
]

\addplot[fill=goodArea, very thin] coordinates {(1,0) (0,1) (0,0)} -- cycle;

%\plotPoint{(0.48,0.35)}{$\tilde{\point}$}
\node[align=center] at (axis cs:0.35,0.25) {achievable\\points};

\end{axis}
\end{tikzpicture}

%% file: pics/achievablePoints_timed.tex
\pgfplotsset{compat = 1.3} % for label shifting to work
\begin{tikzpicture}[scale=\plotscale]
\draw[white, use as bounding box] (-0.58,-0.8) rectangle (4,3.75); 
\begin{axis}[
axis equal image,
xmin=0,
xmax=1,
ymin=0.0,
ymax = 1,
enlargelimits=false,
width=0.52\textwidth,
xlabel = {$\reachObjUniv{\{s_2\}}$},
ylabel shift = -2pt,
xlabel shift = -2pt,
ylabel = {$\boundedReachObjUniv{\{s_4\}}{{[0,2]}}$},
axis on top,
%axis background/.style={fill=badArea},
]

\addplot[fill=goodArea, very thin] table [col sep=comma] {pics/achievablePoints_timed_precise.csv} -- cycle;

%\plotPoint{(0.48,0.35)}{$\tilde{\point}$}
\node[align=center] at (axis cs:0.35,0.25) {achievable\\points};

\end{axis}
\end{tikzpicture}

%% file: 04-ma.tex
\section{Analysis of Markov Automata with Multiple Objectives}
\label{sec:ma}
The state-of-the-art in single-objective model checking of MA is to reduce the MA to an MDP, cf.~\cite{HatefiH12,DBLP:journals/corr/GuckHHKT14,DBLP:conf/atva/GuckTHRS14}, for which efficient algorithms exist.
We aim to lift this approach to multi-objective model checking.
Assume MA $\ma$ and objectives $\obj$ with threshold relations $\rel$.
We discuss how the set of achievable points of $\ma$ relates to the set of achievable points of an MDP.
The key challenge is to deal with timing information---even for \emph{un}timed objectives---and to consider schedulers beyond those optimizing single objectives. 
We obtain:
\begin{itemize}
\item 
For untimed reachability and expected reward objectives, the achievable points of $\ma$ \emph{equal} those of its \emph{underlying} MDP, cf. Theorems~\ref{thm:multiReachObj} and \ref{thm:multiExpReachRewObj}.
\item 
For timed reachability objectives, the set of achievable points of a \emph{digitized} MDP $\dma$ provides a \emph{sound approximation} of the achievable points of $\ma$, cf. Theorem~\ref{thm:bounded:multiBounded}.
Corollary~\ref{corr:bounded:approxBoundedProb} gives the precision of the approximation.
\end{itemize}

\input{04-1-unbounded.tex}
\input{04-2-rewards.tex}

\input{04-3-bounded.tex}

%% file: 04-1-unbounded.tex
\subsection{Untimed Reachability Objectives}
\label{sec:ma:unbounded}
Although timing information is essential for \emph{deterministic} schedulers, cf.\ Theorem~\ref{thm:deterministicTADoesNotSuffice}, timing information does not strengthen randomized schedulers:
\begin{restatable}{theorem}{theoremMultiReachObj}
	\label{thm:multiReachObj}
	For MA $\ma$ and untimed reachability objectives $\obj$ it holds that
$
%	\label{eq:ma:thm:multiReachObj:sameAchievPoints}
	\achievabilityQ{\ma}{\obj \rel \point} \iff \achievabilityQ{\umdp}{\obj \rel \point}. % \tagclaim
$
%	and  for any $\sched \in \TASched$  we have
%	\begin{align}
%		\label{eq:ma:thm:multiReachObj:achievedWithSameSched}
%		 \ma, \sched \models \obj \rel \point
%		\iff
%		\umdp, \sched \models \obj \rel \point 
%		 \ . \tagclaim
%	\end{align}
\end{restatable}
%The set of time-abstract schedulers for $\ma$ and $\umdp$ coincide, allowing us to omit the superscript, i.e., $\TASched[\ma] = \TASched[\umdp] = \TASched[]$.
\noindent
%This allows us to answer a multi-objective query for $\ma$ by conducting the corresponding analysis on $\umdp$.
%\ref{eq:ma:thm:multiReachObj:achievedWithSameSched} enables scheduler synthesis for MAs as we infer that a scheduler $\sched$ for $\umdp$ satisfying  the given objectives also satisfies them when applied to $\ma$.
%The correctness of the theorem is not obvious:
%
The main idea for proving Theorem~\ref{thm:multiReachObj} is to construct for scheduler $\sched \in \GMSched[\ma]$ a time-abstract scheduler $\taOfSched \in \TASched[\umdp]$ such that they both induce the same untimed reachability probabilities.
%The construction satisfies $\sched = \taOfSched$ for any time-abstract scheduler $\sched \in \TASched$.
%It follows that $\point \in \RR^\numObjectives$ is achievable in $\ma$ with scheduler $\sched $ iff $\point$ is achievable in $\umdp$ with scheduler $\taOfSched$.
To this end, we discuss the connection between probabilities of paths of MA $\ma$ and paths of MDP $\umdp$.
\begin{definition}[Induced paths of a time-abstract path]
\label{def:InducedTimeStampedPaths}
The set of \emph{induced paths on MA} $\ma$ of a path $\pathTa$ of $\umdp$ is given by

\begin{displaymath}
\inducedPathsTa{\pathTa} =
\ta^{-1}(\pathTa) = 
\{\ppath \in \FPaths[\ma] \cup \IPaths[\ma] \mid  \taOf{\ppath} = \pathTa \}.%\ .
\end{displaymath}
\end{definition}
%\begin{example}
%	\label{ex:unbounded:inducedPaths}
%	Let $\ma$ and $\umdp$ be as in Figure~\ref{fig:ma_taOfSched}. 
%	For $\pathTa = s_0 \pathTransUniv{\markovianAct} \state_1 \pathTransUniv{\act} \state_3 \in \FPaths[\umdp]$ we have
%	\begin{displaymath}
%		\inducedPathsTa{\pathTa} = \{  s_0 \pathTransUniv{\ttime} \state_1 \pathTransUniv{\act} \state_3 \in \FPaths[\ma] \mid \ttime \ge 0 \}.
%	\end{displaymath}
%\end{example}
%Let us fix some $\pathTa \in \FPaths[\umdp]$ and a scheduler $\sched \in \GMSched[\ma]$.
The set $\inducedPathsTa{\pathTa}$ contains all paths of $\ma$ where replacing sojourn times by $\markovianAct$ yields $\pathTa$.

For $\sched \in \GMSched$, the probability distribution $\schedEval{\sched}{\ppath}{\cdot} \in \Dist{\Actions}$ might depend on the sojourn times of the path $\ppath$. 
The time-abstract scheduler $\taOfSched$ weights the distribution $\schedEval{\sched}{\ppath}{\cdot}$ with the probability masses of the paths $\ppath \in \inducedPathsTa{\pathTa}$.
%within the probability distribution $\schedEval{\taOfSched}{\pathTa}{\cdot}$, where $\taOfSched$ is the \emph{time-abstraction} of $\sched$.
\begin{definition}[Time-abstraction of a scheduler]
\label{def:timeAbstractionOfScheduler}
The time-abstraction of $\sched \in \GMSched[\ma]$ is defined as $\taOfSched \in \TASched[\umdp]$ such that for any $\pathTa \in \FPaths[\umdp]$
\begin{displaymath}
\schedEval{\taOfSched}{\pathTa}{\act} = \int_{\ppath \in \inducedPathsTa{\pathTa}} \schedEval{\sched}{\ppath}{\act} \diff\ProbModelScheduler{\ma}{\sched} ( \ppath \mid \inducedPathsTa{\pathTa}).
\end{displaymath}
\end{definition}
The term $\ProbModelScheduler{\ma}{\sched} ( \ppath \mid \inducedPathsTa{\pathTa})$ represents the probability for a path in $\inducedPathsTa{\pathTa}$ to have sojourn times as given by $\ppath$. 
The value $\schedEval{\taOfSched}{\pathTa}{\act}$ coincides with the probability that  $\sched$ picks action $\act$, given that the time-abstract path $\pathTa$ was observed.
\begin{example}
	\label{ex:taOfScheduler} 
	Consider the MA $\ma$ in Fig.~\ref{fig:models:ma} and the scheduler $\sched$ choosing $\act$ at state $\state_3$ iff the sojourn time at $s_0$ is at most one.
%	\begin{align*}
%	\schedEval{\taOfSched}{\pathTa}{\act} 
%	& = \int_{\ppath \in \inducedPathsTa{\pathTa}} \schedEval{\sched}{\ppath}{\act} \diff\ProbModelScheduler{\ma}{\sched} ( \ppath \mid \inducedPathsTa{\pathTa})\\
%	& = \int_{\ppath \in \inducedPathsTa{\pathTa}} \schedEval{\sched}{\ppath}{\act} \diff \ProbModelScheduler{\ma}{\sched}(\ppath ) \\
%	&= \int_{0}^{\infty} \schedEval{\sched}{s_0\pathTransUniv{\ttime} s_1}{\act} \cdot \rateAtState{s_0} \cdot e^{-\rateAtState{s_0} \ttime} \diff \ttime \\
%	&= \int_{0}^{1}  \rateAtState{s_0} \cdot e^{-\rateAtState{s_0} \ttime} \diff \ttime 
%	= 1-e^{-1}\ .
%	\end{align*}
	Then $\schedEval{\taOfSched}{\state_0 \pathTransUniv{\markovianAct} \state_3}{\act} = 1-e^{-\rateAtState{\state_0}}$, the probability that $\state_0$ is left within one time unit.
	For $\pathDi = \state_0 \pathTransUniv{\markovianAct} \state_3 \pathTransUniv{\act} \state_6$ we have
	\begin{displaymath}
	\ProbModelScheduler{\ma}{\sched}(\eventually \{\state_6\})
	= \ProbModelScheduler{\ma}{\sched}(\inducedPathsTa{\pathDi})
	= 1-e^{-\rateAtState{\state_0}}
	= \ProbModelScheduler{\umdp}{\taOfSched}(\pathDi)
	= \ProbModelScheduler{\umdp}{\taOfSched}(\eventually \{\state_6\}).
	\end{displaymath}
\end{example}
In the example, the considered scheduler and its time-abstraction induce the same untimed reachability probabilities.
We generalize this observation.
\begin{restatable}{lemma}{lemmaTaSchedSamePathProb}
	\label{lem:taSchedSamePathProb}
	For any   $\pathTa \in \FPaths[\umdp]$ we have
$
	\ProbModelScheduler{\ma}{\sched}(\inducedPathsTa{\pathTa}) =
	\ProbModelScheduler{\umdp}{\taOfSched}(\pathTa).%\ .
$
\end{restatable}
\noindent
The result is lifted to untimed reachability probabilities.
\begin{restatable}{proposition}{propositionTaSchedSameReachProb}
	\label{pro:taSchedSameReachProb}
	For any $\goalStates \subseteq \States$ it holds that
$
	\reachProbMa = \reachProbUMdp[].
$
\end{restatable}
\noindent
As the definition of $\taOfSched$ is independent of the considered set of goal states $\goalStates \subseteq \States$, Proposition~\ref{pro:taSchedSameReachProb} can be lifted to multiple untimed reachability objectives.
%This provides the basis for the proof of Theorem~\ref{thm:multiReachObj}.
\paragraph{Proof of Theorem~\ref{thm:multiReachObj} (sketch).}
By applying Proposition~\ref{pro:taSchedSameReachProb}, we can show that
$
		 \ma, \sched \models \obj \rel \point
		\iff
		\umdp, \taOfSched \models \obj \rel \point 
$
for any scheduler $\sched \in \GMSched[\ma]$ and untimed reachability objectives $\obj = (\reachObj[1], \dots, \reachObj[\numObjectives])$ with thresholds $\rel \point$.
Theorem~\ref{thm:multiReachObj} is a direct consequence of this.

%% file: 04-2-rewards.tex
\subsection{Expected Reward Objectives}
%\jpk{Somewhere refer to appendix with notation for rewards}
\label{sec:ma:rewards}
The results for expected reward objectives are similar to untimed reachability objectives: An analysis of the underlying MDP suffices.
We show the following extension of Theorem~\ref{thm:multiReachObj} to expected reward objectives. 
\begin{restatable}{theorem}{theoremMultiExpReachRewObj}
	\label{thm:multiExpReachRewObj}
	For MA $\ma$ and untimed reachability and expected reward objectives $\obj$:
	%% with threshold relations ${\rel}$, and point  $\point$:
$
	\achievabilityQ{\ma}{\obj \rel \point} \iff \achievabilityQ{\umdp}{\obj \rel \point}.
$
%	and for any $\sched \in \TASched$   we have
%	\begin{displaymath}
%	\ma, \sched \models \obj \rel \point
%	\iff 
%	\umdp, \sched \models \obj \rel \point \ .
%	\end{displaymath}
\end{restatable}
\noindent
To prove this, we show that a scheduler $\sched \in \GMSched[\ma]$ and its time-abstraction $\taOfSched \in \TASched[]$ induce the same expected rewards on $\ma$ and $\umdp$, respectively. Theorem~\ref{thm:multiExpReachRewObj} follows then analogously to Theorem~\ref{thm:multiReachObj}.
\begin{restatable}{proposition}{propositionTaSchedSameExpReachRew}
	\label{pro:taSchedSameExpReachRew}
	Let $\rewFct[]$ be some reward function of $\ma$ and let $\rewUMdpFct[]$ be its counterpart for $\umdp$.
	For $\goalStates \subseteq \States$ we have
$
	\expReachRewMa[] = \expReachRewUMdp[].%\ .
$
\end{restatable}
\noindent
Notice that $\rewUMdpFct$ encodes the \emph{expected} reward of $\ma$ obtained in a state $s$ by assuming the sojourn time to be the expected sojourn time $\nicefrac{1}{E(s)}$.
Although the claim is  similar to Proposition~\ref{pro:taSchedSameReachProb}, its proof cannot be adapted straightforwardly.
In particular, the analogon to Lemma~\ref{lem:taSchedSamePathProb} does not hold:
%Let $\rewFct[]$ be a reward function for $\ma$ and let $\rewUMdpFct[]$ be the corresponding counterpart for $\umdp$.
The expected reward  collected along a time-abstract path  $\pathTa \in \FPaths[\umdp]$ does in general not coincide for $\ma$ and $\umdp$.
%\begin{displaymath}
%\int_{\ppath \in \inducedPathsTa{\pathTa}} \rewPath[\ma]{\rewFct}{\ppath} \diff\ProbModelScheduler{\ma}{\sched}(\ppath) \neq  \rewPath[\umdp]{\rewUMdpFct}{\pathTa} \cdot \ProbModelScheduler{\umdp}{\taOfSched}(\pathTa).%\ .
%\end{displaymath}
\begin{example}
\label{ex:expRewardsOfTaFinPathAreNeq}
We consider standard notations for rewards as detailed in \tech{App:ExpectedReward}.
Let $\ma$ be the MA with underlying MDP $\umdp$ as shown in Fig.~\ref{fig:models}.
Let $\rewFct(\state_0) =1$ and zero otherwise. Reconsider the scheduler $\sched$
% and its time-abstraction $\taOfSched$ 
from Example~\ref{ex:taOfScheduler}.
Let $\pathTa_\act = \state_0 \pathTransUniv{\markovianAct} \state_3 \pathTransUniv{ \act} \state_6$. 
The probability $\ProbModelScheduler{\ma}{\sched}(\{ \state_0 \pathTransUniv{\ttime} \state_3 \pathTransUniv{\act} \state_6 \in \inducedPathsTa{\pathTa_\act}\mid \ttime > 1  \})$ is zero since $\sched$ chooses $\altact$ on such paths.
For the remaining paths in $\inducedPathsTa{\pathTa_\act}$, action $\act$ is chosen with  probability one.
The expected reward in $\ma$ along $\pathTa_\act$ is:
\begin{align*}
\int_{\ppath \in \inducedPathsTa{\pathTa_\act}} \rewPath[\ma]{\rewFct}{\ppath} \diff\ProbModelScheduler{\ma}{\sched}(\ppath) 
& =\int_0^1 \rewState[]{\state_0} \cdot \ttime \cdot   \rateAtState{\state_0}  \cdot e^{-\rateAtState{\state_0} \ttime}  \diff \ttime 
%\\
%& =  \frac{\rewState{\state_0}}{\rateAtState{\state_0}} \cdot \Big( 1 -  (\rateAtState{\state_0}   + 1) \cdot e^{-\rateAtState{\state_0}}\Big)
= 1-2e^{-1}.
\end{align*}
The expected reward in $\umdp$  along $\pathTa_\act$ differs as
\begin{displaymath}
 \rewPath[\umdp]{\rewUMdpFct}{\pathTa_\act} \cdot \ProbModelScheduler{\umdp}{\taOfSched}(\pathTa_\act) = \rewUMdp{\state_0}{\markovianAct} \cdot \schedEval{\taOfSched}{\state_0 \pathTransUniv{\markovianAct} \state_3}{\act} = 1-e^{-1}.
\end{displaymath}
The intuition is as follows:
If path $\state_0 \pathTransUniv{\ttime} \state_3 \pathTransUniv{\act} \state_6$ of $\ma$ under $\sched$ occurs, we have  $\ttime \le 1$ since $\sched$ chose $\act$.
Hence, the reward collected from paths in $\inducedPathsTa{\pathTa_\act}$ is at most $1 \cdot \rewFct(\state_0) = 1$.
There is thus a dependency between the choice of the scheduler at $\state_3$ and the collected reward at $\state_0$. This dependency is absent in $\umdp$ as the reward at a state is independent of the subsequent performed actions.

Let $\pathTa_\altact = \state_0 \pathTransUniv{\markovianAct} \state_3 \pathTransUniv{\altact} \state_4$. The expected reward  along $\pathTa_\altact$ is $2e^{-1}$ for $\ma$  and $e^{-1}$ for $\umdp$.
As the rewards for $\pathTa_\act$ and $\pathTa_\altact$ sum up to one in both $\ma$ and $\umdp$, the expected reward along all paths of length two coincides for $\ma$ and $\umdp$.

\end{example}
%The set of all paths of length $n=2$ induce the same expected reward for the example MA and its underlying MDP.
This observation can be generalized to arbitrary MA and paths of arbitrary length.
\paragraph{Proof of Proposition~\ref{pro:taSchedSameExpReachRew} (sketch).}
For every $n\ge0$, the expected reward collected along paths of length at most $n$ coincides for $\ma$ under $\sched$ and $\umdp$ under $\taOfSched$.
The proposition follows by letting $n$ approach infinity.

Thus, queries on MA with mixtures of untimed reachability and expected reward objectives can be analyzed on the underlying MDP $\umdp$.

%% file: 04-3-bounded.tex
\subsection{Timed Reachability Objectives}
\label{sec:ma:bounded}
Timed reachability objectives cannot be analyzed on $\umdp$ as it abstracts away from sojourn times.
We lift the digitization approach for single-objective timed reachability \cite{HatefiH12,DBLP:journals/corr/GuckHHKT14} to multiple objectives.
Instead of abstracting timing information, it is \emph{digitized}.
Let  $\dma$ denote the digitization of $\ma$ for arbitrary digitization constant $\digConstant \in \RRgz$, see Def.~\ref{def:dma}.
A time interval $\interval \subseteq \RRnn$ of the form $[\lowerTimeBound,\infty)$ or $[\lowerTimeBound, \upperTimeBound]$ with $ \lowerStepBound \coloneqq \nicefrac{\lowerTimeBound}{\digConstant} \in \NN$ and $\upperStepBound \coloneqq \nicefrac{\upperTimeBound}{\digConstant}  \in \NN$ is called \emph{well-formed}. 
For the remainder, we only consider well-formed intervals, ensured by an appropriate digitization constant. 
An interval for time-bounds $\interval$ is transformed to digitization step bounds $\diOfInterval \subseteq \NN$. 
Let $a=\inf \interval $, we set $
\diOfInterval =
 \{ \nicefrac{\ttime}{\digConstant} \in \NN \mid  \ttime \in \interval \} \setminus \{ 0 \mid a > 0 \}
$.

%\tq{assert that point intervals are excluded at this point}

We first relate paths in $\ma$ to paths in its digitization.
\begin{definition}[Digitization of a path]
	\label{def:digitizationOfPath}
The \emph{digitization} $\diOfPath$ of path $\ppath = \pathIseqTimed$ in $\ma$ is the path in $\dma$ given by
		\begin{displaymath}
	\diOfPath = \big( \state_0 \pathTransUniv{\actionOfStamp[0]}\!\! \big)^{\waitingSteps[0]}  \state_0 \pathTransUniv{\actionOfStamp[0]} \big(  \state_1 \pathTransUniv{\actionOfStamp[1]} \!\!  \big)^{\waitingSteps[1]} \state_1 \pathTransUniv{\actionOfStamp[1]} \dots\ 
	\end{displaymath}
where $\waitingSteps[i] = \max \{ \waitingSteps[] \in \NN \mid \waitingSteps[]  \digConstant \le \timeOfStamp[i]  \}$ for each $i\ge 0$.
%	Let $\ppath = \pathIseqTimed$ be a path of $\ma$ and let  $\waitingSteps[i] = \max \{ \waitingSteps[] \in \NN \mid \waitingSteps[]  \digConstant \le \timeOfStamp[i]  \}$ for each $i\ge 0$.
%	The \emph{digitization} $\diOfPath$ of $\ppath$ is a path of $\dma$ given by
%	\begin{displaymath}
%	\diOfPath = \big( \state_0 \pathTransUniv{\actionOfStamp[0]}\!\! \big)^{\waitingSteps[0]}  \state_0 \pathTransUniv{\actionOfStamp[0]} \big(  \state_1 \pathTransUniv{\actionOfStamp[1]} \!\!  \big)^{\waitingSteps[1]} \state_1 \pathTransUniv{\actionOfStamp[1]} \dots\ .
%	\end{displaymath}
\end{definition}
\begin{example}
	\label{ex:digOfPath}
	For the path $\ppath = \state_0 \pathTransUniv{1.1} \state_3 \pathTransUniv{\altact} \state_4 \pathTransUniv{\eta} \state_5 \pathTransUniv{0.3} \state_4 $ of the MA $\ma$ in Fig.~\ref{fig:models:ma} and $\digConstant = 0.4$, we get
	$
	\diOfPath = \state_0 \pathTransUniv{\markovianAct}  \state_0 \pathTransUniv{\markovianAct}  \state_0 \pathTransUniv{\markovianAct} \state_3 \pathTransUniv{\altact} \state_4 \pathTransUniv{\eta} \state_5 \pathTransUniv{\markovianAct} \state_4 $.
\end{example}
The $\waitingSteps[i]$ in the definition above represent a digitization of the sojourn times $\timeOfStamp[i]$ such that $\waitingSteps[i]  \digConstant \le \timeOfStamp[i] < (\waitingSteps[i]{+}1)  \digConstant$. 
These digitized times are incorporated into the digitization of a path by taking the self-loop at state $\state_i \in \MS$ $\waitingSteps[i]$ times.
We also refer to the paths of $\dma$ as \emph{digital paths (of $\ma$)}. 
The number $\numOfDS{\pathDi}$ of \emph{digitization steps} of a digital path $\pathDi$ is the number of transitions emerging from Markovian states, i.e.,
$%\begin{displaymath}
\numOfDS{\pathDi} = \lvert \{ i < \length{\pathDi} \mid  \ithStateOfPath{\pathDi}{i} \in \MS\} \rvert
$. %\end{displaymath}
One digitization step represents the elapse of at most $\digConstant$ time units---either by staying at some $\state \in \MS$ for $\digConstant$ time or by leaving $\state$ within $\digConstant$ time.
The number $\numOfDS{\diOfPath}$ multiplied with $\digConstant$ yields an estimate for the duration $\timeOfPath{\ppath}$.
A digital path $\pathDi$ can be interpreted as representation of the set of paths of $\ma$ whose digitization is $\pathDi$.
\begin{definition}[Induced paths of a digital path]
	\label{def:ma:pathsInducedByDigPath}
	The set of \emph{induced paths} of a (finite or infinite) digital path $\pathDi$ of $\dma$ is
	\begin{displaymath}
	\inducedPathsDi{\pathDi} = \di^{-1}(\pathDi) = \{ \ppath \in \FPaths[\ma] \cup \IPaths[\ma] \mid \diOfPath = \pathDi \}.%\ .
	\end{displaymath}
\end{definition}
%\unpolished{
%\begin{example}
%	\label{ex:ma:bounded:inducedPaths}
%	Let $\ma$ be the MA with digitization $\dma$ as shown in Figure~\ref{fig:ma:bounded}.
%	We assume $\digConstant = 0.1$.
%	Consider the digital path $\pathDi_1 = \state_0 \pathTransUniv{\markovianAct} \state_0 \pathTransUniv{\markovianAct} \state_0 \pathTransUniv{\markovianAct} \state_1 \pathTransUniv{\act} \state_3$.
%	The set $\inducedPathsDi{\pathDi_1}$ contains, e.g., the path $\ppath_1  = \state_0 \pathTransUniv{0.25} \state_1 \pathTransUniv{\act} \state_3$ from Example~\ref{ex:ma:bounded:digOfPaths} since $\diOf{\ppath_1} = \pathDi_1$.
%	More generally, the set of paths induced by $\pathDi_1$ is given by
%	\begin{displaymath}
%	\inducedPathsDi{\pathDi_1} 
%	= \{ \state_0 \pathTransUniv{\ttime} \state_1 \pathTransUniv{\act} \state_3 \in \FPaths[\ma] \mid 0.2 \le \ttime < 0.3  \}.%\ .
%	\end{displaymath}
%\end{example}
%}
\noindent For sets of digital paths $\Pathset$ %% \subseteq \IPaths[\dma] \cup \FPaths[\dma]$ 
we define the \emph{induced paths} $\inducedPathsDi{\Pathset} = \bigcup_{\pathDi \in \Pathset} \inducedPathsDi{\pathDi}$.
To relate timed reachability probabilities for $\ma$ under scheduler $\sched \in \GMSched[\ma]$ with \ds-bounded reachability probabilities for $\dma$, relating $\sched$ to a scheduler for $\dma$ is necessary.
%Given a scheduler $\sched$ for $\ma$, we construct a scheduler $\diOfSched$ for $\dma$ mimicking  $\sched$. 

\begin{definition}[Digitization of a scheduler]
	The \emph{ digitization} of $\sched\in \GMSched[\ma]$ is given by  $\diOfSched \in \TASched[\dma]$ such that
	for any $\pathDi \in \FPaths[\dma]$ with $\last{\pathDi} \in \PS$
	\begin{displaymath}
	\schedEval{\diOfSched}{\pathDi}{\act} = \int_{\ppath \in \inducedPathsDi{\pathDi}} \schedEval{\sched}{\ppath}{\act} \diff \ProbModelScheduler{\ma}{\sched}(\ppath \mid \inducedPathsDi{\pathDi}).
	\end{displaymath}
\end{definition}
The digitization $\diOfSched$ is similar to the time-abstraction $\taOfSched$ as both schedulers get a path with restricted timing information as input and mimic the choice of $\sched$.
However, while $\taOfSched$ receives no information regarding sojourn times, $\diOfSched$ receives the digital estimate.
Intuitively, $\schedEval{\diOfSched}{\pathDi}{\act}$ considers  $\schedEval{\sched}{\ppath}{\act}$ for each $\ppath\in \inducedPathsDi{\pathDi}$, weighted with the probability that the sojourn times of a path in $\inducedPathsDi{\pathDi}$ are as given by $\ppath$.
The restriction $\last{\pathDi} \in \PS$ asserts that $\pathDi$ does not end with a self-loop on a Markovian state, implying $\inducedPathsDi{\pathDi} \neq \emptyset$. 
%Note that $\diOfSched = \sched$ if $\sched \in \TASched[\dma]$.

\begin{example}
	\label{ex:diOfSched}
	Let MA $\ma$ in Fig.~\ref{fig:models:ma} and $\digConstant=0.4$. Again, $\sched\in\GMSched[\ma]$ chooses $\act$ at state $\state_3$ iff the sojourn time at $\state_0$ is at most one.
	  Consider the digital paths $\pathDi_m = ( \state_0 \pathTransUniv{\markovianAct} \!  )^{m} \state_0 \pathTransUniv{\markovianAct} \state_3$.
	For $\ppath \in \inducedPathsDi{\pathDi_1} = \{\state_0 \pathTransUniv{\ttime} \state_3 \mid 0.4 \le \ttime < 0.8 \}$ we have $\schedEval{\sched}{\ppath}{\act} = 1$. 
	It follows $\schedEval{\diOfSched}{\ppath_1}{\act} = 1$.
	For $\ppath \in \inducedPathsDi{\pathDi_2} = \{\state_0 \pathTransUniv{\ttime} \state_3 \mid 0.8 \le \ttime < 1.2 \}$ it is unclear whether $\sched$ chooses $\act$ or $\altact$.
	Hence, $\diOfSched$ randomly guesses:
%	We have 
		\begin{align*}
		\schedEval{\diOfSched}{\pathDi_2}{\act}
		 = 
		\int_{\ppath \in \inducedPathsDi{\pathDi_2}} \schedEval{\sched}{\ppath}{\act} \diff \ProbModelScheduler{\ma}{\sched}(\ppath \mid \inducedPathsDi{\pathDi_2}) 
%		= \frac{ \int_{0.8}^{1.0} \rateAtState{\state_0} e^{-\rateAtState{\state_0} \ttime} \diff\ttime }{\ProbModelScheduler{\ma}{\sched}(\inducedPathsDi{\pathDi_2}) } 
		= \frac{ \int_{0.8}^{1.0} \rateAtState{\state_0} e^{-\rateAtState{\state_0} \ttime} \diff\ttime }{\int_{0.8}^{1.2} \rateAtState{\state_0} e^{-\rateAtState{\state_0} \ttime} \diff\ttime} 
%		 = \frac{e^{-0.8} - e^{-1.0}}{e^{-0.8} - e^{-1.2}}
 \approx 0.55\ .
		\end{align*}
\end{example}
On $\dma$ we consider \ds-bounded reachability instead of timed reachability.
\begin{definition}[\ds-bounded reachability]
	\label{def:ma:bounded:digStepBoundedPaths}
	The set of infinite digital paths that reach $\goalStates \subseteq \States$ within the interval $\altInterval \subseteq \NN$ of consecutive natural numbers is
	\begin{align*}
	\eventuallyds[\altInterval] \goalStates = \{\pathDi \in \IPaths[\dma] \mid  \ \exists n \ge 0 \colon \ithStateOfPath{\pathDi}{n} \in \goalStates \text{ and } \numOfDS{\pref{\pathDi}{n}} \in \altInterval \}.%\ .\\
	\end{align*}
\end{definition}
The timed reachability probabilities for $\ma$ are estimated by \ds-bounded reachability probabilities for $\dma$.
The induced \ds-bounded reachability probability for $\ma$ (under $\sched$) coincides with \ds-bounded reachability probability on $\dma$ (under $\diOfSched$). 
%This is formalized by the following proposition.
\begin{restatable}{proposition}{propositionDsBoundedProbEqual}
	\label{pro:bounded:dsBoundedProbEqual}
	Let $\ma$ be an MA with $\goalStates \subseteq \States$, $\sched \in \GMSched$, and digitization $\dma$.
	Further, let $\altInterval \subseteq \NN$ be a  set of consecutive natural numbers.
	It holds that
	\begin{displaymath} 
	\ProbModelScheduler{\ma}{\sched}(\inducedPathsDi{\eventuallyds[\altInterval] \goalStates})
	=	\ProbModelScheduler{\dma}{\diOfSched}(\eventuallyds[\altInterval] \goalStates).
	\end{displaymath}
\end{restatable}
\noindent Thus, induced \ds-bounded reachability on MAs can be computed on their digitization.
Next, we relate \ds-bounded and timed reachability on MAs, \ie,
 we quantify the maximum difference between time-bounded and $\ds$-bounded reachability probabilities.
\begin{example}
	\label{ex:bounded:dsBoundedPaths}
	Let $\ma$ be the MA given in Fig.~\ref{fig:bounded:dsBoundedPaths:ma}.
	We consider the well-formed time interval $\interval = [0, 5 \digConstant]$, yielding digitization step bounds $\diOfInterval = \{0, \dots, 5\}$.
	The digitization constant $\digConstant \in \RRgz$ remains unspecified in this example.
	Fig.~\ref{fig:bounded:dsBoundedPaths:paths} illustrates paths $\ppath_1$, $\ppath_2$, and $\ppath_3$ of $\ma$.
	We depict sojourn times by arrow length.
	A black dot indicates that the path stays at the current state for a multiple of $\digConstant$ time units.
	All depicted paths reach $\goalStates = \{s_3\}$ within $5 \digConstant$ time units.
	However, the digitizations of $\ppath_1$, $\ppath_2$, and $\ppath_3$ reach $\goalStates$ within $5$, $4$, and $6$ digitization steps, respectively. This yields
%	 only the digitizations of $\ppath_1$ and $\ppath_2$ also reach $\goalStates$ within $5$ digitization steps, \ie, 
	%$\ppath_3$ induces more than 5 digitization steps before reaching $\goalStates$
	\begin{displaymath}
		\ppath_1, \ppath_2 \in \eventually[\interval] \goalStates \cap \inducedPathsDi{\eventuallyds[\diOfInterval]\goalStates} \quad \text{ and } \quad \ppath_3 \in \eventually[\interval] \goalStates \setminus \inducedPathsDi{\eventuallyds[\diOfInterval]\goalStates}.
	\end{displaymath} 
\end{example}
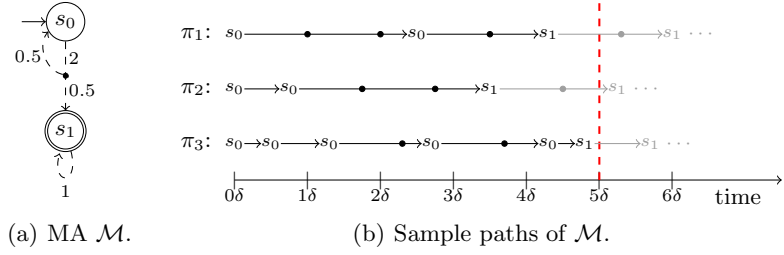
\begin{figure}[t]
	\centering
	\subfigure[MA $\ma$.]{
		\scalebox{\picscale}{
			\input{pics/ma_bounded_upperBound_ma}
		}
		\label{fig:bounded:dsBoundedPaths:ma}
	}
	\subfigure[Sample paths of $\ma$.]{
		\scalebox{\picscale}{
			\input{pics/ma_bounded_upperBound_paths}
		}
		\label{fig:bounded:dsBoundedPaths:paths}
	}
	\caption{MA $\ma$ and illustration of paths of $\ma$ (cf. Example~\ref{ex:bounded:dsBoundedPaths}).}
	\label{fig:bounded:dsBoundedPaths}
\end{figure}
Let $\rateMax = \max \{\rateAtState{\state} \mid \state \in \MS \}$ be the maximum exit rate of $\ma$.
For $\lowerTimeBound \neq 0$ define
\begin{align*}
\lowerErrorBound([\lowerTimeBound, \upperTimeBound]) &= \lowerErrorBound([\lowerTimeBound, \infty)) = 1 - (1+ \rateMax \digConstant)^{\lowerStepBound} \cdot e^{- \rateMax \lowerTimeBound}\ , \quad    \lowerErrorBound([0,\upperTimeBound)) = \lowerErrorBound([0,\infty]) = 0 ,  \\
\upperErrorBound([\lowerTimeBound,\upperTimeBound]) &=\underbrace{ 1 - (1+ \rateMax \digConstant)^{\upperStepBound} \cdot e^{- \rateMax \upperTimeBound}}_{=\upperErrorBound([0,\upperTimeBound])}  + \underbrace{1 - e^{-\rateMax \digConstant}}_{=\upperErrorBound([\lowerTimeBound, \infty))}\ , \text{ and }\   \upperErrorBound([0,\infty)) = 0 .
\end{align*}
$\lowerErrorBound(\interval)$ and $\upperErrorBound(\interval)$ approach $0$ for small digitization constants $\digConstant \in \RRgz$.

\begin{restatable}{proposition}{propositionApproxBoundedProb}
	\label{prop:bounded:approxBoundedProb}
	For MA $\ma$, scheduler $\sched \in \GMSched$, goal states $\goalStates \subseteq \States$, digitization constant $\digConstant \in \RRgz$ and time interval $\interval$
%	\begin{displaymath}
%		\intervalBoundedReachProbMa[]  \in
%	 [p - \lowerErrorBound(\interval),\,  p+ \upperErrorBound(\interval)] 
%	 \text{, where } p = 	\ProbModelScheduler{\dma}{\diOfSched}(\eventuallyds[\diOfInterval] \goalStates).
%	\end{displaymath}
	\begin{displaymath}
		\intervalBoundedReachProbMa[]  \in
		\ProbModelScheduler{\ma}{\sched}(\inducedPathsDi{\eventuallyds[\interval] \goalStates})
		+
	 \Big[{-}\lowerErrorBound(\interval),\,  \upperErrorBound(\interval)\Big] 
	\end{displaymath}
%	\begin{displaymath}
%		\ProbModelScheduler{\dma}{\diOfSched}(\eventuallyds[\diOfInterval] \goalStates)  - \lowerErrorBound(\interval)
%		\le \intervalBoundedReachProbMa[]  \le
%		 \ProbModelScheduler{\dma}{\diOfSched}(\eventuallyds[\diOfInterval] \goalStates)  + \upperErrorBound(\interval).
%	\end{displaymath}
\end{restatable}
\noindent 
\begin{figure}[b]
	\centering
	\scalebox{\picscale}{
		\input{pics/bounded_blobs}
	}
	\caption{Illustration of the sets $\eventually[\interval] \goalStates$ and $\inducedPathsDi{\eventuallyds[\diOfInterval] \goalStates}$. }
	\label{fig:bounded:blobs}
%	\vspace{-3mm}
\end{figure}
\paragraph{Proof (sketch).}
The sets $\eventually[\interval] \goalStates$ and $\inducedPathsDi{\eventuallyds[\diOfInterval] \goalStates}$ are illustrated in Fig.~\ref{fig:bounded:blobs}.
We have
\begin{align*}
\label{eq:bounded:relationBetweenTimedAndDSBoundedProb}
\begin{split}
\ProbModelScheduler{}{\sched}(\eventually[\interval] \goalStates) =
\ProbModelScheduler{}{\sched}(\inducedPathsDi{\eventuallyds[{\diOfInterval}] \goalStates}) 
+	\ProbModelScheduler{}{\sched}(\eventually[\interval] \goalStates \setminus \inducedPathsDi{\eventuallyds[{\diOfInterval}] \goalStates})
-	\ProbModelScheduler{}{\sched}(\inducedPathsDi{\eventuallyds[{\diOfInterval}] \goalStates} \setminus \eventually[\interval] \goalStates).
\end{split}
\end{align*}
One then shows
\begin{displaymath}
%\label{eq:bounded:upperBoundsForError}
	\ProbModelScheduler{\ma}{\sched}(\eventually[\interval] \goalStates \setminus \inducedPathsDi{\eventuallyds[{\diOfInterval}] \goalStates}) \le \upperErrorBound(\interval) 
	 \text{\quad and \quad} 
	\ProbModelScheduler{\ma}{\sched}(\inducedPathsDi{\eventuallyds[{\diOfInterval}] \goalStates} \setminus \eventually[\interval] \goalStates) \le \lowerErrorBound(\interval).
\end{displaymath}
To this end, show for any $\stepBound \in \NN$ that   $1- (1+ \rateMax \digConstant)^{\stepBound} \cdot e^{- \rateMax \digConstant \stepBound}$
  is an upper bound for the probability of paths that induce more then $\stepBound$ digitization steps within the the first $\stepBound \digConstant$ time units.
Then, this probability can be related to the probability of paths in $\eventually[\interval] \goalStates \setminus \inducedPathsDi{\eventuallyds[{\diOfInterval}] \goalStates}$ and $\inducedPathsDi{\eventuallyds[{\diOfInterval}] \goalStates} \setminus \eventually[\interval] \goalStates$, respectively.

From Prop.~\ref{pro:bounded:dsBoundedProbEqual} and Prop.~\ref{prop:bounded:approxBoundedProb}, we immediately have Cor.~\ref{corr:bounded:approxBoundedProb}, which ensures that the value $\intervalBoundedReachProbMa[]$ can be approximated with arbitrary precision by computing $\ProbModelScheduler{\dma}{\diOfSched}(\eventuallyds[\diOfInterval] \goalStates)$ for a sufficiently small $\digConstant$.
\begin{restatable}{corollary}{corrApproxBoundedProb}
	\label{corr:bounded:approxBoundedProb}
	For MA $\ma$, scheduler $\sched \in \GMSched$, goal states $\goalStates \subseteq \States$, digitization constant $\digConstant \in \RRgz$ and time interval $\interval$
%	\begin{displaymath}
%		\intervalBoundedReachProbMa[]  \in
%	 [p - \lowerErrorBound(\interval),\,  p+ \upperErrorBound(\interval)] 
%	 \text{, where } p = 	\ProbModelScheduler{\dma}{\diOfSched}(\eventuallyds[\diOfInterval] \goalStates).
%	\end{displaymath}
	\begin{displaymath}
		\intervalBoundedReachProbMa[]  \in
		\ProbModelScheduler{\dma}{\diOfSched}(\eventuallyds[\diOfInterval] \goalStates)
		+
	 \Big[{-}\lowerErrorBound(\interval),\,  \upperErrorBound(\interval)\Big] 
	\end{displaymath}
%	\begin{displaymath}
%		\ProbModelScheduler{\dma}{\diOfSched}(\eventuallyds[\diOfInterval] \goalStates)  - \lowerErrorBound(\interval)
%		\le \intervalBoundedReachProbMa[]  \le
%		 \ProbModelScheduler{\dma}{\diOfSched}(\eventuallyds[\diOfInterval] \goalStates)  + \upperErrorBound(\interval).
%	\end{displaymath}
\end{restatable}
\noindent 
This generalizes existing results~\cite{HatefiH12,DBLP:journals/corr/GuckHHKT14} that only consider schedulers which maximize (or minimize) the corresponding probabilities.
More details are given in \tech{app:singleObj}.

Next, we lift Cor.~\ref{corr:bounded:approxBoundedProb} to multiple objectives $\obj = \objTuple$.
We define the satisfaction of a \emph{timed} reachability objective $\intervalBoundedReachObj$ for the digitization $\dma$ as
$
\dma, \sched \models \intervalBoundedReachObj \rel[i] \pointi{i} \text{ iff } \ProbModelScheduler{\dma}{\sched}(\eventuallyds[\diOfInterval] \goalStates)  \rel[i]\pointi{i}
$.
This allows us to consider notations like $\achievabilityQ{\dma}{\obj \rel \point}$, where $\obj$ contains one or more timed reachability objectives.
For a point $\point  = \pointTuple \in \RR^\numObjectives$ we consider the hyperrectangle
\begin{displaymath}
\errorBoundObjPoint = \bigtimes_{i =1}^{\numObjectives}  \big[ \pointi{i} - \lowerErrorBound[i] ,\, \pointi{i} + \upperErrorBound[i] \big] \subseteq \RR^\numObjectives\ \text{, where } \upperErrorBound[i] =
\begin{cases}
\upperErrorBound(\interval) & \text{if } \obj[i] = \intervalBoundedReachObj\\
0                           & \text{if } \obj[i] = \expReachRewObj
\end{cases}
\end{displaymath}
and $\lowerErrorBound[i]$ is defined similarly.
The next example shows how the set of achievable points of $\ma$ can be approximated using achievable points of $\dma$.
\begin{example}
	\label{ex:ma:bounded:multi:overUnderApprox}
Let $\obj[] = ( \intervalBoundedReachObj[1], \intervalBoundedReachObj[2])$ be two timed reachability objectives for an MA $\ma$ with digitization $\dma$ such that
$\lowerErrorBound[1] =   0.13$, $\upperErrorBound[1] = 0.22$,
$\lowerErrorBound[2] = 0.07$, and $\upperErrorBound[2] =  0.15$.
The blue rectangle in Fig.~\ref{fig:ma:bounded:multiObjApprox:rectangle} illustrates the set $\errorBoundObjPoint$ for the point $\point = (0.4, 0.3)$.
Assume $\achievabilityQ{\dma}{\obj \rel \point}$ holds for threshold relations ${\rel} = \{\ge, \ge\}$, i.e., $\point$ is achievable for the digitization $\dma$.
	From Cor.~\ref{corr:bounded:approxBoundedProb}, we infer that $\errorBoundObjPoint$ contains at least one point $\point'$ that is achievable for $\ma$.
	Hence, the bottom left corner point of the rectangle is  achievable for $\ma$.
	This holds for any rectangle $\errorBoundObjPointUniv{\obj}{\altpoint}$ with $\altpoint \in \achievablePoints$, where $\achievablePoints$ %= \{ \point \in \RR^2 \mid \achievabilityQ{\dma}{\obj \rel \point}\}$ 
	is the set of achievable points of $\dma$ denoted by
	the gray area\footnote{In the figure, $\underApprox$ partly overlaps  $\achievablePoints$, i.e., the green area also belongs to $\achievablePoints$.} in Fig.~\ref{fig:ma:bounded:multiObjApprox:coarse}.
	It follows that any point in $\underApprox$ (depicted by the green area) is achievable for $\ma$.
	On the other hand, an achievable point of $\ma$ has to be contained in a set $\errorBoundObjPointUniv{\obj}{\altpoint}$ for at least one $\altpoint \in \achievablePoints$.
	The red area depicts the points $\RR^\numObjectives \setminus \overApprox$ for which this is not the case, i.e., points that are not achievable for $\ma$.
%	It follows that $\underApprox$ is an under-approximation and $\overApprox$ is an over-approximation of the set of achievable points of $\ma$.
	The digitization constant $\digConstant$ controls the accuracy of the resulting approximation.
	Fig.~\ref{fig:ma:bounded:multiObjApprox:fine} depicts a possible result when a smaller digitization constant $\altDigConstant < \digConstant$ is considered.
\end{example}
	\begin{figure}[t]
	\centering
	\subfigure[The set $\errorBoundObjPoint$.]{
		\scalebox{\picscale}{
			\input{pics/ma_bounded_multi_errorArea_errobjPoint}
		}
		\label{fig:ma:bounded:multiObjApprox:rectangle}
	}
	\subfigure[Coarse approximation.]{
		\scalebox{\picscale}{
			\input{pics/ma_bounded_multi_overUnder_coarse}
		}
		\label{fig:ma:bounded:multiObjApprox:coarse}
	}
	\subfigure[Refined approximation.]{
		\scalebox{\picscale}{
			\input{pics/ma_bounded_multi_overUnder_fine}
		}
		\label{fig:ma:bounded:multiObjApprox:fine}
	}
	\vspace{-3mm}
	\caption{Approximation of achievable points.}
	\label{fig:ma:bounded:multiObjApprox}
	\vspace{-3mm}
\end{figure}
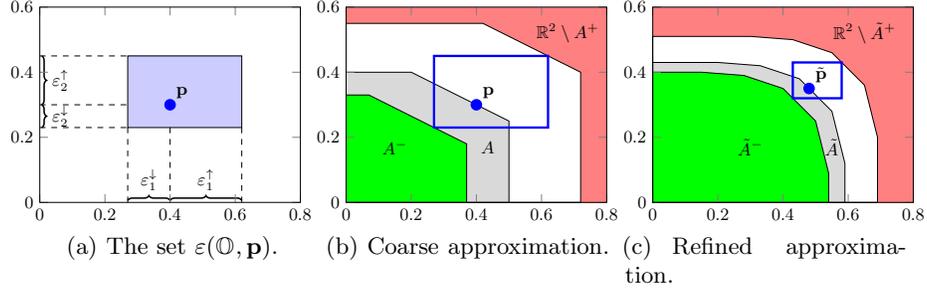

The observations from the example above are formalized in the following theorem.
The theorem also covers unbounded reachability objectives by considering the time interval $\interval = [0,\infty)$.
For expected reward objectives of the form $\expReachRewObj$ it can be shown that $\expReachRew{\ma}{\sched}{\rewFct[\rewFctIndex]} = \expReachRew{\dma}{\diOfSched}{\rewDmaFct[\rewFctIndex]}$.
This claim is similar to Proposition~\ref{pro:taSchedSameExpReachRew} and can be shown analogously. 
This enables multi-objective model checking of MAs with timed reachability objectives.
%The theorem considers two sets $\underApprox$ and $\overApprox$ (as in Example~\ref{ex:ma:bounded:multi:overUnderApprox}) and states that  the points in $\underApprox$ are achievable (for $\ma$) while the points in $\RR^\numObjectives \setminus \overApprox$ are not.
%We refer to  $\underApprox$  as  \emph{under-approximation} and to  $\overApprox$ as  \emph{over-approximation} of the set of achievable points of $\ma$. 
\begin{restatable}{theorem}{theoremMultiBounded}
	\label{thm:bounded:multiBounded}
	Let $\ma$ be an MA with digitization $\dma$.
	Furthermore, let  $\obj$ be (un)timed reachability or expected reward objectives with threshold relations ${\rel}$ and $|\obj| = d$.
	It holds that
%%	\begin{displaymath}
	$\underApprox \subseteq \{\point \in \RR^\numObjectives \mid \achievabilityQ{\ma}{\obj \rel \point} \} \subseteq \overApprox$ with:
	\begin{align*}
	\underApprox &= \{\point' \in \RR^\numObjectives \mid \forall \point \in \RR^\numObjectives \colon \point' \in \errorBoundObjPointUniv{\obj}{\point} \text{ implies } \achievabilityQ{\dma}{\obj \rel \point}  \} \text{ and} \\
	\overApprox &= \{\point' \in \RR^\numObjectives \mid \exists \point \in \RR^\numObjectives \colon \point' \in \errorBoundObjPointUniv{\obj}{\point} \text{ and } \achievabilityQ{\dma}{\obj \rel \point}\} .
	\end{align*}\end{restatable}

%\tq{describe how this theorem enables multi-obj model checking}

%% file: pics/ma_bounded_upperBound_ma.tex
\begin{tikzpicture}[scale=1]
    \draw[white, use as bounding box] (-0.8,-2.6) rectangle (0.8,1.1); 
    
    \node [state] (s0) at (0,0) {$\state_0$};
    \node [state, accepting] (s1) [on grid, below=15mm of s0] {$\state_1$};
    
    \initstateLeft{s0}

    \draw (s0) edge[markovian] node[pos=0.25, right=-2pt] {\scriptsize$2$} node[action] (s0markovian) {} node[pos=0.7, right=-2pt] {\scriptsize$0.5$} (s1);
    \draw (s0markovian) edge[markovian, bend left=60] node[left=-2pt] {\scriptsize$0.5$} (s0);
    \draw (s1) edge[markovian, loop below] node[below] {\scriptsize$1$} (s1);
\end{tikzpicture}

%% file: pics/ma_bounded_upperBound_paths.tex
\begin{tikzpicture}[scale=1]
%    \draw[white, use as bounding box] (-0.5,-3) rectangle (2.6,1.1); 
    
   	\draw (0,0) edge[->] (7.5,0);
    \foreach \n in {0,...,6}{
    	\node[] at (\n,0) {\scriptsize$\mid$};
    	\node[below] at (\n,0) {\scriptsize$\n\digConstant$};
   	}
    \node[below] at (6.9,0) {time};
    \draw (5,0) edge[thresholdColor, thick, dashed] (5,2.5);
%    \node[left] at (0,0) {$\ttime$};
%    \node[anchor=east]   [left=0 of (0,0)]  {$\ppath_1$:};
%    \node[anchor=east]   at ($(0,0) + (-0.2,0)$)  {$\ttime$:};
    
    \node [inner sep=0pt, outer sep=0pt] (p1-1) at (0,2) {\scriptsize$\state_0$};
    \node [action] (p1-2) [on grid, right=1 of p1-1] {};
    \node [action] (p1-3) [on grid, right=1 of p1-2] {};
    \node [inner sep=0pt, outer sep=0pt] (p1-4) [on grid, right=0.5 of p1-3]  {\scriptsize$\state_0$};
    \node [action] (p1-5) [on grid, right=1 of p1-4] {};
    \node [inner sep=0pt, outer sep=0pt] (p1-6) [on grid, right=0.8 of p1-5]  {\scriptsize$\state_1$};
    \node [neutralColor, action] (p1-7) [on grid, right=1 of p1-6] {};
    \node [neutralColor, inner sep=0pt, outer sep=0pt] (p1-8) [on grid, right=0.7 of p1-7]  {\scriptsize$\state_1$};
    \draw[->] (p1-1) edge (p1-4)  (p1-4) edge (p1-6)   (p1-6) edge[neutralColor] (p1-8);
%    \node[]  [anchor=west, on grid, left=0.1 of p1-1] {$\ppath_1$:};
    \node[anchor=east]   at ($(p1-1) + (-0.2,0)$)  {$\ppath_1$:};
    \node[neutralColor, anchor=west]   at ($(p1-8) + (0.1,0)$)  {\scriptsize$\cdots$};
    
%    \node[ anchor=west]   at ($(p1-8) + (0.6,0)$)  {\scriptsize$\numOfDS{\prefTime{\ppath_1}{5\digConstant}} = 5$};

    \node [inner sep=0pt, outer sep=0pt](p2-1) at (0,1.25) {\scriptsize$\state_0$};
    \node [inner sep=0pt, outer sep=0pt] (p2-2) [on grid, right=0.75 of p2-1]  {\scriptsize$\state_0$};
    \node [action] (p2-3) [on grid, right=1 of p2-2] {};
    \node [action] (p2-4) [on grid, right=1 of p2-3] {};
    \node [inner sep=0pt, outer sep=0pt] (p2-5) [on grid, right=0.75 of p2-4]  {\scriptsize$\state_1$};
    \node [neutralColor, action] (p2-6) [on grid, right=1 of p2-5] {};
%    \node [neutralColor, action] (p2-7) [on grid, right=1 of p2-6] {};
    \node [neutralColor, inner sep=0pt, outer sep=0pt] (p2-7) [on grid, right=0.75 of p2-6]  {\scriptsize$\state_1$};
    \draw[->] (p2-1) edge (p2-2)  (p2-2) edge (p2-5)   (p2-5) edge[neutralColor] (p2-7);
%    \node[]  [anchor=west, on grid, left=0.1 of p2-1] {$\ppath_1$:};
    \node[anchor=east]   at ($(p2-1) + (-0.2,0)$)  {$\ppath_2$:};
    \node[neutralColor, anchor=west]   at ($(p2-7) + (0.1,0)$)  {\scriptsize$\cdots$};
    
    \node [inner sep=0pt, outer sep=0pt](p3-1) at (0,0.5) {\scriptsize$\state_0$};
    \node [inner sep=0pt, outer sep=0pt] (p3-2) [on grid, right=0.5 of p3-1]  {\scriptsize$\state_0$};
    \node [inner sep=0pt, outer sep=0pt] (p3-3) [on grid, right=0.8 of p3-2]  {\scriptsize$\state_0$};
    \node [action] (p3-4) [on grid, right=1 of p3-3] {};
    \node [inner sep=0pt, outer sep=0pt] (p3-5) [on grid, right=0.4 of p3-4]  {\scriptsize$\state_0$};
    \node [, action] (p3-6) [on grid, right=1 of p3-5] {};
    \node [, inner sep=0pt, outer sep=0pt] (p3-7) [on grid, right=0.6 of p3-6]  {\scriptsize$\state_0$};
    \node [, inner sep=0pt, outer sep=0pt] (p3-8) [on grid, right=0.5 of p3-7]  {\scriptsize$\state_1$};
%    \node [neutralColor, action] (p3-9) [on grid, right=1 of p3-8] {};
    \node [neutralColor, inner sep=0pt, outer sep=0pt] (p3-9) [on grid, right=0.9 of p3-8]  {\scriptsize$\state_1$};
    \draw[->] (p3-1) edge (p3-2)  (p3-2) edge (p3-3)   (p3-3) edge[] (p3-5) (p3-5) edge[] (p3-7) (p3-7) edge[] (p3-8) (p3-8) edge[neutralColor] (p3-9);
%    \node[]  [anchor=west, on grid, left=0.1 of p3-1] {$\ppath_1$:};
    \node[anchor=east]   at ($(p3-1) + (-0.2,0)$)  {$\ppath_3$:};
    \node[neutralColor, anchor=west]   at ($(p3-9) + (0.1,0)$)  {\scriptsize$\cdots$};
%    
%    
%    
%    
%    \node [state] (s0) at (0,0) {$\state_0$};
%    \node [state] (s1) [on grid, below=15mm of s0] {\scriptsize$\state_1$};
%    \node [state] (s2) [on grid, below right=15mm of s1] {$\state_2$};
%    \node [state, accepting] (s3) [on grid, above right=15mm of s2] {$\state_3$};
%    \node [state] (s4) [on grid, above=15mm of s3] {$\state_4$};
%    
%    \initstateLeft{s0}
%    
%    \draw (s0) edge[markovian] node[left] {\scriptsize$\rate$} (s1);
%    \draw (s1) edge[probabilistic] node[pos=0.2, above] {\scriptsize$\act$} node[action] (s1alpha) {} node [pos=0.7, above] {\scriptsize$0.5$} (s4);
%    \draw (s1alpha) edge[probabilistic] node[pos=0.75, above] {\scriptsize$0.5$} (s3);
%    \draw (s1) edge[probabilistic] node[pos=0.4, above] {\scriptsize$\altact$} node[action] (s1beta) {} node [pos=0.9, above] {\scriptsize$1$} (s2);
%    \draw (s2) edge[markovian] node[pos=0.6, below] {\scriptsize$\rate$} (s3);
%    \draw (s3) edge[markovian] node[right] {\scriptsize$\rate$} (s4);
%    \draw (s4) edge[markovian, loop above] node[above] {\scriptsize$\rate$} (s4);
\end{tikzpicture}

%% file: pics/bounded_blobs.tex
\begin{tikzpicture}
\draw[white, use as bounding box] (-6,-0.2) rectangle (6.4,0.75); 
\draw[fill=orange, fill opacity=0.4] (-1.2,0) ellipse (4.0cm and 0.5cm);
\draw[fill=blue, fill opacity=0.3] (1.2,0) ellipse (4.0cm and 0.5cm);
\node at (-3.9,0) {$\eventually[\interval] \goalStates \setminus \inducedPathsDi{\eventuallyds[\diOfInterval]\goalStates}$};
\node at (3.9,0) {$\inducedPathsDi{\eventuallyds[\diOfInterval]\goalStates} \setminus \eventually[\interval] \goalStates$};
\node at (-0,0) {$\eventually[\interval] \goalStates \cap \inducedPathsDi{\eventuallyds[\diOfInterval]\goalStates}$};
\node[pin=20:{$\inducedPathsDi{\eventuallyds[\diOfInterval]\goalStates}$}] at (4.6,0.2) {};
\node[pin=160:{$\eventually[\interval]\goalStates$}] at (-4.6,0.2) {};
\end{tikzpicture}

%% file: pics/ma_bounded_multi_errorArea_errobjPoint.tex
\begin{tikzpicture}[scale=\plotscale]
\draw[white, use as bounding box] (0.2,-0.4) rectangle (4.5,3.45); 
\begin{axis}[
axis equal image,
xmin=0,
xmax = 0.8,
width=0.52\textwidth,
ymin=0,
ymax=0.6,
axis on top,
]

\addplot[fill=importantArea, very thin] coordinates  {(0.27, 0.23) (0.62,0.23) (0.62, 0.45)   (0.27, 0.45)} -- cycle ;

\plotPoint{(0.4,0.3)}{$\point$}

\draw[dashed] (axis cs:0.27,0.0) -- (axis cs:0.27,0.23);
\draw[dashed] (axis cs:0.4,0.0) -- (axis cs:0.4,0.23);
\draw[dashed] (axis cs:0.62,0.0) -- (axis cs:0.62,0.23);

\draw[thick, decoration={brace,raise=1pt},decorate] (axis cs:0.27,0.0) -- node[above=3pt] {$\lowerErrorBound[1]$} (axis cs:0.4,0.0); 
\draw[thick, decoration={brace,raise=1pt},decorate] (axis cs:0.40,0.0) -- node[above=3pt] {$\upperErrorBound[1]$} (axis cs:0.62,0.0);

\draw[thick, decoration={brace,raise=1pt, mirror},decorate] (axis cs:0.0,0.23) -- node[right=3pt] {$\lowerErrorBound[2]$} (axis cs:0,0.3); 
\draw[thick, decoration={brace,raise=1pt, mirror},decorate] (axis cs:0,0.3) -- node[right=3pt] {$\upperErrorBound[2]$} (axis cs:0,0.45); 
\draw[dashed] (axis cs:0,0.23) -- (axis cs:0.27,0.23);
\draw[dashed] (axis cs:0,0.3) -- (axis cs:0.27,0.3);
\draw[dashed] (axis cs:0,0.45) -- (axis cs:0.27,0.45);

\end{axis}
\end{tikzpicture}

%% file: pics/ma_bounded_multi_overUnder_coarse.tex
\begin{tikzpicture}[scale=\plotscale]
\draw[white, use as bounding box] (-0.2,-0.4) rectangle (4.5,3.45); 
\begin{axis}[
axis equal image,
xmin=0,
xmax = 0.8,
width=0.52\textwidth,
ymin=0,
ymax=0.6,
axis on top,
axis background/.style={fill=badArea},
]

\addplot[fill=white, very thin] coordinates  {(0.72, 0) (0.72,0.4) (0.42, 0.55)  (0.0, 0.55) (0,0)} -- cycle ;
\addplot[fill=neutralArea, very thin] coordinates  {(0.5, 0) (0.5,0.25) (0.2, 0.4)  (0.0, 0.4) (0,0)} -- cycle ;
\addplot[fill=goodArea, very thin] coordinates  {(0.37, 0) (0.37,0.18) (0.07, 0.33) (0.0, 0.33)   (0, 0)} -- cycle ;

\addplot[lineColor, very thick] coordinates  {(0.27, 0.23) (0.62,0.23) (0.62, 0.45)   (0.27, 0.45)} -- cycle ;
\plotPoint{(0.4,0.3)}{$\point$}

\node at (axis cs:0.15,0.17) {$\underApprox$};
\node at (axis cs:0.435,0.17) {$\achievablePoints$};
\node at (axis cs:0.68,0.52) {$\RR^2 \setminus \overApprox$};

\end{axis}
\end{tikzpicture}

%% file: pics/ma_bounded_multi_overUnder_fine.tex
\begin{tikzpicture}[scale=\plotscale]
\draw[white, use as bounding box] (-0.2,-0.4) rectangle (4.12,3.45); 
\begin{axis}[
axis equal image,
xmin=0,
xmax = 0.8,
width=0.52\textwidth,
ymin=0,
ymax=0.6,
axis on top,
axis background/.style={fill=badArea},
]

%\addplot[fill=orange, very thin] coordinates  {(0.72, 0) (0.72,0.4) (0.42, 0.55)  (0.0, 0.55) (0,0)} -- cycle ;
\addplot[fill=white, very thin] coordinates  {(0.69, 0) (0.69,0.2) (0.65,0.36) (0.55, 0.46) (0.43, 0.5)  (0.3, 0.51) (0,0.51) (0,0)} -- cycle ;
\addplot[fill=neutralArea, very thin] coordinates  {(0.59, 0) (0.59,0.12) (0.55,0.28) (0.45, 0.38) (0.33, 0.42)  (0.2, 0.43) (0,0.43) (0,0)} -- cycle ;
\addplot[fill=goodArea, very thin] coordinates  {(0.54, 0) (0.54,0.09) (0.5,0.25) (0.4, 0.35) (0.28, 0.39)  (0.15, 0.4) (0,0.4) (0,0)} -- cycle ;
%\addplot[fill=blue, very thin] coordinates  {(0.37, 0) (0.37,0.18) (0.07, 0.33) (0.0, 0.33)   (0, 0)} -- cycle ;

\addplot[lineColor, very thick] coordinates  {(0.43, 0.32) (0.58,0.32) (0.58, 0.43)   (0.43, 0.43)} -- cycle ;

% x down= 0.05    x up=0.10
% y down= 0.03    y up=0.08

\plotPoint{(0.48,0.35)}{$\tilde{\point}$}
\node at (axis cs:0.3,0.17) {$\altUnderApprox$};
\node at (axis cs:0.55,0.17) {$\tilde{\achievablePoints}$};
\node at (axis cs:0.65,0.52) {$\RR^2 \setminus \altOverApprox$};

\end{axis}
\end{tikzpicture}

%% file: 05-evaluation.tex
\section{Experimental Evaluation}
\label{sec:evaluation}

\paragraph{Implementation.}
We implemented multi-objective model checking of MAs into \storm~\cite{DBLP:journals/corr/DehnertJK016}.
The input model is given in the \prism language\footnote{We slightly extend the \prism language in order to describe MAs.} and translated into a sparse representation.
For MA $\ma$, the implementation performs a multi-objective analysis on the underlying MDP $\umdp$ or a digitization $\dma$ and infers (an approximation of) the achievable points of $\ma$ by exploiting the results from Sect.~\ref{sec:ma}.
For computing the achievable points of $\umdp$ and $\dma$, we apply the approach of \cite{ForejtKPatva12}. 
It repeatedly checks weighted combinations of the objectives (by means of \emph{value iteration}~\cite{Put94} -- a standard technique in single-objective MDP model checking) to refine an approximation of the set of achievable points.
This procedure is extended as follows. Full details can be found in \cite{quatmannMastersThesisMultiObjMA}.
\begin{itemize}
	\item We support $\ds$-bounded reachability objectives by combining the approach  of \cite{ForejtKPatva12} (which supports step-bounded reachability on MDPs) with techniques from single-objective MA analysis \cite{HatefiH12}.
	Roughly, we reduce $\ds$-bounded reachability to untimed reachability by storing the digitized time-epoch (i.e., the current number of digitization steps) into the state space.
	A blow-up of the resulting model is avoided by considering each time-epoch separately.
	\item In contrast to \cite{ForejtKPatva12}, we allow a simultaneous analysis of minimizing and maximizing expected reward objectives. This is achieved by performing additional preprocessing steps that comprise an analysis of end components. 
\end{itemize}
%More details of our extensions of \cite{ForejtKPatva12} can be found in \cite{} \tq{Cite thesis}.
%The digitization constant $\digConstant$ is chosen depending on the current weights assigned to the different objectives.
%A high weight for a time-bounded reachability objective means that 
%There are two types of approximations when time-bounded reachability objectives are considered 
%due to the digitization (cf. Theorem~\ref{thm:bounded:multiBounded}) and the successive refinement of the set of achievable points of $\dma$.
%Digitization constants are chosen such that Theorem~\ref{thm:bounded:multiBounded} yields an $\nicefrac{\eta}{2}$ approximation of the set of achievable points of $\ma$, given the achievable points of $\dma$. 
%\begin{remark}
%Both implementations (ours and the one considered in \cite{ForejtKPatva12}) are based on value iteration, which generally does not give any guarantee for the accuracy of the result (see, e.g., \cite{HaddadM14} for more details).
%This is a widespread issue for MDP model checking -- even in the single-objective case.
%%An implementation considering sound computations is considered for future work.
%\end{remark}

\noindent The source code including all material to reproduce the experiments is available at \url{http://www.stormchecker.org/benchmarks.html}.

\paragraph{Setup.}
%We evaluated the performance of our implementation on an HP BL685C G7 with 48 cores (2.0GHz each) and 192GB of RAM.
Our implementation uses a single core (2GHz) of a 48-core HP BL685C G7 limited to 20GB RAM.
The timeout (TO) is two hours.
For a model, a set of objectives, and a precision $\eta \in \RRgz$, we measure the time to compute an $\eta$-approximation\footnote{An $\eta$-approximation of $\achievablePoints \subseteq \RR^\numObjectives$ is given by $\underApprox, \overApprox \subseteq \RR^\numObjectives$ with $\underApprox \subseteq \achievablePoints \subseteq \overApprox$ and for all $\point \in \overApprox$ exists a $\altpoint \in \underApprox$ such that the distance between $\point$ and $\altpoint$ is at most $\eta$.} of the set of achievable points.
This set-up coincides with Pareto queries as discussed in~\cite{ForejtKPatva12}.
The digitization constant $\digConstant$ is chosen heuristically such that recalculations with smaller constants $\altDigConstant < \digConstant$ are avoided.
We set the precision for value-iteration to $\varepsilon = 10^{-6}$.
We use classical value iteration; the use of improved algorithms~\cite{HaddadM14} is left for future work.

\paragraph{Results for MAs.}
\begin{table}[tb]
	\centering
	\caption{Experimental results for multi-objective MAs.}
	\vspace{-2mm}
	\label{tab:maRes}
	\input{tables/maRes}
	\vspace{-3mm}
\end{table}
We consider four case studies:
(i) a \emph{job scheduler}~\cite{DBLP:journals/jacm/BrunoDF81}, see Sect.~\ref{sec:introduction};
(ii) a \emph{polling system}~\cite{Srinivasan:1991,DBLP:conf/concur/TimmerKPS12} containing a server processing jobs that arrive at two stations;
(iii) a \emph{video streaming client} buffering received packages and deciding when to start playback; and
(iv) a randomized \emph{mutual exclusion algorithm}~\cite{DBLP:conf/concur/TimmerKPS12}, a variant of~\cite{Pnueli1986} with a process-dependent random delay in the critical section.
%in which the delay in the critical section is a process-dependent rate.
Details on the benchmarks and the objectives are given in \tech{App:evaluation-details:ma}.
 
Tab.~\ref{tab:maRes} lists results.
For each instance we give the defining constants, the number of states of the MA  and the used $\eta$-approximation.
A multi-objective query is given by the triple $(l,m,n)$ indicating $l$ untimed, $m$ expected reward, and $n$ timed objectives.
For each MA and query we depict the total run-time of our implementation (time) and the number of vertices of the obtained under-approximation (\emph{pts}). 

Queries analyzed on the underlying MDP are solved efficiently on large models with up to millions of states.
For timed objectives the run-times increase drastically due to the costly analysis of digitized reachability objectives on the digitization, cf.~\cite{DBLP:journals/corr/GuckHHKT14}.
Queries with up to four objectives can be dealt with within the time limit.
Furthermore, for an approximation one order of magnitude better, the number of vertices of the result increases approximately by a factor three.  In addition, a lower digitization constant has then to be considered which often leads to timeouts in experiments with timed objectives.
 
\paragraph{Comparison with \prism~\cite{KNP11} and  \imca~\cite{DBLP:journals/corr/GuckHHKT14}.}
We compared the performance of our implementation with both \prism and \imca.
Verification times are summarized in Fig.~\ref{fig:eval:mdpScatterPlot}: On points above the diagonal, our implementation is faster.
For the comparison with \prism (no MAs), we considered the multi-objective MDP benchmarks from~\cite{ForejtKPatva12,ForejtKNPQtacas11}. 
Both implementations are based on~\cite{ForejtKPatva12}.
For the comparison with \imca (no multi-objective queries) we used the benchmarks from Tab.~\ref{tab:maRes}, with just a single objective.
We observe that our implementation is competitive.
Details are given in 
\iftech
\tech{App:evaluation-details:prism} and \tech{App:evaluation-details:imca}.
\else
\tech{}.
\fi
%Moreover, \prism gave an incorrect answer for a case study with infinite expected reward under any scheduler.
\pgfplotsset{every axis/.append style={
		legend style={font=\footnotesize,line width=1pt,mark size=6pt},
}}
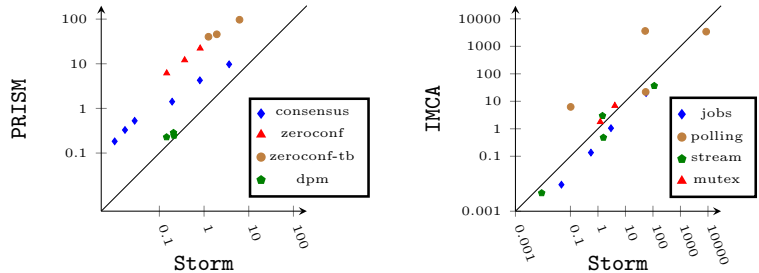
\begin{figure}[t] 
	\centering
	\subfigure{
	\scalebox{\picscale}{
		\input{pics/mdpScatterPlot}
	} 
		\label{fig:eval:scatter:prism}
  }
    \subfigure{
      \scalebox{\picscale}{
	    \input{pics/singleObjScatterPlot}
	  }
		\label{fig:eval:scatter:imca}
  }
  \vspace{-3mm}
	\caption{Verification times (in seconds) of our implementation and other tools.}
	\label{fig:eval:mdpScatterPlot}
	\vspace{-2mm}
\end{figure}

%% file: tables/maRes.tex
\setlength{\tabcolsep}{3pt}
\scriptsize
\begin{tabular}{crr|rr|rr|rrrr}
\multicolumn{3}{c|}{benchmark}                                        & \multicolumn{2}{c|}{$(\eventually, \eR, \eventually[\interval])$}                         & \multicolumn{2}{c|}{$(\eventually, \eR, \eventually[\interval])$}                         & \multicolumn{2}{c|}{$(\eventually, \eR, \eventually[\interval])$}                           & \multicolumn{2}{c}{$(\eventually, \eR, \eventually[\interval])$}                          \\
N(-K) & \multicolumn{1}{c}{\#states} & \multicolumn{1}{c|}{$\log_{10}(\eta)$} & \multicolumn{1}{c}{pts} & \multicolumn{1}{c|}{time} & \multicolumn{1}{c}{pts} & \multicolumn{1}{c|}{time} & \multicolumn{1}{c}{pts} & \multicolumn{1}{c|}{time}   & \multicolumn{1}{c}{pts} & \multicolumn{1}{c}{time}  \\ \hline
\multicolumn{3}{c|}{\textbf{job scheduling}}                       &\multicolumn{2}{c|}{$(0,3,0)$}                     &\multicolumn{2}{c|}{$(0,1,1)$}                     &\multicolumn{2}{c|}{$(1,3,0)$}                       &\multicolumn{2}{c}{$(1,1,2)$} \\ \hline
\multirow{2}{*}{10-2}    & \multirow{2}{*}{12\,554    }   & $-2$   & 9                       & 1.8                     & 9                       & 41                      & 15                    & \multicolumn{1}{r|}{435}    & 16             & \multicolumn{1}{r}{2\,322} \\
                         &                                & $-3$   & 44                      & 128                     & 21                      & 834                     &                  \TO{c|}{76}{0.002}                 &              \TO{c}{10}{0.013}              \\
\multirow{2}{*}{12-3}    & \multirow{2}{*}{116\,814   }   & $-2$   & 11                      & 42                      & 9                       & 798                     & 21                    & \multicolumn{1}{r|}{2\,026} &               \TO{c}{7}{0.023}              \\
                         &                                & $-3$   & 53                      & 323                     &                \TO{c|}{8}{0.005}                  &                   \TO{c|}{75}{0.003}                &               \TO{c}{2}{$\infty$}           \\
\multirow{2}{*}{17-2}    & \multirow{2}{*}{$4.6\cdot 10^6$} & $-2$ & 14                      & 1\,040                  &                \TO{c|}{1}{$\infty$}               & 22                    & \multicolumn{1}{r|}{4\,936} &               \TO{c}{2}{$\infty$}           \\
                         &                                & $-3$   & 58                      & 2\,692                  &                \TO{c|}{1}{$\infty$}               &                   \TO{c|}{68}{0.002}                &               \TO{c}{2}{$\infty$}           \\ \hline
\multicolumn{3}{c|}{\textbf{polling}}                     & \multicolumn{2}{c|}{$(0,2,0)$}                             & \multicolumn{2}{c|}{$(0,4,0)$}                    & \multicolumn{2}{c|}{$(0,0,2)$}                      & \multicolumn{2}{c}{$(0,2,2)$}               \\ \hline
\multirow{2}{*}{3-2}     & \multirow{2}{*}{1\,020  }      & $-2$   & 4                       & 0.3                     & 5                & 0.6                            & 3                    & \multicolumn{1}{r|}{130}     & 12             & \multicolumn{1}{r}{669}    \\
                         &                                & $-3$   & 4                       & 0.3                     & 5                & 0.8                            & 7                    & \multicolumn{1}{r|}{3\,030}  &              \TO{c}{16}{0.004}              \\
\multirow{2}{*}{3-3}     & \multirow{2}{*}{9\,858  }      & $-2$   & 5                       & 1.3                     & 8                & 23                             & 6                    & \multicolumn{1}{r|}{2\,530}  &              \TO{c}{17}{0.034}              \\
                         &                                & $-3$   & 6                       & 2.0                     & 19               & 3\,199                         &                  \TO{c|}{2}{0.113}                  &               \TO{c}{3}{$\infty$}           \\
\multirow{2}{*}{4-4}     & \multirow{2}{*}{827\,735}      & $-2$   & 10                      & 963                     & 20               & 4\,349                         &                  \TO{c|}{0}{$\infty$}               &               \TO{c}{2}{$\infty$}           \\
                         &                                & $-3$   & 11                      & 1\,509                  &               \TO{c|}{66}{0.001}                  &                  \TO{c|}{0}{$\infty$}               &               \TO{c}{2}{$\infty$}           \\ \hline
\multicolumn{3}{c|}{\textbf{stream}}                      & \multicolumn{2}{c|}{$(0,2,0)$}                             & \multicolumn{2}{c|}{$(0,1,1)$}                    & \multicolumn{2}{c|}{$(0,0,2)$}                      & \multicolumn{2}{c}{$(0,2,1)$}               \\ \hline
\multirow{2}{*}{30  }    & \multirow{2}{*}{1\,426     }   & $-2$   & 20                      & 0.9                     & 16                      & 90                      &  16                  &  \multicolumn{1}{r|}{55}     & 26             & \multicolumn{1}{r}{268}    \\
                         &                                & $-3$   & 51                      & 8.8                     & 46                      & 2\,686                  &  38                  & \multicolumn{1}{r|}{1\,341}  &            \TO{c}{91}{0.001}                \\ 
\multirow{2}{*}{250 }    & \multirow{2}{*}{94\,376    }   & $-2$   & 31                      & 50                      & 15                      & 5\,830                  &  16                  & \multicolumn{1}{r|}{4\,050}  &            \TO{c}{13}{0.036}                \\
                         &                                & $-3$   & 90                      & 184                     &                        \TO{c|}{2}{0.980}          &                         \TO{c|}{3}{0.094}           &            \TO{c}{2}{$\infty$}              \\
\multirow{2}{*}{1000}    & \multirow{2}{*}{$1.5 \cdot 10^6$} & $-2$& 41                      & 3\,765                  &                         \TO{c|}{1}{$\infty$}      &                         \TO{c|}{2}{0.542}           &            \TO{c}{2}{$\infty$}              \\
                         &                                & $-3$   &                         \TO{c|}{78}{0.002}        &                         \TO{c|}{1}{$\infty$}      &                         \TO{c|}{0}{$\infty$}        &            \TO{c}{2}{$\infty$}              \\ \hline
\multicolumn{3}{c|}{\textbf{mutex}}                       & \multicolumn{2}{c|}{$(0,0,3)$}                             & \multicolumn{2}{c|}{$(0,0,3)$}                    & \multicolumn{2}{c}{}                                & \multicolumn{2}{c}{}                        \\ \cline{1-7}
\multirow{2}{*}{2}       & \multirow{2}{*}{13\,476}       & $-2$   & 16                      & 351                     & 13                      & 1\,166                  &                     &                               &                &                            \\
                         &                                & $-3$   & 13                      & 2\,739                  &                        \TO{c|}{20}{0.018}         &                     &                               &                &                            \\ 
\multirow{1}{*}{3}       & \multirow{1}{*}{38\,453}       & $-2$   & 15                      & 2\,333                  &                        \TO{c|}{25}{0.013}         &                     &                               &                &                           
\end{tabular}

%% file: pics/mdpScatterPlot.tex
%\begin{tikzpicture}
%\begin{axis}[
%width=5.7cm,
%height=5.7cm,
%xmin=0.1,
%ymin=0.1,
%ymax=25000,
%xmax=25000,
%xmode=log,
%ymode=log,
%axis x line=bottom,
%axis y line=left,
%x label style={at={(axis description cs:0.5,0.01)},anchor=north},
%y label style={at={(axis description cs:0.11,.5)},anchor=south},
%legend pos=south east, 
%ytick={1, 10, 60, 600, 3600, 15000}, 
%yticklabels={1, 10, 60, 600, 3600, 15000},
%%extra y ticks = { 14400, 21600 },extra y tick labels = {TO,MO},extra y tick style = { grid = major },
%xtick={1, 10, 60, 600, 3600, 15000},
%xticklabels={1, 10,60, 600, 3600, 15000},
%% extra x ticks = {14400, 21600},extra x tick labels = {TO,MO},extra x tick style = { grid = major },
%xlabel=\storm,
%ylabel=\prism,
%yticklabel style={font=\tiny},
%xticklabel style={rotate=290, anchor=west, font=\tiny},
%legend style={font=\tiny}]
%\addplot[ 
%% % clickable coords={\thisrow{label}},
%scatter/classes={
%	consensus={mark=diamond*,blue,mark size=1.30},%
%	zeroconf={mark=triangle*,red,mark size=1.3},%
%	zeroconf-tb={mark=*,draw=brown,fill=brown, mark size=1.30},%
%	team={mark=square*,purple, mark size=1.3},%%
%	dpm={mark=pentagon*,green!50!black, mark size=1.3}%%
%},
%scatter,only marks,
%scatter src=explicit symbolic]
%table[x=stormverif,y=prismverif,meta=benchmark]
%{pics/mdpScatterPlot.csv};
%\addplot[no marks,forget plot] coordinates
%{(0.10,0.10) (20000,20000) };
%%\addplot[no marks,forget plot, dashed] coordinates
%%{(10,100) (720,7200) };
%\legend{consensus,zeroconf,zeroconf-tb,team-form.,dpm}
%\end{axis}
%\end{tikzpicture} 
\begin{tikzpicture}
\begin{axis}[
width=4.4cm,
height=4.4cm,
xmin=0.005,
ymin=0.005,
ymax=200,
xmax=200,
xmode=log,
ymode=log,
axis x line=bottom,
axis y line=left,
x label style={at={(axis description cs:0.5,0.01)},anchor=north},
y label style={at={(axis description cs:0.11,.5)},anchor=south},
%legend pos=south east, 
ytick={0.001, 0.1, 1, 10, 100}, 
yticklabels={0.001, 0.1, 1, 10, 100}, 
%extra y ticks = { 14400, 21600 },extra y tick labels = {TO,MO},extra y tick style = { grid = major },
xtick={0.001, 0.1, 1, 10, 100},
xticklabels={0.001, 0.1, 1, 10, 100},
% extra x ticks = {14400, 21600},extra x tick labels = {TO,MO},extra x tick style = { grid = major },
xlabel=\storm,
ylabel=\prism,
yticklabel style={font=\tiny},
xticklabel style={rotate=290, anchor=west, font=\tiny},
legend style={font=\tiny,at={(axis cs:10,0.01)},anchor=south west}]
\addplot[ 
% % clickable coords={\thisrow{label}},
scatter/classes={
	consensus={mark=diamond*,blue,mark size=1.30},%
	zeroconf={mark=triangle*,red,mark size=1.3},%
	zeroconf-tb={mark=*,draw=brown,fill=brown, mark size=1.30},%
	dpm={mark=pentagon*,green!50!black, mark size=1.3}%%
},
scatter,only marks,
scatter src=explicit symbolic]
table[x=stormverif,y=prismverif,meta=benchmark]
{pics/mdpScatterPlot.csv};
\addplot[no marks,forget plot] coordinates
{(0.005,0.005) (200,200) };
%\addplot[no marks,forget plot, dashed] coordinates
%{(10,100) (720,7200) };
\legend{consensus,zeroconf,zeroconf-tb,dpm}
\end{axis}
\end{tikzpicture} 

%% file: pics/singleObjScatterPlot.tex
\begin{tikzpicture}
\begin{axis}[
width=4.4cm,
height=4.4cm,
xmin=0.001,
ymin=0.001,
ymax=30000,
xmax=30000,
xmode=log,
ymode=log,
axis x line=bottom,
axis y line=left,
x label style={at={(axis description cs:0.5,0.01)},anchor=north},
y label style={at={(axis description cs:0.11,.5)},anchor=south},
%legend pos=south east, 
ytick={0.001, 0.1, 1, 10, 100, 1000, 10000}, 
yticklabels={0.001, 0.1, 1, 10, 100, 1000, 10000}, 
%extra y ticks = { 14400, 21600 },extra y tick labels = {TO,MO},extra y tick style = { grid = major },
xtick={0.001, 0.1, 1, 10, 100, 1000, 10000},
xticklabels={0.001, 0.1, 1, 10, 100, 1000, 10000},
% extra x ticks = {14400, 21600},extra x tick labels = {TO,MO},extra x tick style = { grid = major },
xlabel=\storm,
ylabel=\imca,
yticklabel style={font=\tiny},
xticklabel style={rotate=290, anchor=west, font=\tiny},
legend style={font=\tiny,at={(axis cs:400,0.003)},anchor=south west}]

\addplot[ 
% % clickable coords={\thisrow{label}},
scatter/classes={
	jobs={mark=diamond*,blue,mark size=1.30},%
	polling={mark=*,draw=brown,fill=brown, mark size=1.30},%
	stream={mark=pentagon*,green!50!black, mark size=1.3},%%
	mutex={mark=triangle*,red,mark size=1.3}%
},
scatter,only marks,
scatter src=explicit symbolic]
table[x=storm,y=imca,meta=benchmark]
{pics/singleObjScatterPlot.csv};

%\addplot[ 
%% % clickable coords={\thisrow{label}},
%scatter/classes={
%	jobs={mark=diamond*,blue,mark size=1.30},%
%	polling={mark=*,draw=brown,fill=brown, mark size=1.30},%
%	stream={mark=pentagon*,green!50!black, mark size=1.3},%%
%	mutex={mark=triangle*,red,mark size=1.3}%
%},
%scatter,only marks,
%scatter src=explicit symbolic]
%table[x=storm-single,y=imca,meta=benchmark]
%{pics/singleObjScatterPlot.csv};

\addplot[no marks,forget plot] coordinates
{(0.001,0.001) (30000,30000) };
%\addplot[no marks,forget plot, dashed] coordinates
%{(10,100) (720,7200) };
\legend{jobs,polling,stream,mutex}
\end{axis}
\end{tikzpicture} 

%% file: 07-conclusion.tex
\section{Conclusion}
\label{sec:conclusion}
We considered multi-objective verification of Markov automata, including in particular timed reachability objectives.
The next step is to apply our algorithms to the manifold applications of MA, such as generalized stochastic Petri nets to enrich the analysis possibilities of such nets.

\subsubsection{Acknowledgement.}
This work was supported by the CDZ project CAP (GZ 1023).

%% file: E-extra-definitions.tex
\section{Additional Preliminaries}
\subsection{Models with Rewards}
\label{App:RewardDefs}

We extend the models with rewards.

%Figure~\ref{fig:models_maRew} depicts the MA with the two reward functions.
%We use tuples $(x_1, x_2) \in \RRnn \times \RRnn$ to denote the values of $\rewFct[1]$ and $\rewFct[2]$ ($x_i$ refers to the value of $\rewFct[i]$ for $i \in \{1,2\}$).
%A tuple next to a state represents state rewards.
%Tuples next to transition arrows denote action rewards, where we use $\rewSeparator$ to separate them from rates, actions, or probabilities.
%Rewards that are not depicted explicitly are assumed to be zero.

\begin{definition}[Markov decision process \cite{Put94}]
	\label{def:app:MarkovDecisionProcess}
	A \emph{Markov decision process} (MDP) is a tuple $\mdpDef$, where
	$\States, \sinit, \Actions, \numRewFunctions$ are as in Definition~\ref{def:MarkovAutomaton}, $\rewMdpFct[1], \dots, \rewMdpFct[\numRewFunctions]$ are \emph{action reward functions} $\rewMdpFct[i] \colon \States \times \Actions \to \RRnn$, and $\probP \colon \States \times \Actions \times \States \to [0,1]$ is a \emph{transition probability function} satisfying 
	$		\sum_{\state' \in \States}\probP(\state, \act, \state') \in \{0,1\}$
	for all $\state \in \States$ and $\act \in \Actions$.	  
\end{definition}
The reward $\rewMdp{\state}{\act}$ is collected when choosing action $\act$ at state $\state$.
Note that we do not consider state rewards for MDPs.

\begin{definition}[Underlying MDP]
	\label{def:app:UnderlyingMdp}
		For MA $\maDef$ with transition probabilities $\probP$ the underlying MDP of $\ma$ is given by $\umdpDef$, where for each $i \in \{1, \dots, \numRewFunctions\}$
		\begin{displaymath}
\rewUMdp[i]{\state}{\act} = 
	\begin{cases}
	\rewAct[i]{\state}{\act} & \text{if } \state \in \PS \\
	\rewAct[i]{\state}{\markovianAct} + \nicefrac{1}{\rateAtState{\state}} \cdot \rewState[i]{\state} & \text{if } \state \in \MS \text{ and } \alpha = \markovianAct \\
	0 &\text{otherwise}.%\ .
	\end{cases}
	\end{displaymath}
\end{definition}
The reward functions $\rewUMdpFct[1], \dots, \rewUMdpFct[\numRewFunctions]$ incorporate the action and state rewards of $\ma$ where the state rewards are multiplied with the expected sojourn times $\nicefrac{1}{\rateAtState{\state}}$ of states $\state \in \MS$.

\begin{definition}[Digitization of an MA]
	\label{def:app:dma}
	For an MA $\ma = (\maIngredientsExceptRew, \{\rewFct[1],\allowbreak \dots, \rewFct[\numRewFunctions]\})$ with transition probabilities $\probP$ and a digitization constant $\digConstant \in \RRgz$, the digitization of $\ma$ \wrt $\digConstant$ is given by the MDP $\dmaDef$, where  $\probPdig$ is as in Definition~\ref{def:dma}  and for each $i \in \{1, \dots, \numRewFunctions\}$
		\begin{displaymath}
		\rewDma[i]{\state}{\act} = 
		\begin{cases}
		\rewAct[i]{\state}{\act} & \text{if } \state \in \PS \\
		\big(\rewAct[i]{\state}{\markovianAct} + \nicefrac{1}{\rateAtState{\state}} \cdot \rewState[i]{\state}\big) \cdot \big(1-e^{-\rateAtState{\state} \digConstant}\big)  & \text{if } \state \in \MS \text{ and } \alpha = \markovianAct \\
		0 &\text{otherwise}.%\ .
		\end{cases}
		\end{displaymath}
\end{definition}

\subsection{Measures}
\subsubsection{Probability measure.}
\label{App:ProbMeasure}
Given a scheduler $\sched \in \GMSched[]$, the probability measure $\ProbModelScheduler{\ma}{\sched}$ is defined for measurable sets of infinite paths of MA $\ma$.
This is achieved by considering the probability measure $\StepMeasure{\ma}{\sched}{\ppath}$ for transition steps.
For a history $\ppath \in \FPaths$ with $\state = \last{\ppath}$ and a measurable set of transition steps $T \subseteq \RRnn \times \Actions \times \States$ we have 
\begin{displaymath}
\StepMeasure{\ma}{\sched}{\ppath}(T) = 
\begin{dcases}
\sum_{(0, \act, \state') \in T} \schedEval{\sched}{\ppath}{\act} \cdot \probP(\state, \act, \state') & \text{if } \state\in \PS \\
%\int_{\substack{\ttime \in \RRnn\\ (\ttime, \markovianAct, \state') \in T}} \rateAtState{\state} \cdot e^{-\rateAtState{\state} \ttime } \cdot \sum_{(\ttime, \markovianAct, \state') \in T} \probP(\state, \markovianAct, \state') \diff \ttime & \text{if } \state \in \MS
\int_{\substack{\{ \ttime \mid (\ttime, \markovianAct, \state') \in T\}}} \rateAtState{\state} \cdot e^{-\rateAtState{\state} \ttime } \cdot \sum_{(\ttime, \markovianAct, \state') \in T} \probP(\state, \markovianAct, \state') \diff \ttime & \text{if } \state \in \MS
\end{dcases}
\end{displaymath}
$\ProbModelScheduler{\ma}{\sched}$ is obtained by lifting $\StepMeasure{\ma}{\sched}{\ppath}$ to sequences of transition steps (i.e., paths).
More information can be found in \cite{Neuhausser2010,HatefiH12}.
To simplify the notations, we write
 $\ProbModelScheduler{\ma}{\sched}(\ppath)$ instead of $\ProbModelScheduler{\ma}{\sched}(\{\ppath\})$.
 For a set of finite paths $\Pathset \subseteq \FPaths[\ma]$ we set
 $\ProbModelScheduler{\ma}{\sched}(\Pathset)$ = $\ProbModelScheduler{\ma}{\sched}(\cylinderset{\Pathset})$, where $\cylinderset{\Pathset}$ is the \emph{Cylinder} of $\Pathset$ given by
\begin{displaymath}
	\cylinderset{\Pathset} = \{ \ppath \pathTransTimed{n} \state_{n+1} \pathTransTimed{n+1} \dots  \in \IPaths[\ma] \mid  \ppath \in \Pathset \}.
\end{displaymath}

%
%\begin{definition}[Cylinder of a Set of Finite Paths]
%	\label{def:Cylinder}
%	The cylinder of a set of finite paths $\Pathset \subseteq \FPaths[\ma]$ of an MA $\ma$ is defined as
%	\begin{displaymath}
%	\cylinderset{\Pathset} = \{ \ppath \pathTransTimed{n} \state_{n+1} \pathTransTimed{n+1} \dots  \in \IPaths[\ma] \mid  \ppath \in \Pathset \}.
%	\end{displaymath}
%\end{definition}

\subsubsection{Expected reward.}
\label{App:ExpectedReward}
%The expected reward $\expReachRewMa$ of MA $\ma$ under scheduler $\sched$ retrieves the expected amount of reward collected \wrt reward function $\rewFct$ until reaching a goal state in $\goalStates \subseteq \States$.

%\begin{definition}[Reward of a Path]	\label{def:models:rewOfPath}
We fix a reward function $\rewFct$ of the MA $\ma$.
The reward of a finite path $\ppath' = \pathFseqTimed \in \FPaths$  is given by
%For a path , the reward of $\ppath'$ \wrt $\rewFct$
\begin{displaymath}
\rewPath[\ma]{\rewFct}{\ppath'} = \sum_{i=0}^{\length{\ppath'}-1} \rewState{\state_i} \cdot \timeOfStamp[i] + \rewAct{\state_i}{\actionOfStamp[i]}.%\ .
\end{displaymath}
Intuitively, $\rewPath[\ma]{\rewFct}{\ppath'}$ is the sum over the rewards obtained in every step $\pathstepTimed{i}$ depicted in the path $\ppath'$.
The reward obtained in step $i$ is composed of the state reward of $\state_i$ multiplied with the sojourn time $\timeOfStamp[i]$  as well as  the action reward given by $\state_i$ and $\actionOfStamp[i]$.
State rewards assigned to probabilistic states do not affect the reward of a path as the sojourn time in such states is zero.

For an infinite path $\ppath = \pathIseqTimed \in \IPaths$, the reward of $\ppath$ up to a set of goal states $\goalStates \subseteq \States$ is given by
\begin{displaymath}
\rewPathToTarget[\ma]{\rewFct}{\ppath}{\goalStates} =
\begin{cases}
\rewPath[\ma]{\rewFct}{\pref{\ppath}{n}} & \text{ if } n = \min \{i\ge 0 \mid \state_i \in \goalStates \}\\%\state_n \in \goalStates \text{ and } \state_i \notin \goalStates \text{ for all } 0 \le i < n \\
\lim_{n\to \infty} \rewPath[\ma]{\rewFct}{\pref{\ppath}{n}}   & \text{ if } \state_i \notin \goalStates \text{ for all } i \ge 0\ .
\end{cases}
\end{displaymath}
Intuitively, we stop collecting reward as soon as $\ppath$ reaches a state in $\goalStates$.
If no state in $\goalStates$ is reached, reward is accumulated along the infinite path, which potentially yields an infinite reward.
The expected reward $\expReachRewMa$ is the expected value of the function $\rewPathToTarget[\ma]{\rewFct}{\cdot}{\goalStates} \colon \IPaths[\ma] \to \RRnn$, i.e.,
%\begin{definition}[Expected Reward]
%\label{def:ExpectedReward}
%Given an MA $\ma$, a reward function $\rewFct$, a scheduler $\sched \in \GMSched[\ma]$, and a set of goal states $\goalStates$, the \emph{expected reward} is given by
\begin{displaymath}
\expReachRewMa =
\int_{\ppath \in \IPaths[\ma]} \rewPathToTarget[\ma]{\rewFct}{\ppath}{\goalStates} \diff\ProbModelScheduler{\ma}{\sched}(\ppath).
\end{displaymath}
%\end{definition}

%% file: G-proofsmomc.tex
\section{Proofs About Sets of Achievable Points}
\label{app:proofsMOMC} 
\subsection{Proof of Theorem~\ref{thm:deterministicTADoesNotSuffice}}

\theoremDeterministicTADoesNotSuffice*

\begin{proof}
Consider the MA $\ma$ in Fig.~\ref{fig:maAchievablePoints:ma} with objectives $\obj = (\reachObjUniv{\{s_2\}}, \reachObjUniv{\{s_4\}})$, relations ${\rel} = (\ge, \ge)$, and point $\point=(0.5,0.5)$.
We have $\achievabilityQ{\ma}{\obj \rel \point}$ (A scheduler achieving both objectives is given in Example~\ref{ex:timingIsImportant}).
However, there are only two deterministic time abstract schedulers for $\ma$: 
\begin{displaymath}
\sched_\act\colon \text{ always choose }  \act \qquad  \text{and} \qquad \sched_\altact \colon \text{ always choose } \altact
\end{displaymath}
and it holds that $\ma, \sched_\act \not\models \reachObjUniv{\{s_4\}} \ge 0.5$ and $\ma,\sched_\altact \not\models \reachObjUniv{\{s_2\}} \ge 0.5$.
\qed\end{proof}

\subsection{Proof of Proposition~\ref{prop:achievablePointsIsConvex}}
\propositionAchievablePointsIsConvex*

\begin{proof}
Let $\ma$ be an MA and let $\obj= \objTuple$ be objectives with relations ${\rel} = \relTuple$ and points $\point_1, \point_2 \in \RR^\numObjectives$ such that $\achievabilityQ{\ma}{\obj \rel \point_1}$ and $\achievabilityQ{\ma}{\obj \rel \point_2}$ holds.
For $i \in {1,2}$, let $\sched_i \in \GMSched$ be a scheduler satisfying $\ma, \sched_i \models \obj \rel \point_i$.
Consider some $w \in [0,1]$. The point $\point = w \cdot \point_1 + (1-w) \cdot \point_2$ is achievable with the  scheduler that makes an initial one-off random choice:
\begin{itemize}
	\item with probability $w$ mimic $\sched_1$ and
	\item with probability $1-w$ mimic $\sched_2$.
\end{itemize}
Hence, $\achievabilityQ{\ma}{\obj \rel \point}$, implying that the set of achievable points is convex.
\qed\end{proof}

\subsection{Proof of Theorem~\ref{Thm:InfinitePolytope}}
\label{app:proofsMOMC:infinite}
\theoremInfinitePolytope*

\begin{proof}
	We show that the claim holds for the MA $\ma$ in Fig.~\ref{fig:maAchievablePoints:ma} with objectives $\obj = (\reachObjUniv{\{s_2\}}, \boundedReachObjUniv{\{s_4\}}{{[0,2]}})$ and relations ${\rel} = (\ge, \ge)$.
	
	For the sake of contradiction assume that the polytope $\achievablePoints = \{\point \in \RR^2 \mid \achievabilityQ{\ma}{\obj \rel \point} \}$ is finite.
	Then, there must be two distinct vertices $\point_1,\point_2$ of $\achievablePoints$ such that $\{w \cdot \point_1 + (1-w) \cdot \point_2 \mid w \in [0,1]\}$ is a face of $\achievablePoints$.
	In particular, this means that $\point = 0.5 \cdot \point_1 + 0.5 \cdot \point_2$ is achievable but $\point_\varepsilon = \point + (0,\varepsilon)$ is not achievable for all $\varepsilon >0$.
	We show that there is in fact an $\varepsilon$ for which $\point_\varepsilon$ is achievable, contradicting our assumption that $\achievablePoints$ is finite.
	
	For $i \in {1,2}$, let $\sched_i \in \GMSched$ be a scheduler satisfying $\ma, \sched_i \models \obj \rel \point_i$.
	$\sched_1 \neq \sched_2$ has to hold as the schedulers achieve different vertices of $\achievablePoints$.
	The point $\point$ is achievable with the \emph{randomized} scheduler $\sched$ that mimics $\sched_1$ with probability 0.5 and mimics $\sched_2$ otherwise.
	Consider $\ttime = -\log(\ProbModelScheduler{\ma}{\sched}(\eventually \{\state_2\}))$ and the \emph{deterministic} scheduler $\sched'$ given by
	\begin{displaymath}
		\schedEval{\sched'}{\state_0 \pathTransUniv{\ttime_0} \state_1}{\act} =
		\begin{cases}
		1 & \text{if } \ttime_0 > \ttime) \\
		0 & \text{otherwise.}
		\end{cases}
	\end{displaymath}
	$\sched'$ satisfies $\ProbModelScheduler{\ma}{\sched'}(\eventually \{\state_2\}) = e^{-\ttime} =  \ProbModelScheduler{\ma}{\sched}(\eventually \{\state_2\})$.
	Moreover, we have
	\begin{displaymath}
	\ProbModelScheduler{\ma}{\sched'}(\eventually[{[0,\ttime]}] \{\state_3\}) = \ProbModelScheduler{\ma}{\sched'}(\eventually \{\state_3\})=  \ProbModelScheduler{\ma}{\sched}(\eventually \{\state_3\}) > \ProbModelScheduler{\ma}{\sched}(\eventually[{[0,\ttime]}] \{\state_3\}),
	\end{displaymath}
	where the last inequality is due to $\sched \neq \sched'$.
	While the probability to reach $s_3$ is equal under both schedulers, $s_3$ is reached earlier when $\sched'$ is considered. This increases the probability to reach $s_4$ in time, i.e., $\ProbModelScheduler{\ma}{\sched'}(\eventually[{[0,2]}] \{\state_4\}) > \ProbModelScheduler{\ma}{\sched}(\eventually[{[0,2]}] \{\state_4\})$.
%	The latter holds as $\sched \neq \sched'$ and $\sched'$ is the unique scheduler that maximizes the probability to reach $\state_4$ in time, given that $\state_2$ has to be reached with probability  $\ProbModelScheduler{\ma}{\sched}(\eventually \{\state_2\})$.
	It follows that $\ma, \sched' \models \obj \rel \point_\varepsilon$ for some $\varepsilon>0$.
\qed\end{proof}

%% file: A-proofsunbounded.tex
\section{Proofs for Untimed Reachability}
\label{app:proofsUnbounded}
\subsection{Proof of Lemma~\ref{lem:taSchedSamePathProb}}
\lemmaTaSchedSamePathProb*
\begin{proof}
	The proof %uses the definition of the probability measure $\ProbModelScheduler{\ma}{\sched}$ of $\ma$ (cf. Definition~\ref{def:ProbabilityMeasureOfMA})  and 
	is by induction over  the length of the considered path $\length{\pathTa} = n$. 
	Let $\maDef$ and $\umdpDef$.
	If $n=0$, then $\{\pathTa\} = \inducedPathsTa{\pathTa} = \{\sinit\}$. Hence, $\ProbModelScheduler{\ma}{\sched}(\inducedPathsTa{\pathTa}) = 1 = \ProbModelScheduler{\umdp}{\taOfSched}(\pathTa)$.
	In the induction step, we assume that the lemma holds for a fixed path $\pathTa \in \FPaths[\umdp]$ with length $\length{\pathTa}=n$ and $\last{\pathTa} = \state$. Consider the path $\pathTa \pathTransTa{} \state' \in \FPaths[\umdp]$.
	\paragraph{\underline{Case $\state \in \PS$:}} It follows that
	\begin{align*}
	\ProbModelScheduler{\ma}{\sched}(\inducedPathsTa{\pathTa \pathTransTa{} \state'})
	&= \int_{\ppath \in \inducedPathsTa{\pathTa}} \schedEval{\sched}{\ppath}{\act} \cdot \probP(\state, \act, \state') \diff\ProbModelScheduler{\ma}{\sched}(\ppath) \\
	%&= \probP(\state, \act, \state') \cdot \int_{\ppath \in \inducedPathsTa{\pathTa}} \schedEval{\sched}{\ppath}{\act} \diff\ProbModelScheduler{\ma}{\sched}(\ppath) \\
	&= \probP(\state, \act, \state') \cdot \int_{\ppath \in \inducedPathsTa{\pathTa}} \schedEval{\sched}{\ppath}{\act} \diff\ProbModelScheduler{\ma}{\sched}(\{\ppath\} \cap \inducedPathsTa{\pathTa}) \\
	&= \probP(\state, \act, \state')  \cdot \int_{\ppath \in \inducedPathsTa{\pathTa}} \schedEval{\sched}{\ppath}{\act} \diff\big[\ProbModelScheduler{\ma}{\sched}(\ppath \mid \inducedPathsTa{\pathTa}) \cdot \ProbModelScheduler{\ma}{\sched}(\inducedPathsTa{\pathTa})\big] \\
	&= \ProbModelScheduler{\ma}{\sched}(\inducedPathsTa{\pathTa}) \cdot \probP(\state, \act, \state') \cdot \int_{\ppath \in \inducedPathsTa{\pathTa}} \schedEval{\sched}{\ppath}{\act} \diff\ProbModelScheduler{\ma}{\sched}(\ppath \mid \inducedPathsTa{\pathTa})  \\
	&= \ProbModelScheduler{\ma}{\sched}(\inducedPathsTa{\pathTa}) \cdot \probP(\state, \act, \state') \cdot \schedEval{\taOfSched}{\pathTa}{\act}  \\
	&\eqInductionHypothesis \ProbModelScheduler{\umdp}{\taOfSched}(\pathTa) \cdot \probP(\state, \act, \state') \cdot \schedEval{\taOfSched}{\pathTa}{\act}  \\
	&= \ProbModelScheduler{\umdp}{\taOfSched}(\pathTa \pathTransTa{} \state').%\ .
	\end{align*}
	\paragraph{\underline{Case $\state \in \MS$:}}  As $s \in \MS$ we have $\act = \markovianAct$ and it follows
	\begin{align*}
	\ProbModelScheduler{\ma}{\sched}(\inducedPathsTa{\pathTa \pathTransUniv{\markovianAct} \state'})
	&= \int_{\ppath \in \inducedPathsTa{\pathTa}} \int_{0}^{\infty} \rateAtState{\state}\cdot e^{-\rateAtState{\state} \ttime} \cdot \probP(\state, \markovianAct, \state') \diff\ttime \diff\ProbModelScheduler{\ma}{\sched}(\ppath) \\
	&= \probP(\state, \markovianAct, \state')\cdot \int_{\ppath \in \inducedPathsTa{\pathTa}} \int_{0}^{\infty} \rateAtState{\state}\cdot e^{-\rateAtState{\state} \ttime} \diff\ttime \diff\ProbModelScheduler{\ma}{\sched}(\ppath) \\
	%&= \probP(\state, \markovianAct, \state')\cdot \int_{\ppath \in \inducedPathsTa{\pathTa}} 1 \diff\ProbModelScheduler{\ma}{\sched}(\ppath) \\
	&= \probP(\state, \markovianAct, \state')\cdot \ProbModelScheduler{\ma}{\sched}(\inducedPathsTa{\pathTa}) \\
	&\eqInductionHypothesis \probP(\state, \markovianAct, \state')\cdot   \ProbModelScheduler{\umdp}{\taOfSched}(\pathTa) \\
	&= \ProbModelScheduler{\umdp}{\taOfSched}(\pathTa \pathTransUniv{\markovianAct} \state').%\ .
	\end{align*}
\qed\end{proof}

\subsection{Proof of Proposition~\ref{pro:taSchedSameReachProb}}
\label{app:proofsUnbounded:taSchedSameReachProb}

\propositionTaSchedSameReachProb*

\begin{proof}
	\label{proof:pro:taSchedSameReachProb}
	Let $\Pathset$ be the set of finite time-abstract  paths of $\umdp$  that end at the first visit of a state in $\goalStates$, i.e.,
	\begin{displaymath}
	\Pathset = \{\pathFseqTa \in \FPaths[\umdp] \mid \state_n \in \goalStates \text{ and } \forall i < n \colon  \state_i \notin \goalStates \}.
	\end{displaymath}
	Every path $\ppath \in \eventually \goalStates \subseteq \IPaths[\ma]$ has a unique prefix $\ppath'$ with $\taOf{\ppath'} \in \Pathset$.
	We have
	\begin{displaymath}
	\eventually \goalStates = \bigcupdot_{\pathTa \in \Pathset} \cylinderset{\inducedPathsTa{\pathTa}}.%\ . 
	\end{displaymath}
	The claim follows with Lemma~\ref{lem:taSchedSamePathProb} since
	\begin{displaymath}
	\reachProbMa
	= \sum_{\pathTa \in \Pathset} \ProbModelScheduler{\ma}{\sched}(\inducedPathsTa{\pathTa})
	\overset{Lem.\,\ref{lem:taSchedSamePathProb}}{=}\sum_{\pathTa \in \Pathset} \ProbModelScheduler{\umdp}{\taOfSched}(\pathTa) = 
	\reachProbUMdp[].
	\end{displaymath}
\qed\end{proof}

\subsection{Proof of Theorem~\ref{thm:multiReachObj}}
\theoremMultiReachObj*
	
	\begin{proof}
		\label{thm:multiReachObj:proof}
		Let $\obj = (\reachObj[1], \dots, \reachObj[\numObjectives])$ be the considered list of objectives with threshold relations ${\rel }= \relTuple$.
		The following equivalences hold for any $\sched \in \GMSched[\ma]$ and $\point \in \RR^\numObjectives$.
		\begin{align*}
		\ma, \sched \models\obj \rel \point\ \, & \Longleftrightarrow \ \,\forall i \colon \ma, \sched \models \reachObj[i] \rel[i] \pointi{i} \\ %\text{ for all } 1 \le i \le \numObjectives\\
		& \Longleftrightarrow \ \,\forall i \colon \reachProbMa[i] \rel[i] \pointi{i}  \\ %\text{ for all } 1 \le i \le \numObjectives\\
		& \overset{\mathclap{Prop.\,\ref{pro:taSchedSameReachProb}}}{\Longleftrightarrow}\ \,  \forall i \colon \reachProbUMdp[i] \rel[i] \pointi{i}  \\ %\text{ for all } 1 \le i \le \numObjectives\\
		& \Longleftrightarrow \ \,\forall i \colon  \umdp, \taOfSched \models \reachObj[i] \rel[i] \pointi{i}  \\%\text{ for all } 1 \le i \le \numObjectives\\
		&   \Longleftrightarrow\ \,  \umdp, \taOfSched \models \obj \rel \point\ .
		\end{align*}
		Assume that $\achievabilityQ{\ma}{\obj \rel \point}$ holds, i.e., there is a $\sched \in \GMSched[\ma]$ such that $\ma, \sched \models\obj \rel \point$.
		It follows that $\umdp, \taOfSched \models \obj \rel \point$  which means that $\achievabilityQ{\umdp}{\obj \rel \point}$ holds as well.
		For the other direction assume $\achievabilityQ{\umdp}{\obj \rel \point}$, i.e., $\umdp, \sched \models \obj \rel \point$ for some time-abstract scheduler $\sched \in \TASched[]$.
		We have $\taOfSched=\sched $. It follows that $\umdp, \taOfSched \models \obj \rel \point$.
		Applying the equivalences above yields $\ma, \sched \models \obj \rel \point$ and thus $\achievabilityQ{\ma}{\obj \rel \point}$.
	\qed\end{proof}

%% file: B-proofsreward.tex
\section{Proofs for Expected Reward}
\label{app:proofsExpRew}

\subsection{Proof of Proposition~\ref{pro:taSchedSameExpReachRew}}

Let $n \ge 0$ and $\goalStates\subseteq \States$.
The set of time-abstract paths that end after $n$ steps or at the first visit of a state in $\goalStates$ is denoted by
\begin{align*}
\PathsNOrTarget =  \{\pathstepTa{0} \dots \pathTransTa{m-1} \state_m \in \FPaths[\umdp] \mid\  & (m=n \text{ or } \state_m \in \goalStates) \text{ and } \\
&\state_i \notin \goalStates \text{ for all } 0\le i < m\}.%\ .
\end{align*}
For $\ma$ under $\sched \in \GMSched[\ma]$ and $\umdp$ under $\taOfSched \in \TASched[]$, we define the expected reward collected along the paths of $\PathsNOrTarget$ as
\begin{align*}
\eRModelScheduler{\ma}{\sched}(\rewFct[], \PathsNOrTarget) &= \sum_{\pathTa \in \PathsNOrTarget} \int_{\ppath \in \inducedPathsTa{\pathTa}} \rewPath[\ma]{\rewFct}{\ppath} \diff \ProbModelScheduler{\ma}{\sched}(\ppath)
\text{ and }\\
\eRModelScheduler{\umdp}{\taOfSched}(\rewUMdpFct[], \PathsNOrTarget) &= \sum_{\pathTa \in \PathsNOrTarget} \rewPath[\umdp]{\rewUMdpFct}{\pathTa} \cdot  \ProbModelScheduler{\umdp}{\taOfSched}(\pathTa),%\ ,
\end{align*}
respectively.
Intuitively, $\eRModelScheduler{\ma}{\sched}(\rewFct[], \PathsNOrTarget)$ corresponds to $\expReachRewMa[]$ assuming that no more reward is collected after the $n$-th transition.
It follows that the value $\eRModelScheduler{\ma}{\sched}(\rewFct[], \PathsNOrTarget)$ approaches $\expReachRewMa[]$ for large $n$.
Similarly, $\eRModelScheduler{\umdp}{\taOfSched}(\rewUMdpFct[], \PathsNOrTarget)$ approaches $\expReachRewUMdp$ for large $n$.
This observation is formalized by the following lemma.
\begin{restatable}{lemma}{lemmaPathsNOrTargetApproachesReachRew}
	\label{lem:ma:rew:PathsNOrTargetApproachesReachRew}
	For MA $\maDef$ with $\goalStates \subseteq \States$,  $\sched \in \GMSched$, and reward function $\rewFct[]$ it holds that
	\begin{align*}
	\lim_{n\to\infty} \eRModelScheduler{\ma}{\sched}(\rewFct[], \PathsNOrTarget)  = \expReachRewMa[].
	\end{align*}
	Furthermore, any reward function $\rewUMdpFct[]$ for $\umdp$ satisfies
	\begin{align*}
	\lim_{n\to\infty} \eRModelScheduler{\umdp}{\taOfSched}(\rewUMdpFct[], \PathsNOrTarget) =  \expReachRewUMdp.
	\end{align*}
\end{restatable}
\begin{proof}
	We show the first claim. The second claim follows analogously.
	For each $n \ge 0$, consider the function $f_n \colon \IPaths[\ma] \to \RRnn$ given by
	\begin{displaymath}
	%		f_n(\pathIseqTimed) = 
	%		\begin{cases}
	%			\rewPath[\ma]{\rewFct}{\pathstepTimed{0} \dots \pathstepTimedEnd{m-1}{m}} & \text{if } \state_m \in \goalStates \text{ and } 
	%		\end{cases}
	f_n(\ppath) = 
	\begin{cases}
	\rewPath[\ma]{\rewFct}{\pref{\ppath}{m}} & \text{if } m = \min\big\{i \in \{0, \dots, n\} \mid \state_i \in \goalStates\big\} \\ 
	\rewPath[\ma]{\rewFct}{\pref{\ppath}{n}} & \text{if } \state_i \notin \goalStates \text{ for all }  i \le n
	\end{cases}
	\end{displaymath}
	for every path $\ppath = \pathIseqTimed \in \IPaths[\ma]$.
	Intuitively, $f_n(\ppath)$ is the reward  collected on $\ppath$ within the first $n$ steps and only  up to the first visit of $\goalStates$.
	This allows us to express the expected reward collected along the paths of $\PathsNOrTarget[n]$ as
	\begin{align*}
%	\label{eq:lem:ma:rew:PathsNOrTargetApproachesRew:overInfPaths}
	\eRModelScheduler{\ma}{\sched}( \PathsNOrTarget[n]) 
	=  \sum_{\pathTa \in \PathsNOrTarget[n]} \int_{\ppath \in \inducedPathsTa{\pathTa}} \rewPath[\ma]{\rewFct}{\ppath} \diff \ProbModelScheduler{\ma}{\sched}(\ppath)  
	=  \int_{\ppath \in \IPaths[\ma]} f_n(\ppath) \diff \ProbModelScheduler{\ma}{\sched}(\ppath).
	\end{align*}
	It holds that   $\lim_{n \to \infty} f_n(\ppath) = \rewPathToTarget[\ma]{\rewFct[]}{\ppath}{\goalStates}$ which is  a direct consequence  from the definition of the  reward of $\ppath$ up to $\goalStates$ (cf. App.~\ref{App:ExpectedReward}).
	Furthermore, note  that the sequence of functions $f_0, f_1, \dots$ is non-decreasing, i.e., we have $f_n(\ppath) \le f_{n+1}(\ppath)$ for all $n\ge 0$ and $\ppath \in \IPaths[\ma]$.
	By applying the \emph{monotone convergence theorem} \cite{ash2000probability} we obtain
	\begin{align*}
	\lim_{n\to\infty} \eRModelScheduler{\ma}{\sched}( \PathsNOrTarget[n]) 
%	& \overset{\mathclap{\ref{eq:lem:ma:rew:PathsNOrTargetApproachesRew:overInfPaths}}}{=} \lim_{n\to\infty}   \int_{\ppath \in \IPaths[\ma]} f_n(\ppath) \diff \ProbModelScheduler{\ma}{\sched}(\ppath)   \\
	& =\lim_{n\to\infty}   \int_{\ppath \in \IPaths[\ma]} f_n(\ppath) \diff \ProbModelScheduler{\ma}{\sched}(\ppath)   \\
	&=  \int_{\ppath \in \IPaths[\ma]} \lim_{n \to \infty} f_n(\ppath) \diff \ProbModelScheduler{\ma}{\sched}(\ppath)  \\
	&=  \int_{\ppath \in \IPaths[\ma]} \rewPathToTarget[\ma]{\rewFct[]}{\ppath}{\goalStates} \diff \ProbModelScheduler{\ma}{\sched}(\ppath) 
	= \expReachRewMa.
	\end{align*}
\qed\end{proof}
The next step is to show that the expected reward collected along the paths of $\PathsNOrTarget[n]$ coincides for $\ma$ under $\sched$ and $\umdp$ under $\taOfSched$.
\begin{restatable}{lemma}{lemmaTaSchedSameExpRewOfPathsN}
	\label{lem:taSchedSameExpRewOfPathsN}
	Let $\rewFct[]$ be some reward function of $\ma$ and let $\rewUMdpFct[]$ be its counterpart for $\umdp$.
	Let $\maDef$ be an MA with $\goalStates \subseteq \States$ and  $\sched \in \GMSched$.
	For all $\goalStates \subseteq \States$ and $n\ge 0$ it holds that
	\begin{displaymath}
	\eRModelScheduler{\ma}{\sched}(\rewFct[], \PathsNOrTarget) = \eRModelScheduler{\umdp}{\taOfSched}(\rewUMdpFct[], \PathsNOrTarget) .
	\end{displaymath}
\end{restatable}
\begin{proof}
	The proof is by induction over the path length $n$.
	To simplify the notation, we often omit the reward functions $\rewFct[]$ and $\rewUMdpFct[]$ and write, e.g., $\rewPath[\umdp]{}{\ppath}$ instead of $\rewPath[\umdp]{\rewUMdpFct}{\ppath}$ or  $\eRModelScheduler{\ma}{\sched}( \PathsNOrTarget[n])$ instead of  $\eRModelScheduler{\ma}{\sched}(\rewFct[], \PathsNOrTarget[n])$.
	
	If $n = 0$, then $\PathsNOrTarget = \{ \sinit\}$. The claim holds since $\rewPath[\ma]{}{\sinit} = \rewPath[\umdp]{}{\sinit} = 0$.
	
	In the induction step, we assume that the lemma is true for some fixed $n \ge 0$.
	We split the term $\eRModelScheduler{\ma}{\sched}( \PathsNOrTarget[n+1])$ into the reward that is obtained by paths which reach $\goalStates$ within $n$ steps and the reward obtained by paths of length $n+1$.
	In a second step, we consider the sum of the reward collected within the first $n$ steps and the reward obtained in the  $(n+1)$-th step:
	\begin{align}
	& \mathrel{\phantom{=}} \eRModelScheduler{\ma}{\sched}( \PathsNOrTarget[n+1]) \nonumber \\
	%&= \sum_{\pathTa \in \PathsNOrTarget[n+1]} \int_{\ppath \in \inducedPathsTa{\pathTa}} \rewPath[\ma]{\rewFct}{\ppath} \diff \ProbModelScheduler{\ma}{\sched}(\ppath) \nonumber \\
	&= \sum_{\substack{\pathTa \in \PathsNOrTarget[n+1] \\ \length{\pathTa} \le n}}   \int_{\substack{\ppath \in \inducedPathsTa{\pathTa}}} \rewPath[\ma]{}{\ppath} \diff \ProbModelScheduler{\ma}{\sched}(\ppath)\nonumber \\
	&\newlineIndention + \sum_{\substack{\pathTa \in \PathsNOrTarget[n+1] \\ \length{\pathTa} = n+1}}  \int_{\substack{\ppath = \ppath' \pathTransTimed{} \state' \in \inducedPathsTa{\pathTa}\\ \last{\ppath'} = \state}} \rewPath[\ma]{}{\ppath'} +  \rewState[]{\state} \cdot \timeOfStamp[] + \rewAct[]{\state}{\actionOfStamp[]}\diff \ProbModelScheduler{\ma}{\sched}(\ppath) \nonumber\\
	&=  \sum_{\substack{\pathTa \in \PathsNOrTarget[n+1]}}  \int_{\substack{\ppath \in \inducedPathsTa{\pathTa}}} \rewPath[\ma]{}{\pref{\ppath}{n}} \diff \ProbModelScheduler{\ma}{\sched}(\ppath) \label{eq:taSchedSameExpRewOfPathsN:RewN}\\
	&\newlineIndention + \sum_{\substack{\pathTa \in \PathsNOrTarget[n+1] \\ \length{\pathTa} = n+1}}  \int_{\substack{\ppath = \ppath' \pathTransTimed{} \state' \in \inducedPathsTa{\pathTa}\\ \last{\ppath'} = \state}}   \rewState[]{\state} \cdot \timeOfStamp[] + \rewAct[]{\state}{\actionOfStamp[]}\diff \ProbModelScheduler{\ma}{\sched}(\ppath), \label{eq:taSchedSameExpRewOfPathsN:RewNplus1}
	\end{align}
	where we define  $\pref{\ppath}{n}$ for paths with $\length{\ppath} \le n$ such that $\pref{\ppath}{n} = \ppath$.
	The two terms \eqref{eq:taSchedSameExpRewOfPathsN:RewN} and \eqref{eq:taSchedSameExpRewOfPathsN:RewNplus1} are treated separately.
	
	\paragraph{\underline{Term \eqref{eq:taSchedSameExpRewOfPathsN:RewN}}:}
	Let $\PathsAtMostNEventuallyTarget =  \{\pathTa \in \PathsNOrTarget[n+1] \mid  \length{\pathTa} \le n \}$ be the paths in $\PathsNOrTarget[n+1]$ of length at most $n$.
	We have $\PathsAtMostNEventuallyTarget \subseteq \PathsNOrTarget$ and every path in $\PathsAtMostNEventuallyTarget$ visits a state in $\goalStates$. 
	Correspondingly, $\PathsEqualNNeverTarget = \PathsNOrTarget \setminus \PathsAtMostNEventuallyTarget$ is the set of time-abstract paths of length $n$ that do not visit a state in $\goalStates$.
	Hence, the paths in $\PathsNOrTarget[n+1]$ with length $n+1$ have a prefix in $\PathsEqualNNeverTarget$.
	% every extension of a path $\pathTa \in \PathsEqualNNeverTarget$ to a path of length $n+1$ is  contained in $\PathsNOrTarget[n+1]$.
	The set $\PathsNOrTarget[n+1]$ is partitioned such that
	\begin{align*}
	\PathsNOrTarget[n+1]
	& = \PathsAtMostNEventuallyTarget \cupdot \left\{\pathTa \in \PathsNOrTarget[n+1]  \mid \length{\pathTa} = n+1 \right\}\\
	&= \PathsAtMostNEventuallyTarget \cupdot \{\pathTa = \pathTa' \pathTransTa{} \state' \in \FPaths[\umdp]  \mid \pathTa' \in \PathsEqualNNeverTarget \}.
	%= \PathsAtMostNEventuallyTarget \cupdot \bigcupdot_{\pathTa \in \PathsEqualNNeverTarget} \left\{ \pathTa \pathTransTa{}  \state' \in \FPaths[\umdp] \right\}.%\ .
	%= \PathsAtMostNEventuallyTarget \cupdot \bigcupdot_{\pathTa \in \PathsEqualNNeverTarget}  \bigcupdot_{\pathTa \pathTransTa{} \state' \in \FPaths[\umdp]}   \left\{ \pathTa \pathTransTa{}  \state'  \right\}\ .
	\end{align*}
	The reward obtained within the first $n$ steps is independent of the $(n+1)$-th transition. 
	To show this formally, we fix a path $\pathTa' \in \PathsEqualNNeverTarget$ with $\last{\pathTa'} = \state$ and derive
	\begin{align}
	% &\mathrel{\phantom{=}} \sum_{\substack{\pathTa  \in \FPaths[\umdp] \\ \pathTa = \pathTa' \pathTransTa{} \state' }}  \int_{\substack{\ppath \in \inducedPathsTa{\pathTa}}} \rewPath[\ma]{}{\pref{\ppath}{n}} \diff \ProbModelScheduler{\ma}{\sched}(\ppath) \nonumber\\
	&\mathrel{\phantom{=}} \sum_{\substack{\pathTa' \pathTransTa{} \state'  \in \FPaths[\umdp]  }}\  \int_{\substack{\ppath \in \inducedPathsTa{\pathTa' \pathTransTa{} \state'}}} \rewPath[\ma]{}{\pref{\ppath}{n}} \diff \ProbModelScheduler{\ma}{\sched}(\ppath) \nonumber\\
	&=
	\begin{dcases}
	\int_{\substack{\ppath' \in \inducedPathsTa{\pathTa'}}}   \rewPath[\ma]{}{\ppath'} \cdot \sum_{\substack{(\act, \state') \in \Actions \times \States}} \schedEval{\sched}{\ppath'}{\act} \cdot \probP(\state, \act, \state') \diff \ProbModelScheduler{\ma}{\sched}(\ppath')  & \text{if } \state \in \PS \\
	\int_{\substack{\ppath' \in \inducedPathsTa{\pathTa'}}}  \rewPath[\ma]{}{\ppath'}  \cdot \sum_{\substack{\state' \in \States}}  \probP(\state, \markovianAct, \state')  \diff \ProbModelScheduler{\ma}{\sched}(\ppath') & \text{if } \state \in \MS 
	%   \int_{\substack{\ppath' \in \inducedPathsTa{\pathTa'}}}  \rewPath[\ma]{}{\ppath'} \cdot  \int_{0}^{\infty}    \rateAtState{\state} \cdot e^{-\rateAtState{\state} \ttime} \cdot \sum_{\substack{\state' \in \States}}  \probP(\state, \markovianAct, \state') \diff\ttime \diff \ProbModelScheduler{\ma}{\sched}(\ppath') & \text{if } \state \in \MS 
	\end{dcases} \nonumber\\
	% &=
	% \begin{dcases}
	%  \int_{\substack{\ppath' \in \inducedPathsTa{\pathTa'}}}   \rewPath[\ma]{}{\ppath'} \cdot \Big(\sum_{\substack{\state' \in \States}} \probP(\state, \markovianAct, \state') \Big) \cdot \int_{0}^{\infty}  \rateAtState{\state} \cdot e^{-\rateAtState{\state} \ttime} \diff\ttime \diff \ProbModelScheduler{\ma}{\sched}(\ppath')  & \text{if } \state \in \MS\\
	%    \int_{\substack{\ppath' \in \inducedPathsTa{\pathTa'}}} \rewPath[\ma]{}{\ppath'} \cdot \Big(\sum_{\substack{\act \in \Actions}}  \schedEval{\sched}{\ppath'}{\act}\cdot \sum_{\state' \in \States}  \probP(\state, \act, \state')  \Big) \diff \ProbModelScheduler{\ma}{\sched}(\ppath')  & \text{if } \state \in \PS \\
	% \end{dcases} \\
	&=  \int_{\substack{\ppath' \in \inducedPathsTa{\pathTa'}}}   \rewPath[\ma]{}{\ppath'} \diff \ProbModelScheduler{\ma}{\sched}(\ppath'). \label{eq:taSchedSameExpRewOfPathsN:RewNIndependentOfNplus1}
	\end{align}
	With the above-mentioned partition of the set $\PathsNOrTarget[n+1]$, it follows that the expected reward obtained within the first $n$ steps is given by
	\begin{align}
	& \mathrel{\phantom{=}} \sum_{\substack{\pathTa \in \PathsNOrTarget[n+1]}}  \int_{\substack{\ppath \in \inducedPathsTa{\pathTa}}} \rewPath[\ma]{}{\pref{\ppath}{n}} \diff \ProbModelScheduler{\ma}{\sched}(\ppath) \nonumber \\
	&= \sum_{\substack{\pathTa \in \PathsAtMostNEventuallyTarget}} \int_{\substack{\ppath \in \inducedPathsTa{\pathTa}}}   \rewPath[\ma]{}{\ppath} \diff \ProbModelScheduler{\ma}{\sched}(\ppath) \nonumber\\  
	%     & \newlineIndention + \sum_{\substack{\pathTa' \in \PathsEqualNNeverTarget}} \  \sum_{\substack{\pathTa  \in \FPaths[\umdp] \\ \pathTa = \pathTa' \pathTransTa{} \state' }} \  \int_{\substack{\ppath \in \inducedPathsTa{\pathTa}}} \rewPath[\ma]{}{\pref{\ppath}{n}} \diff \ProbModelScheduler{\ma}{\sched}(\ppath) \nonumber \\
	& \newlineIndention + \sum_{\substack{\pathTa' \in \PathsEqualNNeverTarget}} \  \sum_{\substack{\pathTa' \pathTransTa{} \state'  \in \FPaths[\umdp]   }} \  \int_{\substack{\ppath \in \inducedPathsTa{\pathTa' \pathTransTa{} \state'}}} \rewPath[\ma]{}{\pref{\ppath}{n}} \diff \ProbModelScheduler{\ma}{\sched}(\ppath) \nonumber \\
	&\overset{\mathclap{\eqref{eq:taSchedSameExpRewOfPathsN:RewNIndependentOfNplus1}}}{=} \sum_{\substack{\pathTa \in \PathsAtMostNEventuallyTarget}} \int_{\substack{\ppath \in \inducedPathsTa{\pathTa}}}   \rewPath[\ma]{}{\ppath} \diff \ProbModelScheduler{\ma}{\sched}(\ppath) +  \sum_{\substack{\pathTa \in \PathsEqualNNeverTarget}} \int_{\substack{\ppath \in \inducedPathsTa{\pathTa}}}  \rewPath[\ma]{}{\ppath} \diff \ProbModelScheduler{\ma}{\sched}(\ppath)  \nonumber \\
	% &=  \sum_{\substack{\pathTa \in \PathsNOrTarget[n]}} \int_{\substack{\ppath \in \inducedPathsTa{\pathTa}}}  \rewPath[\ma]{}{\ppath} \diff \ProbModelScheduler{\ma}{\sched}(\ppath)  \\
	&=  \eRModelScheduler{\ma}{\sched}(\PathsNOrTarget)   \nonumber  \displaybreak[0] \\
	& \eqInductionHypothesis  \eRModelScheduler{\umdp}{\taOfSched}(\PathsNOrTarget) \nonumber \\
	% &=  \sum_{\pathTa \in \PathsNOrTarget} \rewPath[\umdp]{}{\pathTa} \cdot \ProbModelScheduler{\umdp}{\taOfSched}(\pathTa)\\
	&= \sum_{\pathTa \in \PathsAtMostNEventuallyTarget} \rewPath[\umdp]{}{\pathTa} \cdot \ProbModelScheduler{\umdp}{\taOfSched}(\pathTa) + \sum_{\pathTa \in \PathsEqualNNeverTarget}    \rewPath[\umdp]{}{\pathTa} \cdot \ProbModelScheduler{\umdp}{\taOfSched}(\pathTa)\nonumber  \\ 
	&= \sum_{\pathTa \in \PathsAtMostNEventuallyTarget} \rewPath[\umdp]{}{\pathTa} \cdot \ProbModelScheduler{\umdp}{\taOfSched}(\pathTa) \nonumber\\
	& \newlineIndention + \sum_{\pathTa' \in \PathsEqualNNeverTarget}   \quad \sum_{\substack{\pathTa \in \FPaths[\umdp] \\ \pathTa = \pathTa' \pathTransTa{} \state' }}  \rewPath[\umdp]{}{\pref{\pathTa}{n}} \cdot \ProbModelScheduler{\umdp}{\taOfSched}(\pathTa) \nonumber\\ 
	&= \sum_{\pathTa \in \PathsNOrTarget[n+1]} \rewPath[\umdp]{}{\pref{\pathTa}{n}} \cdot \ProbModelScheduler{\umdp}{\taOfSched}(\pathTa).%\ .
	\label{eq:taSchedSameExpRewOfPathsN:umdpRewN}
	\end{align}
	
	\paragraph{\underline{Term \eqref{eq:taSchedSameExpRewOfPathsN:RewNplus1}}:}
	For the  expected reward obtained in step $n+1$,  consider a path $\pathTa = \pathTa' \pathTransTa{} \state' \in \PathsNOrTarget[n+1]$ such that $  \length{\pathTa'} = n$ and $\last{\pathTa'} = \state$.
	\begin{itemize}
		\item If $\state \in \MS$, we have $\pathTa = \pathTa' \pathTransUniv{\markovianAct} \state'$. It follows that
		\begin{align}
		& \mathrel{\phantom{=}} \int_{\ppath = \ppath'\pathTransUniv{\ttime} \state' \in \inducedPathsTa{\pathTa}} \rewState{\state}  \cdot \ttime + \rewAct{\state}{\markovianAct}   \diff \ProbModelScheduler{\ma}{\sched}(\ppath) \nonumber\\
		&=
		\int_{\substack{\ppath = \ppath'\pathTransUniv{\ttime} \state' \in \inducedPathsTa{\pathTa}}}  \rewState{\state}  \cdot \ttime \diff \ProbModelScheduler{\ma}{\sched}(\ppath)
		+ 
		\int_{\ppath \in \inducedPathsTa{\pathTa}} \rewAct{\state}{\markovianAct} \diff \ProbModelScheduler{\ma}{\sched}(\ppath) 
		\nonumber \\
		&=
		\rewState{\state} \cdot \int_{\ppath' \in \inducedPathsTa{\pathTa'}} \int_{0}^{\infty} \ttime \cdot \rateAtState{\state} \cdot  e^{-\rateAtState{\state} \ttime} \cdot \probP(\state, \markovianAct, \state') \diff \ttime  \diff \ProbModelScheduler{\ma}{\sched}(\ppath') \nonumber\\
		& \newlineIndention +
		\rewAct{\state}{\markovianAct} \cdot \ProbModelScheduler{\ma}{\sched}(\inducedPathsTa{\pathTa}) 
		\nonumber\\
		&=
		\frac{\rewState{\state}}{\rateAtState{\state}} \cdot \ProbModelScheduler{\ma}{\sched}(\inducedPathsTa{\pathTa})
		+
		\rewAct{\state}{\markovianAct} \cdot \ProbModelScheduler{\ma}{\sched}(\inducedPathsTa{\pathTa}) 
		\nonumber \\
		&=\mathrel{} \rewUMdp{\state}{\markovianAct} \cdot \ProbModelScheduler{\ma}{\sched}(\inducedPathsTa{\pathTa})
		\overset{Lem.\,\ref{lem:taSchedSamePathProb}}{=}  \rewUMdp{\state}{\markovianAct} \cdot \ProbModelScheduler{\umdp}{\taOfSched}(\pathTa).%\ .
		\label{eq:taSchedSameExpRewOfPathsN:umdpRewNplus1}
		\end{align}
		\item If $\state \in \PS$, then $\int_{\ppath = \ppath'\pathTransUniv{\act} \state' \in \inducedPathsTa{\pathTa}} \rewAct{\state}{\act}   \diff \ProbModelScheduler{\ma}{\sched}(\ppath) = \rewUMdp{\state}{\act} \cdot \ProbModelScheduler{\umdp}{\taOfSched}(\pathTa)$ follows similarly.
	\end{itemize}

	Combining the two results yields
	\begin{align*}
	\eRModelScheduler{\ma}{\sched}(\PathsNOrTarget[n+1]) 
	\ \ &\overset{\mathclap{\ref{eq:taSchedSameExpRewOfPathsN:RewN},\, \ref{eq:taSchedSameExpRewOfPathsN:RewNplus1}}}{=} \ \sum_{\substack{\pathTa \in \PathsNOrTarget[n+1]}}  \int_{\substack{\ppath \in \inducedPathsTa{\pathTa}}} \rewPath[\ma]{}{\pref{\ppath}{n}} \diff \ProbModelScheduler{\ma}{\sched}(\ppath) \\
	&\newlineIndention + \sum_{\substack{\pathTa \in \PathsNOrTarget[n+1] \\ \length{\pathTa} = n+1}}  \int_{\substack{\ppath = \ppath' \pathTransTimed{} \state' \in \inducedPathsTa{\pathTa}\\ \last{\ppath'} = \state}}   \rewState[]{\state} \cdot \timeOfStamp[] + \rewAct[]{\state}{\actionOfStamp[]}\diff \ProbModelScheduler{\ma}{\sched}(\ppath)  \displaybreak[0]\\
	&\overset{\mathclap{\ref{eq:taSchedSameExpRewOfPathsN:umdpRewN},\, \ref{eq:taSchedSameExpRewOfPathsN:umdpRewNplus1}}}{=}\  \sum_{\pathTa \in \PathsNOrTarget[n+1]} \rewPath[\umdp]{}{\pref{\pathTa}{n}} \cdot \ProbModelScheduler{\umdp}{\taOfSched}(\pathTa)\\
	& \newlineIndention  + \sum_{\substack{\pathTa = \pathTa' \pathTransTa{} \state' \in \PathsNOrTarget[n+1] \\ \length{\pathTa} = n+1}} \rewUMdp{\last{\pathTa'}}{\act} \cdot \ProbModelScheduler{\umdp}{\taOfSched}(\pathTa)\\
	&=\  \sum_{\substack{\pathTa \in \PathsNOrTarget[n+1]}} \rewPath[\umdp]{}{\pathTa} \cdot \ProbModelScheduler{\umdp}{\taOfSched}(\pathTa)
	= \eRModelScheduler{\umdp}{\taOfSched}(\PathsNOrTarget[n+1]).%\ .
	\end{align*}
\qed\end{proof}

We now show Proposition~\ref{pro:taSchedSameExpReachRew}.

\propositionTaSchedSameExpReachRew*

\begin{proof}
	The proposition is a direct consequence of Lemma~\ref{lem:ma:rew:PathsNOrTargetApproachesReachRew} and Lemma~\ref{lem:taSchedSameExpRewOfPathsN} as
	\begin{align*}
	\expReachRewMa[]  
	& = \lim_{n\to\infty} \eRModelScheduler{\ma}{\sched}(\rewFct[], \PathsNOrTarget)\\
	%   \overset{Lem.\,\ref{lem:taSchedSameExpRewOfPathsN}}{=} 
	& = \lim_{n\to\infty} \eRModelScheduler{\umdp}{\taOfSched}(\rewUMdpFct[], \PathsNOrTarget) 
	=  \expReachRewUMdp.%\ .
	\end{align*}
\qed\end{proof}

\subsection{Proof of Theorem~\ref{thm:multiExpReachRewObj}}
\theoremMultiExpReachRewObj*

\begin{proof}
	Let $\obj = \objTuple$ be the considered list of untimed reachability and expected reward objectives with threshold relations ${\rel }= \relTuple$.
	The following equivalences hold for any $\sched \in \GMSched[\ma]$ and $\point \in \RR^\numObjectives$.
	\begin{align*}
	\ma, \sched \models\obj \rel \point & \Longleftrightarrow \forall i \colon \ma, \sched \models \obj[i] \rel[i] \pointi{i} \\ %\text{ for all } 1 \le i \le \numObjectives\\
	%	& \Longleftrightarrow \forall i \colon \expReachRewMa[i] \rel[i] \pointi{i}  \\ %\text{ for all } 1 \le i \le \numObjectives\\
	%	& \overset{\equalityMarker}{\Longleftrightarrow} \forall i \colon \expReachRewUMdp[i] \rel[i] \pointi{i}  \\ %\text{ for all } 1 \le i \le \numObjectives\\
	& \overset{\equalityMarker}{\Longleftrightarrow} \forall i \colon  \umdp, \taOfSched \models \obj[i] \rel[i] \pointi{i}   %\text{ for all } 1 \le i \le \numObjectives\\
	\Longleftrightarrow  \umdp, \taOfSched \models \obj \rel \point\ ,
	\end{align*}
	where for the equivalence marked with $\equalityMarker$ we consider two cases:
	If $\obj[i]$ is of the form $\reachObj[]$, Proposition~\ref{pro:taSchedSameReachProb} yields
	\begin{align*}
	\ma, \sched \models \obj[i] \rel[i] \pointi{i} 
	& \Longleftrightarrow \reachProbMa[] \rel[i] \pointi{i}\\
	&	\Longleftrightarrow \reachProbUMdp[] \rel[i] \pointi{i}
	\Longleftrightarrow	\umdp, \taOfSched \models \obj[i] \rel[i] \pointi{i} \ .
	\end{align*}
	Otherwise, $\obj[i]$ is of the form $\expReachRewObj[]$ and with Proposition~\ref{pro:taSchedSameExpReachRew} it follows that
	\begin{align*}
	\ma, \sched \models \obj[i] \rel[i] \pointi{i} 
	& \Longleftrightarrow \expReachRew{\ma}{\sched}{\rewFct[\rewFctIndex]} \rel[i] \pointi{i}\\
	&	\Longleftrightarrow  \expReachRew{\umdp}{\taOfSched}{\rewUMdpFct[\rewFctIndex]}\rel[i] \pointi{i}
	\Longleftrightarrow	\umdp, \taOfSched \models \obj[i] \rel[i] \pointi{i} \ .
	\end{align*}
	The remaining steps of the  proof are completely analogous to the proof of Theorem~\ref{thm:multiReachObj} conducted on page~\pageref{thm:multiReachObj:proof}.
\qed\end{proof}

%% file: C-proofsbounded.tex
\section{Proofs for Timed Reachability}
\label{app:proofsBounded}
\subsection{Proof of Proposition~\ref{pro:bounded:dsBoundedProbEqual}}
Let $\maDef$ be an MA and let $\dma$ be the digitization of $\ma$ with respect to some $\digConstant \in \RRgz$.
We consider the \emph{infinite} paths of $\ma$ that are represented by a \emph{finite} digital path.
\begin{definition}[Induced cylinder of a digital path]
	\label{def:ma:cylOfDigPath}
	Given a digital path $\pathDi \in \FPaths[\dma]$ of MA $\ma$, the \emph{induced cylinder of $\pathDi$} is given by
	\begin{displaymath}
	\cylDi{\pathDi} = \{\ppath \in \IPaths[\ma] \mid \pathDi \text{ is a prefix of } \diOfPath  \}.%\ .
	\end{displaymath}
\end{definition}
Recall the definition of the cylinder of a set of finite paths (cf.  App. \ref{App:ProbMeasure}).
If $\pathDi \in \FPaths[\dma]$ does not end with a self-loop at a Markovian state, then 
$\cylDi{\pathDi} = \cylinderset{\inducedPathsDi{\pathDi}}$ holds.

	\begin{example}
		\label{ex:ma:bounded:inducedPathsCyl}
		Let $\ma$ and $\dma$ be as in Fig.~\ref{fig:models}.
		We consider the path $\pathDi_1 = \state_0 \pathTransUniv{\markovianAct}  \state_0 \pathTransUniv{\markovianAct}  \state_0 \pathTransUniv{\markovianAct} \state_3 \pathTransUniv{\altact} \state_4 $ and digitization constant $\digConstant = 0.4$.		
		The set $\cylDi{\pathDi_1}$ contains each infinite path whose digitization has the prefix $\pathDi_1$, i.e.,
		\begin{displaymath}
		\cylDi{\pathDi_1} 
		= \{ \state_0 \pathTransUniv{\ttime} \state_3 \pathTransUniv{\altact} \state_4 \pathTransTimed{} \dots \in \IPaths[\ma] \mid 0.8 \le \ttime < 1.2\}.%\ .
		\end{displaymath}
		We observe that these are exactly the paths that have a prefix in $\inducedPathsDi{\pathDi_1}$.
		Put differently, we have  $\cylDi{\pathDi_1} = \cylinderset{\inducedPathsDi{\pathDi_1}}$.
		
		Next, consider the digital path $\pathDi_2 =  \state_0 \pathTransUniv{\markovianAct} \state_0 \pathTransUniv{\markovianAct} \state_0$.
		Note that there is no path $\ppath \in \FPaths[\ma]$ with $\diOfPath = \pathDi_2$, implying $\inducedPathsDi{\pathDi_2} = \emptyset$.
		Intuitively, $\pathDi_2$ depicts a sojourn time at $\last{\pathDi_2}$ but finite paths of MAs do not depict sojourn times at their last state.
		On the other hand, the induced cylinder of $\pathDi_2$ contains all paths that sojourn at least $2 \digConstant$ time units at $\state_0$, i.e.,
		\begin{displaymath}
		\cylDi{\pathDi_2} 
		= \{ \state_0 \pathTransUniv{\ttime} \state_1  \pathTransTimed{} \dots \in \IPaths[\ma] \mid \ttime \ge 0.8  \}.%\ .
		\end{displaymath}
	\end{example}

The schedulers $\sched$ and $\diOfSched$ induce the same probabilities for a given digital path.
This is formalized by the following lemma.
Note that  a similar statement for $\taOfSched$ and time-abstract paths was shown in Lemma~\ref{lem:taSchedSamePathProb}.
\begin{lemma}
	\label{lem:app:bounded:diSchedSamePathProb}
	Let $\ma$ be an MA with scheduler $\sched \in \GMSched$, digitization $\dma$, and digital path $\pathDi \in \FPaths[\dma]$.
	It holds that
	\begin{displaymath}
	\ProbModelScheduler{\ma}{\sched}( \cylDi{\pathDi} ) =
	\ProbModelScheduler{\dma}{\diOfSched}(\pathDi).%\ .
	\end{displaymath}
\end{lemma}
	\begin{proof}
		The proof is by induction over the length $n$ of $\pathDi$.
		Let $\maDef$ and $\dmaDef$.
		If $n=0$, then $\pathDi = \sinit$ and $\cylDi{\pathDi} = \IPaths[\ma]$.
		Hence, $\ProbModelScheduler{\ma}{\sched}(\cylDi{\sinit}) = 1 = \ProbModelScheduler{\dma}{\diOfSched}(\sinit)$.
		In the induction step it is assumed that the lemma holds for a fixed path $\pathDi \in \FPaths[\dma]$ with $\length{\pathDi} = n$ and $\last{\pathDi} = \state$.
		Consider a path $\pathDi \pathTransTa{} \state' \in \FPaths[\dma]$.
		We distinguish the following cases.
		
		\paragraph{\underline{Case $\state \in \PS$:}}
		%All paths in $\cylDi{\pathDi \pathTransTa{} \state'}$ have a prefix in $\inducedPathsDi{\pathDi \pathTransTa{} \state'}$ since $\pathDi\pathTransTa{} \state'$ ends with a probabilistic transition. 
		%Vice versa, any extension of a path in $\inducedPathsDi{\pathDi \pathTransTa{} \state'}$ to an infinite path  is contained in $\cylDi{\pathDi \pathTransTa{} \state'}$.
		%$\pathDi \pathTransTa{} \state'$ ends with a probabilistic transition and that this transition also occurs in any $\ppath \in \cylDi{\pathDi \pathTransTa{} \state'}$.
		%Hence, all paths $\ppath \in \cylDi{\pathDi}$ have a prefix $\ppath' = \dots \state \pathTransTa{} \state'$ whose digitization satisfies $\diOf{\ppath'} = \pathDi\pathTransTa{} \state'$ or, equivalently, $\ppath' \in \inducedPathsDi{\pathDi \pathTransTa{} \state'}$.
		%It follows that all paths $\ppath \in \cylDi{\pathDi}$ have a prefix $\ppath' \in \inducedPathsDi{\pathDi}$.
		It follows that $\cylDi{\pathDi \pathTransTa{} \state'} = \cylinderset{\inducedPathsDi{\pathDi \pathTransTa{} \state'}}$ since $\pathDi \pathTransTa{} \state'$ ends with a probabilistic transition.
		Hence,
		\begin{align*}
		\ProbModelScheduler{\ma}{\sched}(\cylDi{\pathDi \pathTransTa{} \state'})
		& =  \ProbModelScheduler{\ma}{\sched}(\inducedPathsDi{\pathDi \pathTransTa{} \state'})\\
		& = \int_{\ppath \in \inducedPathsDi{\pathDi}} \schedEval{\sched}{\ppath}{\act} \cdot \probP(\state, \act, \state') \diff\ProbModelScheduler{\ma}{\sched}(\ppath) \\
		& = \int_{\ppath \in \inducedPathsDi{\pathDi}} \schedEval{\sched}{\ppath}{\act} \cdot \probP(\state, \act, \state') \diff\ProbModelScheduler{\ma}{\sched}(\{\ppath\} \cap \inducedPathsDi{\pathDi}) \\
		& = \int_{\ppath \in \inducedPathsDi{\pathDi}} \schedEval{\sched}{\ppath}{\act} \cdot \probP(\state, \act, \state') \diff\big[\ProbModelScheduler{\ma}{\sched}(\ppath \mid \inducedPathsDi{\pathDi}) \cdot \ProbModelScheduler{\ma}{\sched}(\inducedPathsDi{\pathDi})\big] \\
		& = \ProbModelScheduler{\ma}{\sched}(\inducedPathsDi{\pathDi}) \cdot \probP(\state, \act, \state') \cdot \int_{\ppath \in \inducedPathsDi{\pathDi}} \schedEval{\sched}{\ppath}{\act} \diff\ProbModelScheduler{\ma}{\sched}(\ppath \mid \inducedPathsDi{\pathDi})  \\
		& = \ProbModelScheduler{\ma}{\sched}(\inducedPathsDi{\pathDi}) \cdot \probP(\state, \act, \state') \cdot \schedEval{\diOfSched}{\pathDi}{\act}  \\
		& \eqInductionHypothesis \ProbModelScheduler{\umdp}{\diOfSched}(\pathDi) \cdot \probP(\state, \act, \state') \cdot \schedEval{\diOfSched}{\pathDi}{\act}  \\
		& = \ProbModelScheduler{\umdp}{\diOfSched}(\pathDi \pathTransTa{} \state').%\ .
		\end{align*}
		
		\paragraph{\underline{Case $\state \in \MS$:}}
		As $\state \in \MS$ we have $\act = \markovianAct$ and it follows
		\begin{align}
		\ProbModelScheduler{\ma}{\sched}(\cylDi{\pathDi \pathTransUniv{\markovianAct} \state'}) 
		& = 	\ProbModelScheduler{\ma}{\sched}(\cylDi{\pathDi} \cap \cylDi{\pathDi \pathTransUniv{\markovianAct} \state'}) \nonumber \\
		& = \ProbModelScheduler{\ma}{\sched}(\cylDi{\pathDi}) \cdot \ProbModelScheduler{\ma}{\sched}(\cylDi{\pathDi \pathTransUniv{\markovianAct} \state'} \mid \cylDi{\pathDi}  ). \label{eq:app:lem:diSchedSamePathProb:MSSplit} %\ .
		\end{align}
		Assume that a path $\ppath \in \cylDi{\pathDi}$ has been observed, i.e., $\pref{\diOfPath}{m} = \pathDi$ holds for some $m\ge 0$.
		The term $\ProbModelScheduler{\ma}{\sched}(\cylDi{\pathDi \pathTransUniv{\markovianAct} \state'} \mid \cylDi{\pathDi}  )$ coincides with the probability that also $\pref{\diOfPath}{m+1} = \pathDi \pathTransUniv{\markovianAct} \state'$ holds.
		%Let us focus on the term $\ProbModelScheduler{\ma}{\sched}(\cylDi{\pathDi \pathTransUniv{\markovianAct} \state'} \mid \cylDi{\pathDi}  )$.
		%It is the probability that a path $\ppath \in \IPaths[\ma]$ with $\pathDi = \pref{\diOfPath}{n}$ for some $n\ge 0$ satisfies $\pathDi \pathTransUniv{\markovianAct} \state' = \pref{\diOfPath}{n+1}$.
		We have either
		\begin{itemize}
			\item $\state \neq \state'$ which means that the transition from $\state$ to $\state'$ has to be taken during a period of $\digConstant$ time units or
			\item $\state = \state'$ where we additionally have to consider the case that no transition is taken at $\state$ for $\digConstant$ time units.
		\end{itemize}
		%which happens with probability $(1-e^{-\rateAtState{\state} \digConstant})$.
		It follows that
		\begin{align}
		\ProbModelScheduler{\ma}{\sched}(\cylDi{\pathDi \pathTransUniv{\markovianAct} \state'} \mid \cylDi{\pathDi}  )
		& = 
		\begin{cases}
		\probP(\state, \markovianAct, \state') (1-e^{-\rateAtState{\state} \digConstant}) & \text{if } \state \neq \state' \\
		\probP(\state, \markovianAct, \state') (1-e^{-\rateAtState{\state} \digConstant}) + e^{-\rateAtState{\state} \digConstant} & \text{if } \state = \state' 
		\end{cases}\nonumber\\
		& = \probPdig(\state, \markovianAct, \state'). \label{eq:app:lem:diSchedSamePathProb:MSConditional}%\ .
		\end{align}
		We conclude that
		\begin{align*}
		\ProbModelScheduler{\ma}{\sched}(\cylDi{\pathDi \pathTransUniv{\markovianAct} \state'}) 
		\ \ & \overset{\mathclap{\ref{eq:app:lem:diSchedSamePathProb:MSSplit},\, \ref{eq:app:lem:diSchedSamePathProb:MSConditional}}}{=} \ \ \ProbModelScheduler{\ma}{\sched}(\cylDi{\pathDi}) \cdot  \probPdig(\state, \markovianAct, \state') \\
		& \eqInductionHypothesis\ \  \ProbModelScheduler{\dma}{\diOfSched}(\pathDi) \cdot  \probPdig(\state, \markovianAct, \state')
		= \ProbModelScheduler{\dma}{\diOfSched}(\pathDi \pathTransUniv{\markovianAct} \state').%\ .
		\end{align*}
	\qed\end{proof}

%
%We extend the notation $\inducedPathsDi{\cdot}$ to \emph{sets} of infinite paths of $\dma$.
%For $\Pathset \subseteq \IPaths{\dma}$ we define
%\begin{displaymath}
%\inducedPathsDi{\Pathset} = \{ \ppath \in \IPaths[\ma] \mid \diOfPath \in \Pathset \}.
%\end{displaymath}

We apply Lemma~\ref{lem:app:bounded:diSchedSamePathProb}  to show Proposition~\ref{pro:bounded:dsBoundedProbEqual}.
The idea of the proof is similar to the proof of Proposition~\ref{pro:taSchedSameReachProb} conducted on page~\pageref{proof:pro:taSchedSameReachProb}.
\propositionDsBoundedProbEqual*
\begin{proof}
  Consider the set $\Pathset_\goalStates^\altInterval \subseteq \FPaths[\dma]$ of paths that (i) visit $\goalStates$ within $\altInterval$ digitization steps and (ii) do not have a proper prefix that satisfies (i).
%  \begin{align*}
%  \Pathset_\goalStates^\altInterval = \{ & \pathDi = \pathFseqTa \in \FPaths[\dma] \mid   \numOfDS{\pathDi} \in \altInterval, \state_n \in \goalStates \text{, and}  \\ &\ \forall i< n \colon \state_i \in \goalStates \text{ implies } \numOfDS{\pref{\pathDi}{i}} \notin \altInterval \}.%\ .
%  \end{align*}
  Every path in $\eventuallyds[\altInterval] \goalStates$ has a unique prefix in $\Pathset_\goalStates^\altInterval$, yielding
  \begin{align*}
  \eventuallyds[\altInterval] \goalStates &= \bigcupdot_{\pathDi \in \Pathset_\goalStates^\altInterval} \cylinderset{\{\pathDi\}}
  \end{align*}
	For the corresponding paths of $\ma$ we obtain
	\begin{align*}
	\inducedPathsDi{\eventuallyds[\altInterval] \goalStates}
  	& = \{\ppath \in \IPaths[\ma] \mid \diOfPath \in \eventuallyds[\altInterval] \goalStates \} \\
	& = \{ \ppath \in \IPaths[\ma] \mid \diOfPath \text{ has a unique prefix in } \Pathset_\goalStates^\altInterval \} \\
	& = \bigcupdot_{\pathDi \in \Pathset_\goalStates^\altInterval} \cylDi{\pathDi}\ .  
	\end{align*}
  The proposition follows with Lemma~\ref{lem:app:bounded:diSchedSamePathProb} since
  \begin{align*}
  \ProbModelScheduler{\dma}{\diOfSched}(\eventuallyds[\altInterval] \goalStates)
  =  \sum_{\pathDi \in \Pathset_\goalStates^\altInterval} \ProbModelScheduler{\dma}{\diOfSched}(\pathDi)
  \overset{Lem.\,\ref{lem:app:bounded:diSchedSamePathProb}}{=}  \sum_{\pathDi \in \Pathset_\goalStates^\altInterval} \ProbModelScheduler{\ma}{\sched}( \cylDi{\pathDi})
   = 	\ProbModelScheduler{\ma}{\sched}(\inducedPathsDi{\eventuallyds[\altInterval] \goalStates}).
  \end{align*}
\qed\end{proof}

\subsection{Proof of Proposition~\ref{prop:bounded:approxBoundedProb}}
\label{app:proofsbounded:thmApprxBoundedProb}

	The notation $\numOfDS{\pathDi}$ for paths $\pathDi$ of $\dma$ is also applied to paths of $\ma$, where $\numOfDS{\ppath} = \numOfDS{\diOfPath}$ for any $\ppath \in \FPaths[\ma]$.
	Intuitively, one digitization step represents the elapse of at most $\digConstant$ time units.
	Consequently, the duration of a path with $\stepBound \in \NN$ digitization steps is at most $\stepBound \digConstant$.
%	This is formalized by the following lemma.
	\begin{lemma}
		\label{lem:app:digStepsBoundDuration}
		For a path $\ppath \in \FPaths[\ma]$  and digitization constant $\digConstant$ it holds that
		\begin{displaymath}
		\timeOfPath{\ppath} 
		\le \numOfDS{\ppath} \cdot \digConstant
		%		= \numOfDS{\diOfPath} \cdot \digConstant
		\ .
		\end{displaymath}
	\end{lemma}
	\begin{proof}
		Let $\ppath = \pathFseqTimed$ and let $\waitingSteps[i] = \max \{ \waitingSteps[] \in \NN \mid \waitingSteps[]  \digConstant \le \timeOfStamp[i]  \}$ for each $i  \in \nobreak \{0, \dots, n-1\}$ (as in Definition \ref{def:digitizationOfPath}).
		The number $\numOfDS{\ppath}$ is given by $\sum_{ 0 \le i < n,\, \state_i \in \MS} (\waitingSteps[i] + 1)$. 
		With $\timeOfStamp[i] \le (\waitingSteps[i] + 1) \digConstant$ it follows that
		\begin{displaymath}
		\timeOfPath{\ppath}
		= \sum_{\substack{0 \le i < n\\ \state_i \in \MS}} \timeOfStamp[i]
		\le \sum_{\substack{0 \le i < n\\ \state_i \in \MS}} ( \waitingSteps[i] + 1)  \digConstant = \numOfDS{\ppath} \cdot \digConstant\ .
		\end{displaymath}
	\qed\end{proof}

For a  path $\ppath$ and $\ttime \in \RRnn$, the prefix of $\ppath$ up to time point $\ttime$ is given by $\prefTime{\ppath}{\ttime} = \pref{\ppath}{\max \{ n \mid \timeOfPath{\pref{\ppath}{n}} \le \ttime \}}$.
For the proof of Proposition~\ref{prop:bounded:approxBoundedProb}, we focus on the probability that (under a given scheduler $\sched$) the digitization approach yields an inaccurate estimate of the actual time.
This is the probability that more than $\stepBound \in \NN$ digitization steps have been performed within $\stepBound \digConstant$ time units.
We denote this value by $\ProbModelScheduler{\ma}{\sched}(\dsBoundedPaths{\stepBound \digConstant}{>}{k})$.

	\begin{definition}[Digitization step bounded paths]
		\label{def:app:PathsBoundedSteps}
		Assume an MA $\ma$ and a digitization constant $\digConstant \in \RRgz$.
		For some $\ttime \in \RRnn$, $\stepBound \in \NN$, and ${\rel} \in \{<,\le, >, \ge\}$ the set of paths whose prefix up to time point $\ttime$ has $\rel j$ digitization steps is defined as
		\begin{displaymath}
		\dsBoundedPaths{\ttime}{\rel}{ \stepBound} = \{\ppath \in \IPaths[\ma] \mid \numOfDS{\prefTime{\ppath}{\ttime}} \rel \stepBound  \}.%\ .
		\end{displaymath}
	\end{definition}

	\begin{example}
		\label{ex:app:bounded:pathsBoundedSteps}
		Let $\ma$ be the MA given in Fig.~\ref{fig:app:bounded:upperBound:ma}.
		We consider the set $\dsBoundedPaths{5 \digConstant}{\le}{ 5}$.
		The digitization constant $\digConstant$ remains unspecified in this example.
		Fig.~\ref{fig:app:bounded:upperBound:paths} illustrates paths $\ppath_1$, $\ppath_2$, and $\ppath_3$ of $\ma$.
		We depict sojourn times by arrow length.
		For instance, the path $\ppath_1$ corresponds to $\state_0 \pathTransUniv{2.5 \digConstant} \state_0 \pathTransUniv{1.8 \digConstant} \state_1 \pathTransUniv{1.7 \digConstant} \dots \in \IPaths[\ma]$.
		Digitization steps that are ``earned'' by sojourning at some state for a multiple of $\digConstant$ time units are indicated by black dots.
		Transitions of $\ppath_i$ (where $i \in \{1,2,3\}$) that do not belong to $\prefTime{\ppath_i}{5 \digConstant}$ are depicted in gray.
		We obtain
		\begin{align*}
		\numOfDS{\prefTime{\ppath_1}{5\digConstant}} = 5  \quad & \implies \quad  \ppath_1 \in  \dsBoundedPaths{5 \digConstant}{\le}{ 5} \\
		\numOfDS{\prefTime{\ppath_2}{5\digConstant}} = 4 \quad & \implies \quad \ppath_2 \in  \dsBoundedPaths{5 \digConstant}{\le}{ 5} \\
		\numOfDS{\prefTime{\ppath_3}{5\digConstant}} = 7 \quad & \implies \quad  \ppath_3 \notin  \dsBoundedPaths{5 \digConstant}{\le}{ 5}\ .
		\end{align*}
		Note that only the digitization steps of the prefix up to time point $5 \digConstant$ are considered.
		For example, the step of $\ppath_2$ at time point $4.5 \digConstant$ is not considered since the corresponding transition is not part of $\prefTime{\ppath_2}{5\digConstant}$.
		However, we have $\numOfDS{\prefTime{\ppath_2}{5.5 \digConstant}} = 6$, implying $\ppath_2 \notin \dsBoundedPaths[]{5.5 \digConstant}{\le}{5}$.
		
		All considered paths reach $\goalStates = \{\state_1 \}$ within $5 \digConstant$ time units but $\ppath_3 \in \dsBoundedPaths{5 \digConstant}{>}{5}$ requires more than $5$ digitization steps.
	\end{example}
	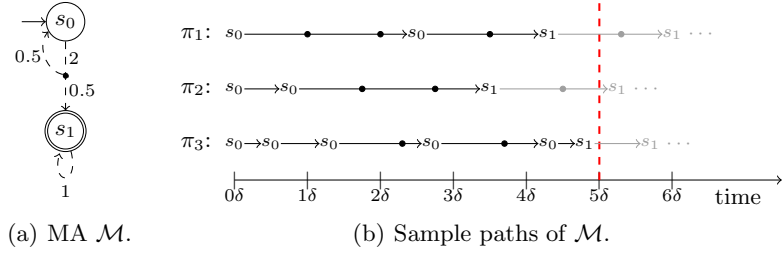
\begin{figure}[t]
		\centering
		\subfigure[MA $\ma$.]{
			\scalebox{\picscale}{
				\input{pics/ma_bounded_upperBound_ma}
			}
			\label{fig:app:bounded:upperBound:ma}
		}
		\subfigure[Sample paths of $\ma$.]{
			\scalebox{\picscale}{
				\input{pics/ma_bounded_upperBound_paths}
			}
			\label{fig:app:bounded:upperBound:paths}
		}
		\caption{MA $\ma$ and illustration of paths of $\ma$ (cf. Example~\ref{ex:app:bounded:pathsBoundedSteps}).}
		\label{fig:app:bounded:upperBound}
	\end{figure}

\noindent The following lemma gives an upper bound for the probability $\ProbModelScheduler{\ma}{\sched}(\dsBoundedPaths{\stepBound \digConstant}{>}{\stepBound})$.
	\begin{lemma}
		\label{lem:app:bounded:upperBoundForDSBoundedPaths}
		Let $\ma$ be an MA with $\sched \in \GMSched$ and maximum rate $\rateMax = \max \{\rateAtState{\state} \mid \state \in \MS \}$.
		Further, let $\digConstant \in \RRgz$  and $\stepBound \in \NN$.
		It holds that
		\begin{displaymath}
		\ProbModelScheduler{\ma}{\sched}(\dsBoundedPaths{\stepBound \digConstant}{>}{\stepBound})
		\le 1- (1+ \rateMax \digConstant)^{\stepBound} \cdot e^{- \rateMax \digConstant \stepBound}
		\end{displaymath}
	\end{lemma}
For the proof of Lemma~\ref{lem:app:bounded:upperBoundForDSBoundedPaths} we employ the following auxiliary lemma.
\begin{lemma}
	\label{lem:app:LowerBoundForPathsWithBoundedStepsPlusTime}
	Let $\ma$ be an MA with $\sched \in \GMSched[]$ and maximum rate $\rateMax = \max \{\rateAtState{\state} \mid \state \in \MS \}$.
	For each $\digConstant \in \RRgz$, $\stepBound \in \NN$, and $\ttime \in \RRnn$ it holds that
	\begin{displaymath}
	\ProbModelScheduler{\ma}{\sched}(\dsBoundedPaths{\stepBound \digConstant + \ttime}{\le}{ \stepBound}) \ge \ProbModelScheduler{\ma}{\sched}(\dsBoundedPaths{\stepBound \digConstant}{\le}{ \stepBound}) \cdot e^{-\rateMax\ttime}\ .
	\end{displaymath}
\end{lemma}
\begin{proof}
	First, we show that the set $\dsBoundedPaths{\stepBound \digConstant + \ttime}{\le}{ \stepBound}$ corresponds to the paths of $\dsBoundedPaths{\stepBound \digConstant}{\le}{ \stepBound}$ with the additional requirement that no transition is taken between the time points $\stepBound \digConstant$ and $\stepBound \digConstant + \ttime$, i.e.,
	% there is no prefix of $\ppath$ that has a time duration between these time points.
	%More formally, for all $\ppath \in \IPaths[\ma]$ we have
	\begin{displaymath}
	\dsBoundedPaths{\stepBound \digConstant + \ttime}{\le}{ \stepBound} 
	= \{\ppath \in  \dsBoundedPaths{\stepBound \digConstant}{\le}{ \stepBound} \mid \text{there is no prefix } \ppath' \text{ of } \ppath \text{ with } \stepBound \digConstant < \timeOfPath{\ppath'} \le \stepBound \digConstant + \ttime\}.%\ .
	\end{displaymath}
	\begin{itemize}
		\item[``$\subseteq$'':] If $\ppath \in \dsBoundedPaths{\stepBound \digConstant + \ttime}{\le}{ \stepBound}$, then $\ppath  \in \dsBoundedPaths{\stepBound \digConstant}{\le}{ \stepBound}$ follows immediately.
		Furthermore, assume towards a contradiction that there is a prefix $\ppath'$ of $\ppath$ with $\stepBound  \digConstant < \timeOfPath{\ppath'}  \le \stepBound \digConstant + \ttime$.
		Then, $\stepBound <  \nicefrac{\timeOfPath{\ppath'}}{\digConstant} \le \numOfDS{\ppath'} $ (cf. Lemma~\ref{lem:app:digStepsBoundDuration}).
		%As $\timeOfPath{\ppath'} \le \stepBound \digConstant + \ttime$, this means that $\ppath$ has more than $\stepBound$ digitization steps up to the first $\stepBound \digConstant + \ttime$ time units which contradicts $\ppath  \in \dsBoundedPaths{\stepBound \digConstant + \ttime}{\le}{ \stepBound}$.
		As $\timeOfPath{\ppath'} \le \stepBound \digConstant + \ttime$, this means that $\numOfDS{\prefTime{\ppath}{\stepBound\digConstant + \ttime}} \ge \numOfDS{\ppath'} > k$ which contradicts $\ppath  \in \dsBoundedPaths{\stepBound \digConstant + \ttime}{\le}{ \stepBound}$.
		\item[``$\supseteq$'':] For $\ppath \in \dsBoundedPaths{\stepBound \digConstant}{\le}{ \stepBound}$ with no prefix $\ppath'$ such that $\stepBound \digConstant < \timeOfPath{\ppath'}  \le \stepBound \digConstant + \ttime$, it holds that $\prefTime{\ppath}{\stepBound \digConstant + \ttime} = \prefTime{\ppath}{\stepBound \digConstant}$.
		Hence, $\numOfDS{\prefTime{\ppath}{\stepBound \digConstant + \ttime}} = \numOfDS{\prefTime{\ppath}{\stepBound \digConstant}} \le \stepBound$ and it follows that $\ppath \in \dsBoundedPaths{\stepBound \digConstant + \ttime}{\le}{ \stepBound}$.
	\end{itemize}
	The probability for no transition to be taken between $\stepBound \digConstant$ and $\stepBound \digConstant + \ttime$ only depends on the current state at time point $\stepBound \digConstant$.
	More precisely, for some state $\state \in \MS$ assume the set of paths $\{\ppath \in \dsBoundedPaths{\stepBound \digConstant}{\le}{ \stepBound} \mid \last{\prefTime{\ppath}{\stepBound \digConstant}} = \state\}$. The probability that a path in this set takes no transition between time points $\stepBound \digConstant$ and $\stepBound \digConstant + \ttime$ is given by $e^{-\rateAtState{\state} \ttime}$.
	With $\rateMax \ge \rateAtState{\state}$ for all $\state \in \MS$ it follows that
	\begin{align*}
	& \mathrel{\phantom{=}}\ProbModelScheduler{\ma}{\sched}(\dsBoundedPaths{\stepBound \digConstant + \ttime}{\le}{ \stepBound}) \\
	& = \ProbModelScheduler{\ma}{\sched}( \{\ppath \in \dsBoundedPaths{\stepBound \digConstant}{\le}{ \stepBound} \mid  \text{there is no prefix } \ppath' \text{ of } \ppath \text{ with } \stepBound \digConstant < \timeOfPath{\ppath'}  \le \stepBound \digConstant + \ttime\} ) \\
	& = \sum_{\state \in \MS} \ProbModelScheduler{\ma}{\sched}( \{\ppath \in \dsBoundedPaths{\stepBound \digConstant}{\le}{ \stepBound} \mid \last{\prefTime{\ppath}{\stepBound \digConstant}} = \state \} ) \cdot e^{-\rateAtState{\state} \ttime}\\
	& \ge  \sum_{\state \in \MS} \ProbModelScheduler{\ma}{\sched}( \{\ppath \in \dsBoundedPaths{\stepBound \digConstant}{\le}{ \stepBound} \mid \last{\prefTime{\ppath}{\stepBound \digConstant}} = \state \} ) \cdot e^{-\rateMax \ttime}\\
	& = \ProbModelScheduler{\ma}{\sched}(\dsBoundedPaths{\stepBound \digConstant}{\le}{ \stepBound}) \cdot e^{-\rateMax \ttime}\ .
	\end{align*}
\qed\end{proof}

\begin{proof}[of Lemma~\ref{lem:app:bounded:upperBoundForDSBoundedPaths}]
	Let $\maDefNR$.
	By induction over $\stepBound$ we show that 
	\begin{align*}
	\label{eq:app:bounded:lowerBoundForDSBoundedPaths}
	\ProbModelScheduler{\ma}{\sched}(\dsBoundedPaths{\stepBound \digConstant}{\le}{ \stepBound}) 
	\ge (1+ \rateMax \digConstant)^{\stepBound} \cdot e^{- \rateMax \digConstant \stepBound}. 
	\end{align*}
	The claim follows as 
	$\dsBoundedPaths{\stepBound \digConstant}{>}{\stepBound} = \IPaths[\ma] \setminus \dsBoundedPaths{\stepBound \digConstant}{\le}{\stepBound}$.
	
	For $\stepBound=0$, we have $\ppath \in \dsBoundedPaths{0 \cdot \digConstant}{\le}{ 0}$ iff $\ppath$ takes no Markovian transition at time point zero.
	As this happens with probability one, it follows that
	\begin{displaymath}
	\ProbModelScheduler{\ma}{\sched}(\dsBoundedPaths{0 \cdot \digConstant}{\le}{ 0}) 
	= 1
	= (1+ \rateMax  \digConstant)^0 \cdot e^{-\rateMax \digConstant \cdot 0 }\ .
	\end{displaymath}
	
	We assume in the induction step that the proposition holds for some fixed $\stepBound$.
	We distinguish between two cases for the initial state $\sinit$ of $\ma$.
	
	\paragraph{\underline{Case $\sinit \in \MS$}:}
	We partition the set 
	$\dsBoundedPaths{\stepBound \digConstant + \digConstant}{\le}{ \stepBound + 1} = \PathsWaitDelta \cupdot \PathsNotWaitDelta$ with
	\begin{align*}
	\PathsWaitDelta &= \{ \sinit \pathTransUniv{\ttime} \pathstepTimed{1} \dots \in \dsBoundedPaths{\stepBound \digConstant + \digConstant}{\le}{ \stepBound + 1} \mid \ttime \ge \digConstant \} 
	\text{ and } \\ 
	\PathsNotWaitDelta &= \{ \sinit \pathTransUniv{\ttime} \pathstepTimed{1} \dots \in \dsBoundedPaths{\stepBound \digConstant + \digConstant}{\le}{ \stepBound + 1} \mid \ttime < \digConstant \}.%\ . 
	\end{align*}
	Hence, $\PathsWaitDelta$ contains the paths where we wait at least $\digConstant$ time units at $\sinit$ and  $\PathsNotWaitDelta$ contains the paths where the first transition is taken within $\ttime < \digConstant$ time units.
	It follows that $\ProbModelScheduler{\ma}{\sched}(\dsBoundedPaths{\stepBound \digConstant + \digConstant}{\le}{ \stepBound + 1} ) = \ProbModelScheduler{\ma}{\sched}(\PathsWaitDelta) + \ProbModelScheduler{\ma}{\sched}(\PathsNotWaitDelta)$.
	We consider the probabilities for $\PathsWaitDelta$ and $\PathsNotWaitDelta$ separately.
	\begin{itemize}
		\item $\ProbModelScheduler{\ma}{\sched}(\PathsWaitDelta)$:
		For a path $\sinit \pathTransUniv{\ttime+\digConstant} \pathstepTimed{1} \dots \in \PathsWaitDelta$, after the first $\digConstant$ time units there are at most $\stepBound$ digitization steps within the next $\stepBound  \digConstant$ time units, i.e.,
		\begin{displaymath}
		\sinit \pathTransUniv{\ttime + \digConstant} \pathstepTimed{1} \dots \in \PathsWaitDelta \iff 
		\sinit \pathTransUniv{\ttime} \pathstepTimed{1} \dots \in \dsBoundedPaths{\stepBound \digConstant}{\le}{ \stepBound}.
		\end{displaymath}
		The probability for $\PathsWaitDelta$ can therefore be derived from the probability to wait at $\sinit$ for at least $\digConstant$ time units and the probability for $\dsBoundedPaths{\stepBound \digConstant}{\le}{ \stepBound}$.
		In order to apply this, we need to modify the considered scheduler as it might depend on the sojourn time in $\sinit$.
		Let $\sched_\digConstant$ be the scheduler for $\ma$ that mimics $\sched$ on paths where the first transition is delayed by $\digConstant$, i.e., $\sched_\digConstant$ satisfies
		\begin{displaymath}
		\schedEval{\sched_\digConstant}{\sinit \pathTransUniv{\ttime} \dots \pathstepTimedEnd{n-1}{n}}{\act}
		=
		\schedEval{\sched}{\sinit \pathTransUniv{\ttime + \digConstant} \dots  \pathstepTimedEnd{n-1}{n}}{\act}.
		\end{displaymath}
		for all $\sinit \pathTransUniv{\ttime} \dots \pathstepTimedEnd{n-1}{n} \in \FPaths[\ma]$ and $\act \in \Actions$.
		It holds that
		\begin{align}
		\ProbModelScheduler{\ma}{\sched}( \PathsWaitDelta )
		& =  e^{-\rateAtState{\sinit} \digConstant} \cdot \ProbModelScheduler{\ma}{\sched_\digConstant}( \dsBoundedPaths{\stepBound \digConstant}{\le}{ \stepBound}) \nonumber \\
		& \geInductionHypothesis e^{-\rateAtState{\sinit} \digConstant} \cdot (1+ \rateMax \digConstant)^{\stepBound} \cdot e^{- \rateMax \digConstant \stepBound}  \nonumber \\
		& =  e^{-\rateAtState{\sinit} \digConstant} \cdot  (1+ \rateMax \digConstant)^{\stepBound} \cdot e^{- \rateMax \digConstant \stepBound}   \cdot e^{-\rateMax \digConstant} \cdot e^{\rateMax \digConstant} \nonumber\\
		& =  (1+ \rateMax \digConstant)^{\stepBound} \cdot e^{-\rateMax \digConstant (\stepBound+1)} \cdot e^{(\rateMax-\rateAtState{\sinit}) \digConstant}\ . \label{eq:app:boundedUpperBoundOnBoundedReachProb:SlowPaths}
		\end{align}
		\item $\ProbModelScheduler{\ma}{\sched}(\PathsNotWaitDelta)$:
		For a path $\sinit \pathTransUniv{\ttime} \pathstepTimed{1} \dots \in \PathsNotWaitDelta$, the first digitization step happens at less than $\digConstant$ time units, i.e., $0 \le  \ttime <\digConstant$.
		It follows that  there are at most $\stepBound$ digitization steps  in the remaining $ \stepBound \digConstant + \digConstant - \ttime$ time units, i.e.,
		\begin{displaymath}
		\sinit \pathTransUniv{\ttime} \pathstepTimed{1} \pathstepTimed{2} \dots \in \PathsNotWaitDelta
		\iff 
		\pathstepTimed{1} \pathstepTimed{2} \dots \in \dsBoundedPaths[\state_1]{ \stepBound \digConstant + \digConstant -  \ttime}{\le \stepBound}\ ,
		\end{displaymath}
		where $\dsBoundedPaths[\state_1]{ \stepBound \digConstant + \digConstant -   \ttime}{\le}{ \stepBound}$ refers to the paths $\dsBoundedPaths{ \stepBound \digConstant + \digConstant -  \ttime}{\le}{ \stepBound}$ of $\maWithInitState{\state_1} = (\States, \Actions,\allowbreak \probTransRel, \state_1, \rewFct[1], \dots, \rewFct[\numRewFunctions])$, the MA obtained from $\ma$ by changing the initial state to $\state_1$. 
		Hence, the probability for $\PathsNotWaitDelta$ can be derived from the probability to take a transition from $\sinit$ to some state $\state$ within $ \ttime < \digConstant$ time units and the probability for $\dsBoundedPaths[\state]{ \stepBound \digConstant + \digConstant -   \ttime}{\le}{ \stepBound} $.
		Again, we need to adapt the considered scheduler.
		Let $\ppath \in \FPaths[\ma]$ with $\last{\ppath} = \state$.
		The scheduler $\schedShift{\sched}{\ppath}$ for $\maWithInitState{\state}$ mimics the scheduler $\sched$ for $\ma$, where $\ppath$ is prepended to the given path, i.e., we set
		\begin{displaymath}
		\schedEval{\schedShift{\sched}{\ppath}}{\state \pathTransTimed{j} \dots \pathstepTimedEnd{n-1}{n}}{\act}
		=
		\schedEval{\sched}{\ppath \pathTransTimed{j} \dots  \pathstepTimedEnd{n-1}{n}}{\act}%\ .
		\end{displaymath}
		for all $\state \pathTransTimed{j} \dots \pathstepTimedEnd{n-1}{n} \in \FPaths[\maWithInitState{\state}]$ and $\act \in \Actions$.
		With Lemma~\ref{lem:app:LowerBoundForPathsWithBoundedStepsPlusTime} it follows that
		\begin{align}
		& \mathrel{\phantom{=}}\ProbModelScheduler{\ma}{\sched}( \PathsNotWaitDelta  )\nonumber\\
		& =  \int_{0}^{\digConstant} \rateAtState{\sinit} \cdot e^{-\rateAtState{\sinit} \ttime} \cdot \left(\sum_{\state \in \States} \probP(\sinit, \markovianAct, \state) \cdot \ProbModelScheduler{\maWithInitState{\state}}{\schedShift{\sched}{\ppath}}( \dsBoundedPaths[\state]{ \stepBound \digConstant + \digConstant - \ttime}{\le}{ \stepBound} ) \right) \diff \ttime \nonumber\\
		& \ge  \int_{0}^{\digConstant} \rateAtState{\sinit} \cdot e^{-\rateAtState{\sinit} \ttime} \cdot \left(\sum_{\state \in \States} \probP(\sinit, \markovianAct, \state) \cdot \ProbModelScheduler{\maWithInitState{\state}}{\schedShift{\sched}{\ppath}}( \dsBoundedPaths[\state]{ \stepBound \digConstant}{\le}{ \stepBound} ) \cdot e^{-\rate (\digConstant - \ttime)} \right)\diff \ttime \nonumber\\
		& \geInductionHypothesis  \int_{0}^{\digConstant} \rateAtState{\sinit} \cdot e^{-\rateAtState{\sinit} \ttime} \cdot \left(\sum_{\state \in \States} \probP(\sinit, \markovianAct, \state) \cdot (1+\rateMax \digConstant)^{\stepBound} \cdot e^{- \rateMax \digConstant \stepBound} \cdot e^{-\rate (\digConstant - \ttime)} \right) \diff \ttime \nonumber\\
		& =   (1+\rateMax \digConstant)^{\stepBound} \cdot e^{- \rateMax \digConstant \stepBound} \cdot  \rateAtState{\sinit} \cdot \int_{0}^{\digConstant}   e^{-\rateAtState{\sinit} \ttime} \cdot e^{-\rate (\digConstant - \ttime)} \cdot \left(\sum_{\state \in \States} \probP(\sinit, \markovianAct, \state)  \right) \diff \ttime \nonumber\\
		& = (1+\rateMax \digConstant)^{\stepBound} \cdot e^{- \rateMax \digConstant \stepBound} \cdot \rateAtState{\sinit} \cdot \int_{0}^{\digConstant}  e^{-\rateAtState{\sinit} \ttime} \cdot   e^{-\rate \digConstant} \cdot   e^{\rate\ttime} \diff \ttime \nonumber\\
		& =   (1+\rateMax \digConstant)^{\stepBound} \cdot e^{-\rateMax \digConstant (\stepBound+1) } \cdot \rateAtState{\sinit} \cdot \int_{0}^{\digConstant} e^{(\rateMax-\rateAtState{\sinit}) \ttime} \diff \ttime\ . \label{eq:app:boundedUpperBoundOnBoundedReachProb:FastPaths}
		\end{align}
	\end{itemize}
	
	Combining the results for $\PathsWaitDelta$ and $\PathsNotWaitDelta$ (i.e., Equations~\ref{eq:app:boundedUpperBoundOnBoundedReachProb:SlowPaths} and~\ref{eq:app:boundedUpperBoundOnBoundedReachProb:FastPaths}), we obtain
	\begin{align*}
	& \mathrel{\phantom{=}} \ProbModelScheduler{\ma}{\sched}(\dsBoundedPaths{\stepBound \digConstant + \digConstant}{\le}{ \stepBound + 1}) \nonumber \\
	&= \ProbModelScheduler{\ma}{\sched}(\PathsWaitDelta) + \ProbModelScheduler{\ma}{\sched}(\PathsNotWaitDelta) \\
	& \ge (1+ \rateMax  \digConstant)^{\stepBound} \cdot e^{-\rateMax \digConstant (\stepBound+1) } \cdot \Big( e^{(\rateMax-\rateAtState{\sinit}) \digConstant} +  \rateAtState{\sinit} \cdot  \int_{0}^{\digConstant} e^{(\rateMax-\rateAtState{\sinit}) \ttime}   \diff\ttime \Big)\\
	& \overset{\equalityMarker}{\ge} (1+ \rateMax \digConstant)^{\stepBound} \cdot e^{-\rateMax \digConstant (\stepBound+1) } \cdot \left( 1 + \rateMax \digConstant \right) 
	 = (1+ \rateMax \digConstant)^{\stepBound+1} \cdot e^{-\rateMax \digConstant (\stepBound+1) }\ ,
	\end{align*}
	where the inequality marked with $\equalityMarker$ is due to
	\begin{align*}
	& \mathrel{\phantom{=}} e^{(\rateMax-\rateAtState{\sinit}) \digConstant} +  \rateAtState{\sinit} \cdot  \int_{0}^{\digConstant} e^{(\rateMax-\rateAtState{\sinit}) \ttime}   \diff\ttime \\
	& =  e^{(\rateMax-\rateAtState{\sinit}) \digConstant} +  (\rateAtState{\sinit} - \rateMax + \rateMax ) \cdot  \int_{0}^{\digConstant} e^{(\rateMax-\rateAtState{\sinit}) \ttime}   \diff\ttime \\
	& =  e^{(\rateMax-\rateAtState{\sinit}) \digConstant} - \big(\rateMax - \rateAtState{\sinit} \big) \cdot  \int_{0}^{\digConstant} e^{(\rateMax-\rateAtState{\sinit}) \ttime}   \diff\ttime  + \rateMax \cdot  \int_{0}^{\digConstant} e^{(\rateMax-\rateAtState{\sinit}) \ttime}   \diff\ttime \\
	& = 
	\begin{dcases}
	1 - 0  + \rateMax \cdot  \int_{0}^{\digConstant} e^{(\rateMax-\rateAtState{\sinit}) \ttime} \diff\ttime & \text{if } \rateAtState{\sinit} = \rateMax \\
	e^{(\rateMax-\rateAtState{\sinit}) \digConstant} -  \big( e^{(\rateMax-\rateAtState{\sinit}) \digConstant} - 1 \big)  + \rateMax \cdot  \int_{0}^{\digConstant} e^{(\rateMax-\rateAtState{\sinit}) \ttime}   \diff\ttime & \text{if } \rateAtState{\sinit} < \rateMax
	\end{dcases}\\
	&=  1 + \rateMax \cdot  \int_{0}^{\digConstant} e^{(\rateMax-\rateAtState{\sinit}) \ttime}   \diff\ttime
	\ge  1 + \rateMax \cdot  \int_{0}^{\digConstant} 1  \diff\ttime  = 1 + \rateMax \digConstant\ .
	\end{align*}
	
	\paragraph{\underline{Case $\sinit \in \PS$:}}
	Since $\ma$ is non-zeno, a state $\state \in \MS$ is reached from $\sinit$ within zero time almost surely (i.e., with probability one).
	From  the previous case, it already follows that the Proposition holds for $\maWithInitState{\state}$ with $\state \in \MS$ and the set $\dsBoundedPaths[\state]{ \stepBound \digConstant + \digConstant}{\le}{ \stepBound+1}$.
	With $\Pathset_\MS = \{\pathFseqTimed \in \FPaths[\ma] \mid \state_n \in \MS \text{ and } \forall i < n \colon \state_i \in \PS\}$ we obtain
	\begin{align*}
	\ProbModelScheduler{\ma}{\sched}(\dsBoundedPaths{\stepBound \digConstant + \digConstant}{\le}{ \stepBound + 1})
	& = \int_{\substack{\ppath \in \Pathset_\MS \\ \last{\ppath} = \state}}     \ProbModelScheduler{\maWithInitState{\state}}{\schedShift{\sched}{\ppath}}(\dsBoundedPaths[\state]{ \stepBound \digConstant + \digConstant}{\le}{ \stepBound+1}) \diff \ProbModelScheduler{\ma}{\sched}(\ppath)\\
	& \ge \int_{\substack{\ppath \in \Pathset_\MS \\ \last{\ppath} = \state}}  (1+\rateMax \digConstant)^{\stepBound+1} \cdot e^{-\rateMax \digConstant (\stepBound+1) }    \diff \ProbModelScheduler{\ma}{\sched}(\ppath)\\
	& =  (1+\rateMax \digConstant)^{\stepBound+1} \cdot e^{-\rateMax \digConstant (\stepBound+1) } \cdot \ProbModelScheduler{\ma}{\sched}(\Pathset_\MS) \\
	& =   (1+\rateMax \digConstant)^{\stepBound+1} \cdot e^{-\rateMax \digConstant (\stepBound+1) } \ .
	\end{align*}
\qed\end{proof}
We now present the proof of Proposition~\ref{prop:bounded:approxBoundedProb}.
\propositionApproxBoundedProb*
\begin{proof}
In Section~\ref{sec:ma:bounded} we already discussed that 
\begin{align*}
\begin{split}
\ProbModelScheduler{}{\sched}(\eventually[\interval] \goalStates) =
\ProbModelScheduler{}{\sched}(\inducedPathsDi{\eventuallyds[{\diOfInterval}] \goalStates}) 
+	\ProbModelScheduler{}{\sched}(\eventually[\interval] \goalStates \setminus \inducedPathsDi{\eventuallyds[{\diOfInterval}] \goalStates})
-	\ProbModelScheduler{}{\sched}(\inducedPathsDi{\eventuallyds[{\diOfInterval}] \goalStates} \setminus \eventually[\interval] \goalStates).
\end{split}
\end{align*}
	The main part of the proof is to show that
	\begin{align}
	\label{eq:bounded:upperBoundsForError}
	\ProbModelScheduler{\ma}{\sched}(\inducedPathsDi{\eventuallyds[{\diOfInterval}] \goalStates} \setminus \eventually[\interval] \goalStates) \le \lowerErrorBound(\interval)
	\text{ and } 
	\ProbModelScheduler{\ma}{\sched}(\eventually[\interval] \goalStates \setminus \inducedPathsDi{\eventuallyds[{\diOfInterval}] \goalStates}) \le \upperErrorBound(\interval).% \tag{a}
	\end{align}
	Then, the proposition follows directly.
We show Equation~\ref{eq:bounded:upperBoundsForError} for the different forms of the time interval $\interval$.

\paragraph{\underline{Case $\interval = [0, \infty)$:}}
In this case we have $\diOfInterval = \NN$. It follows that 
\begin{displaymath}
\inducedPathsDi{\eventuallyds[{\diOfInterval}] \goalStates} = \eventually[\interval] \goalStates = \{\pathIseqTimed \in \IPaths[\ma] \mid \state_i \in \goalStates \text{ for some } i \ge 0  \}.
\end{displaymath}
Hence,
\begin{displaymath}
\ProbModelScheduler{\ma}{\sched}(\inducedPathsDi{\eventuallyds[{\diOfInterval}] \goalStates} \setminus \eventually[\interval] \goalStates) 
=
\ProbModelScheduler{\ma}{\sched}(\eventually[\interval] \goalStates \setminus \inducedPathsDi{\eventuallyds[{\diOfInterval}] \goalStates}) 
= \ProbModelScheduler{\ma}{\sched}(\emptyset) = 0 = \lowerErrorBound(\interval) = \upperErrorBound(\interval).
\end{displaymath}

\paragraph{\underline{Case $\interval = [0,\upperTimeBound]$ for $\upperTimeBound = \upperStepBound\digConstant$:}}
We have $\diOfInterval = \{0,1, \dots, \upperStepBound\}$.

\begin{itemize}
	\item
We show that $\inducedPathsDi{\eventuallyds[{\diOfInterval}] \goalStates} \subseteq \eventually[\interval] \goalStates$ which implies
\begin{displaymath}
\ProbModelScheduler{\ma}{\sched}(\inducedPathsDi{\eventuallyds[{\diOfInterval}] \goalStates} \setminus \eventually[\interval] \goalStates) = \ProbModelScheduler{\ma}{\sched}(\emptyset) = 0 = \lowerErrorBound(\interval).
\end{displaymath}
Let $\ppath \in \inducedPathsDi{\eventuallyds[{\diOfInterval}] \goalStates}$ and let $\ppath'$ be the smallest prefix of $\ppath$ with $\last{\ppath'} \in \goalStates$.
It follows that $\diOf{\ppath'}$ is also the smallest prefix of $\diOfPath$ with $\last{\diOf{\ppath'}} \in \goalStates$.
Since $\diOfPath \in \eventuallyds[\diOfInterval] \goalStates$, it follows that $\numOfDS{\ppath'} = \numOfDS{\diOf{\ppath'}} \le \upperStepBound$.
From  Lemma~\ref{lem:app:digStepsBoundDuration} we obtain 
\begin{displaymath}
\timeOfPath{\ppath'} \le \numOfDS{\ppath'} \cdot \digConstant = \numOfDS{\diOf{\ppath'}} \cdot \digConstant \le \upperStepBound  \digConstant = \upperTimeBound\ .
\end{displaymath}
Hence, the prefix $\ppath'$ reaches $\goalStates$ within $\upperTimeBound$ time units, implying $\ppath \in \eventually[\interval] \goalStates$.

\item
Next, we show $\eventually[\interval] \goalStates \setminus \inducedPathsDi{\eventuallyds[\diOfInterval] \goalStates} \subseteq \dsBoundedPaths{\upperTimeBound}{>}{\upperStepBound}$.
With Lemma~\ref{lem:app:bounded:upperBoundForDSBoundedPaths} we obtain
\begin{displaymath}
\ProbModelScheduler{\ma}{\sched}(\eventually[\interval] \goalStates \setminus \inducedPathsDi{\eventuallyds[{\diOfInterval}] \goalStates}) \le  \ProbModelScheduler{\ma}{\sched}(\dsBoundedPaths{\upperTimeBound}{>}{\upperStepBound})
\le 1 - (1+ \rateMax \digConstant)^{\upperStepBound} \cdot e^{- \rateMax \upperTimeBound}
 = \upperErrorBound(\interval)
\end{displaymath}
	
Consider a path $\ppath \in \eventually[\interval] \goalStates \setminus \inducedPathsDi{\eventuallyds[\diOfInterval] \goalStates}$.
Note that $\ppath$ reaches $\goalStates$ within $\upperTimeBound$ time units but with more than $\upperStepBound$ digitization steps.
Hence, the prefix of $\ppath$ up to time point $\upperTimeBound$ certainly has more than $\upperStepBound$ digitization steps, i.e., $\ppath$ satisfies $\numOfDS{\prefTime{\ppath}{\upperTimeBound}} > \upperStepBound$ which means $\ppath \in \dsBoundedPaths{\upperTimeBound}{>}{\upperStepBound}$.
\end{itemize}

\paragraph{\underline{Case $\interval = [\lowerTimeBound,\infty)$ for $\lowerTimeBound = \lowerStepBound\digConstant$:}}
	We have $\diOfInterval = \{\lowerStepBound+1, \lowerStepBound+2, \dots\}$.
\begin{itemize}
	\item 
	We show that $\inducedPathsDi{\eventuallyds[\diOfInterval] \goalStates} \setminus \eventually[\interval] \goalStates \subseteq \dsBoundedPaths[]{\lowerTimeBound}{>}{\lowerStepBound}$.
	With Lemma~\ref{lem:app:bounded:upperBoundForDSBoundedPaths} we obtain
	\begin{displaymath}
	\ProbModelScheduler{\ma}{\sched}(\inducedPathsDi{\eventuallyds[\diOfInterval] \goalStates} \setminus \eventually[\interval] \goalStates ) \le  \ProbModelScheduler{\ma}{\sched}(\dsBoundedPaths{\lowerTimeBound}{>}{\lowerStepBound})
	\le 1 - (1+ \rateMax \digConstant)^{\lowerStepBound} \cdot e^{- \rateMax \lowerTimeBound}
	= \lowerErrorBound(\interval).
	\end{displaymath}
	Consider a path $\ppath \in \inducedPathsDi{\eventuallyds[\diOfInterval] \goalStates} \setminus \eventually[\interval] \goalStates$.
	As $\ppath \notin \eventually[\interval] \goalStates$, it follows that $\ppath$ has to reach (and leave) $\goalStates$ within less than $\lowerTimeBound$ time units.
	Let $\pathDi$ be the largest prefix of $\diOfPath$ that satisfies $\last{\pathDi} \in \goalStates$.
	Our observations yield that $\ppath$ leaves $\last{\pathDi}$ before time point $\lowerTimeBound$.
	Hence, $\pathDi$ is a prefix of $\diOf{\prefTime{\ppath}{\lowerTimeBound}}$.
	Moreover, $\numOfDS{\pathDi} \in \diOfInterval$ as $\diOfPath \in \eventuallyds[\diOfInterval] \goalStates$.
	It follows that $\numOfDS{\prefTime{\ppath}{\lowerTimeBound}} \ge \numOfDS{\pathDi} > \lowerStepBound$ which implies $\ppath \in \dsBoundedPaths[]{\lowerTimeBound}{>}{\lowerStepBound}$.
\item	
	Now consider a path $\ppath \in \eventually[\interval] \goalStates \setminus \inducedPathsDi{\eventuallyds[\diOfInterval] \goalStates}$.
	$\ppath$ visits $\goalStates$ at least once since $\ppath \in \eventually[\interval] \goalStates$. 
	Moreover, $\diOfPath$  does not visit $\goalStates$ after $\lowerStepBound$ digitization steps  due to $\ppath \notin \inducedPathsDi{\eventuallyds[\diOfInterval] \goalStates}$.
	This means $\ppath$ visits $\goalStates$ only finitely often.
	Let $\ppath' = \pathFseqTimed$ be the largest prefix of $\ppath$ such that $\state_n \in \goalStates$.
	Notice that $\numOfDS{\ppath'} \le \lowerStepBound$ holds.
	Let $\ppath' \pathTransTimed{} \state$ be the prefix of $\ppath$ of length $\length{\ppath'}+1$.
		We show by contradiction that $\lowerTimeBound \le \timeOfPath{\ppath' \pathTransTimed{} \state} < \lowerTimeBound + \digConstant$ holds:
		\begin{itemize}
			\item If $\timeOfPath{\ppath' \pathTransTimed{} \state} < \lowerTimeBound$, then $\last{\ppath'} \in \goalStates$ is left before time point $\lowerTimeBound$ which contradicts $\ppath \in \eventually[\interval] \goalStates$.
			\item Further, assume that $\timeOfPath{\ppath' \pathTransTimed{} \state} \ge \lowerTimeBound + \digConstant$.
			With Lemma~\ref{lem:app:digStepsBoundDuration} we obtain 
			\begin{align*}
			\timeOfStamp[] 
			&	\ge \lowerTimeBound + \digConstant - \timeOfPath{\ppath'}\\
			&	\ge \lowerTimeBound + \digConstant - \numOfDS{\ppath'} \cdot \digConstant\\
			&	\ge (\lowerStepBound + 1 - \underbrace{\numOfDS{\ppath'}}_{\le \lowerStepBound}) \cdot \digConstant > 0\ .
			\end{align*}
			Hence, $\ppath$ stays at $\last{\ppath'}$ for at least $(j + 1 - \numOfDS{\ppath'}) \cdot \digConstant$ time units which means that 
			$\diOf{\ppath'} \big({\pathTransUniv{\markovianAct}} \last{\ppath'}\big)^{j+1 - \numOfDS{\ppath'}} = \pathDi$ is a prefix of $\diOfPath$.
			Since $\numOfDS{\pathDi} = %\numOfDS{\ppath'} + j + 1 - \numOfDS{\ppath'} =
			j+1$, this contradicts $\ppath \notin \inducedPathsDi{\eventuallyds[\diOfInterval] \goalStates}$.
		\end{itemize}
		We infer that $\ppath$ takes at least one transition in the time interval $[\lowerTimeBound, \lowerTimeBound + \digConstant)$.
		The probability for this can be upper bounded by $1-e^{-\rateMax \digConstant}$, i.e.,
	\begin{align*}
	&\phantom{\le} \ProbModelScheduler{\ma}{\sched}(\eventually[\interval] \goalStates \setminus \inducedPathsDi{\eventuallyds[\diOfInterval] \goalStates}) \\
	&\le  \ProbModelScheduler{\ma}{\sched}(\{ \ppath \in \IPaths[\ma] \mid \ppath \text{ takes a transition in time interval } [\lowerTimeBound, \lowerTimeBound + \digConstant)  \}) \\
	&\le  1-e^{-\rateMax \digConstant}
	= \upperErrorBound(\interval).
	\end{align*}
\end{itemize}

\paragraph{\underline{Case $\interval = [\lowerTimeBound,\upperTimeBound]$ for $\lowerTimeBound = \lowerStepBound\digConstant$ and $\upperTimeBound = \upperStepBound \digConstant$:}}
		We have $\diOfInterval = \{\lowerStepBound+1, \lowerStepBound+2, \dots, \upperStepBound\}$.
		
\begin{itemize}
\item 
		As in the case ``$\interval = [\lowerTimeBound, \infty)$'', 
		we show that $\inducedPathsDi{\eventuallyds[\diOfInterval] \goalStates} \setminus \eventually[\interval] \goalStates \subseteq \dsBoundedPaths[]{\lowerTimeBound}{>}{\lowerStepBound}$.
		With Lemma~\ref{lem:app:bounded:upperBoundForDSBoundedPaths} we obtain
		\begin{displaymath}
		\ProbModelScheduler{\ma}{\sched}(\inducedPathsDi{\eventuallyds[\diOfInterval] \goalStates} \setminus \eventually[\interval] \goalStates ) \le  \ProbModelScheduler{\ma}{\sched}(\dsBoundedPaths{\lowerTimeBound}{>}{\lowerStepBound})
		\le 1 - (1+ \rateMax \digConstant)^{\lowerStepBound} \cdot e^{- \rateMax \lowerTimeBound}
		= \lowerErrorBound(\interval).
		\end{displaymath}
		Let $\ppath \in \inducedPathsDi{\eventuallyds[\diOfInterval] \goalStates} \setminus \eventually[\interval] \goalStates$ and let $\pathDi$ be the largest prefix of $\diOfPath$ with $\last{\pathDi} \in \goalStates$ and $\numOfDS{\pathDi} \in \diOfInterval$.
Such a prefix exists due to $\ppath \in \inducedPathsDi{\eventuallyds[\diOfInterval] \goalStates}$.
$\ppath$ reaches $\last{\pathDi}$ with at most $\upperStepBound$ digitization steps and therefore within at most $\upperTimeBound$ time units (cf. Lemma~\ref{lem:app:digStepsBoundDuration}).
As $\ppath \notin \eventually[\interval] \goalStates$, we conclude that $\ppath$ has to reach (and leave) $\last{\pathDi}$ within less than $\lowerTimeBound$ time units.
It follows that $\numOfDS{\prefTime{\ppath}{\lowerTimeBound}} \ge \numOfDS{\pathDi} > \lowerStepBound$ which implies $\ppath \in \dsBoundedPaths[]{\lowerTimeBound}{>}{\lowerStepBound}$.

\item
Next, let $\ppath \in \eventually[\interval] \goalStates \setminus \inducedPathsDi{\eventuallyds[\diOfInterval] \goalStates}$ and let $\ppath' = \pathFseqTimed$ be the largest prefix of $\ppath$ such that $\state_n \in \goalStates$ and $\timeOfPath{\ppath'} \le \upperTimeBound$. Such a prefix exists due to $\ppath \in \eventually[\interval] \goalStates$. We distinguish two cases.
	\begin{itemize}
		\item If $\numOfDS{\ppath'} > \upperStepBound$, then $\ppath \in \dsBoundedPaths[]{\upperTimeBound}{>}{\upperStepBound}$ since
		$
		\numOfDS{\prefTime{\ppath}{\upperTimeBound}} \ge \numOfDS{\ppath'} > \upperStepBound
		$.
		\item If $\numOfDS{\ppath'} \le \upperStepBound$, then $\numOfDS{\ppath'} \le \lowerStepBound$ holds due to $\ppath \notin \inducedPathsDi{\eventuallyds[\diOfInterval] \goalStates}$.
		Similar to the case ``$\interval = [\lowerTimeBound,\infty)''$ we can show that $\ppath$ takes at least one transition in time interval $[\lowerTimeBound, \lowerTimeBound + \digConstant)$.
	\end{itemize}
It follows that 
	\begin{align*}
	&\phantom{\subseteq}\ \ \eventually[\interval] \goalStates \setminus \inducedPathsDi{\eventuallyds[\diOfInterval] \goalStates} \\
	&\subseteq \dsBoundedPaths[]{\upperTimeBound}{>}{\upperStepBound} \cup  \{ \ppath \in \IPaths[\ma] \mid \ppath \text{ takes a transition in time interval } [\lowerTimeBound, \lowerTimeBound + \digConstant)  \}
	\end{align*}
	Hence,
	\begin{align*}
	\ProbModelScheduler{\ma}{\sched}(\eventually[\interval] \goalStates \setminus \inducedPathsDi{\eventuallyds[\diOfInterval] \goalStates}) 
	\le 1 - (1+ \rateMax \digConstant)^{\upperStepBound} \cdot e^{- \rateMax \upperTimeBound}  + 1 - e^{-\rateMax \digConstant} 
	= \upperErrorBound(\interval).
	\end{align*}
\end{itemize}
\qed\end{proof}		
	
\subsection{Proof of Theorem~\ref{thm:bounded:multiBounded}}
	\theoremMultiBounded*
	\begin{proof}
		For simplicity, we assume that only the threshold relation $\ge$ is considered, i.e., ${\rel} = (\ge, \dots, \ge)$.
		Furthermore, we restrict ourself to (un)timed reachability objectives.
		The remaining cases are treated analogously.
		
		First assume a point $\point' = (\pointi{1}', \dots, \pointi{\numObjectives}') \in \underApprox$.
		Consider the point $\point = \pointTuple$  satisfying $\pointi{i}' = \pointi{i} - \lowerErrorBound[i]$  for each index $i$.
		It follows that $\point' \in \errorBoundObjPointUniv{\obj}{\point}$ and thus $\dma, \bar{\sched} \models \obj \rel \point$ for some scheduler $\bar{\sched} \in \TASched[\dma]$.
		Consider the scheduler $\sched \in \GMSched[\ma]$ given by $\schedEval{\sched}{\ppath}{\act} = \schedEval{\bar{\sched}}{\diOfPath}{\act}$
		for each path $\ppath \in \FPaths[\ma]$ and action $\act \in \Actions$.
		Notice that $\bar{\sched} = \diOfSched$.
		For an index $i$ let $\obj[i]$ be the objective $\intervalBoundedReachObj[]$.
		It follows that
		\begin{displaymath}
		\dma, \bar{\sched} \models \obj[i] \ge \pointi{i}
		\iff 
		\dma, \diOfSched \models \obj[i] \ge \pointi{i}
		\iff 
		\ProbModelScheduler{\dma}{\diOfSched}(\eventuallyds[\diOfInterval] \goalStates)\ge \pointi{i}\ ,
		\end{displaymath}
		With Corollary~\ref{corr:bounded:approxBoundedProb} it follows that
		\begin{displaymath}
		\pointi{i}' = \pointi{i} - \lowerErrorBound[i] \le \ProbModelScheduler{\dma}{\diOfSched}(\eventuallyds[\diOfInterval] \goalStates) - \lowerErrorBound[i] \overset{Cor.\,\ref{corr:bounded:approxBoundedProb}}{\le} \intervalBoundedReachProbMa.%\ .
		\end{displaymath}
		As this observation holds for all objectives in $\obj$, it follows that $\ma, \sched \models \obj \rel \point'$, implying $\achievabilityQ{\ma}{\obj \rel \point'}$.
		
		The proof of the second inclusion is similar. 
		Assume that $\ma, \sched \models \obj \rel \point'$ holds for a point $\point' = (\pointi{1}', \dots, \pointi{\numObjectives}') \in \RR^\numObjectives$ and a scheduler $\sched \in \GMSched[\ma]$.
		For some index $i$, consider $\obj[i] = \intervalBoundedReachObj[]$.
		It follows that
		$\intervalBoundedReachProbMa[] \ge \pointi{i}'$. %\ .
		With Corollary~\ref{corr:bounded:approxBoundedProb} we obtain
		\begin{displaymath}
		\pointi{i}' - \upperErrorBound[i] \le \intervalBoundedReachProbMa - \upperErrorBound[i] \overset{Cor.\,\ref{corr:bounded:approxBoundedProb}}{\le} \ProbModelScheduler{\dma}{\diOfSched}(\eventuallyds[\diOfInterval] \goalStates) .%\ .
		\end{displaymath}
		Applying this for all objectives in $\obj$ yields $\dma, \diOfSched \models \obj \rel \point$, where the point $\point = \pointTuple \in \RR^\numObjectives$ satisfies $ \pointi{i} = \pointi{i}' - \upperErrorBound[i]$ or, equivalently,  $ \pointi{i}' = \pointi{i} + \upperErrorBound[i]$ for each index $i$.
		Note that $\point' \in \errorBoundObjPointUniv{\obj}{\point}$ which implies $\point' \in \overApprox$.
	\qed\end{proof}

%% file: D-singleObj.tex
\section{Comparison to Single-objective Analysis}
\label{app:singleObj}
Corollary~\ref{corr:bounded:approxBoundedProb} generalizes existing results from single-objective timed reachability analysis:
For MA $\ma$, goal states $\goalStates$, time bound $\upperTimeBound \in \RRgz$, and digitization constant $\digConstant \in \RRgz$ with $\nicefrac{\upperTimeBound}{\digConstant} = \upperStepBound \in \NN$, 
	\cite[Theorem 5.3]{DBLP:journals/corr/GuckHHKT14} states that
		\begin{displaymath}
\sup_{\sched \in \GMSched[\ma]} \boundedReachProb{\ma}{\sched}{{[0,\upperTimeBound]}}  \in
\sup_{\sched \in \TASched[\dma]}	\ProbModelScheduler{\dma}{\sched}(\eventuallyds[\{0, \dots, \upperStepBound\}] \goalStates)
	+
	\Big[{-}\lowerErrorBound([0, \upperTimeBound]),\,  \upperErrorBound([0, \upperTimeBound])\Big].
	\end{displaymath}
%	\begin{align*}
%		& \phantom{\le} \mathrel{} \sup_{\sched \in \TASched[\dma]} \ProbModelScheduler{\dma}{\sched}(\eventuallyds[\le \upperTimeBound] \goalStates)\\
%		&	\le  \sup_{\sched \in \GMSched[\ma]} \ProbModelScheduler{\ma}{\sched}(\eventually[{[0,\upperTimeBound]}] \goalStates)\\
%		& \le \sup_{\sched \in \TASched[\dma]} \ProbModelScheduler{\dma}{\sched}(\eventuallyds[\le \upperTimeBound] \goalStates) + 1 - (1+ \rateMax \digConstant)^{\upperStepBound} \cdot e^{- \rateMax \upperTimeBound}
%	\end{align*}
%	Besides the extension to lower time bounds (which also have been considered in \cite{HatefiH12}), 
Corollary~\ref{corr:bounded:approxBoundedProb} generalizes this result by
	explicitly referring to the schedulers $\sched \in \GMSched[\ma]$ and $\diOfSched\in\TASched[\dma]$ under which the claim holds.
	This extension is necessary as a multi-objective analysis can not be restricted to schedulers that only optimize a single objective.
%	Further details are given in Section~\ref{sec:ma:bounded:multi}.
	
	We remark that the proof in \cite[Theorem 5.3]{DBLP:journals/corr/GuckHHKT14} can not be adapted to show our result.
	The main reason is that the  proof relies on an auxiliary lemma  which claims that\footnote{We adapt \cite[Lemma G.2]{DBLP:journals/corr/GuckHHKT14} to our notations from Appendix~\ref{app:proofsbounded:thmApprxBoundedProb}.}
	\begin{align}
		\label{eq:ma:wrongClaim}
		\ProbModelScheduler{\ma}{\sched}(\eventually[{[0,\upperTimeBound]}] \goalStates \mid \dsBoundedPaths{\digConstant}{<}{2}) \le \ProbModelScheduler{\ma}{\sched}(\eventually[{[0,\upperTimeBound]}] \goalStates)
	\end{align}
	holds for \emph{all} schedulers $\sched \in \GMSched[\ma]$.
	We show that this claim does \emph{not} hold.
	The intuition is as follows.
	Assume we observe that at most one Markovian transition is taken in $\ma$ within the first $\digConstant$ time units (i.e., we observe a path in $\dsBoundedPaths{\digConstant}{<}{2}$).
	The lemma claims that under this observation the probability to reach $\goalStates$ within $\upperTimeBound$ time units does not increase.
	We give a counterexample to illustrate that there are schedulers for which this is not true.
	Consider the MA $\ma$ from Figure~\ref{fig:ma:counterexample} and let $\sched$ be the scheduler for $\ma$ satisfying
	\begin{displaymath}
	\schedEval{\sched}{\state_0 \pathTransUniv{\ttime_1} \state_1 \pathTransUniv{\ttime_2} \state_2}{\act} =
	\begin{cases}
	1 & \text{if }  \ttime_1 + \ttime_2 > \digConstant \\
	0 & \text{otherwise}.
	\end{cases}
	\end{displaymath}
	Hence, $\sched$ chooses $\act$ iff there are less than two digitization steps  within the first $\digConstant$ time units.
	It follows that the probability to reach $\goalStates = \{\state_3\}$ on a path in  $\dsBoundedPaths{\digConstant}{\ge}{2}$ is zero.
	We conclude that
	\begin{align*}
		\ProbModelScheduler{\ma}{\sched}(\eventually[{[0,\upperTimeBound]}] \{\state_3\})
		& = \ProbModelScheduler{\ma}{\sched}(\eventually[{[0,\upperTimeBound]}] \{\state_3\} \cap \dsBoundedPaths{\digConstant}{<}{2})  + \underbrace{\ProbModelScheduler{\ma}{\sched}(\eventually[{[0,\upperTimeBound]}] \{\state_3\} \cap \dsBoundedPaths{\digConstant}{\ge}{2})}_{=0} \\
		& = \ProbModelScheduler{\ma}{\sched}(\eventually[{[0,\upperTimeBound]}] \{\state_3\} \mid \dsBoundedPaths{\digConstant}{<}{2}) \cdot  \underbrace{\ProbModelScheduler{\ma}{\sched}(\dsBoundedPaths{\digConstant}{<}{2})}_{<1} \\
		& < \ProbModelScheduler{\ma}{\sched}(\eventually[{[0,\upperTimeBound]}] \{\state_3\} \mid \dsBoundedPaths{\digConstant}{<}{2})
	\end{align*}
	which contradicts Equation~\ref{eq:ma:wrongClaim}.

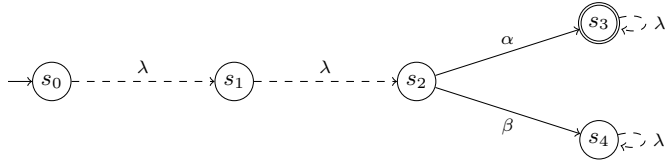
\begin{figure}[t]
	\centering
	\scalebox{\picscale}{
		\input{pics/ma_counterexample}
	}
	\caption{MA $\ma$  (cf. Appendix~\ref{app:singleObj}).}
	\label{fig:ma:counterexample}
\end{figure}

%% file: pics/ma_counterexample.tex
\begin{tikzpicture}[scale=1]

\node [state,] (s0) at (0,0) {$\state_0$};
\node [state,] (s1) [on grid, right=25mm of s0] {$\state_1$};
\node [state,] (s2) [on grid, right=25mm of s1] {$\state_2$};
\node [ ] (aux) [on grid,  right=25mm of s2] {};
\node [state,accepting] (s3) [on grid, above=8mm of aux] {$\state_3$};
\node [state,] (s4) [on grid, below=8mm of aux] {$\state_4$};

\initstateLeft{s0}

\draw (s0) edge[markovian,]  node [, above] {\scriptsize$\rate$} (s1);
\draw (s1) edge[markovian,]  node [, above] {\scriptsize$\rate$} (s2);
    \draw (s2) edge[probabilistic] node[above] {\scriptsize$\act$} (s3);
    \draw (s2) edge[probabilistic] node[below] {\scriptsize$\altact$} (s4);
\draw (s4) edge[markovian, loop right] node[pos=0.5, right] {\scriptsize$\rate$}  (s4);
\draw (s3) edge[markovian, loop right] node[pos=0.5, right] {\scriptsize$\rate$}  (s3);

\end{tikzpicture}

%% file: F-evaluation-details.tex
\section{Further Details for the Experiments}
\label{App:evaluation-details}

\subsection{Benchmark Details}
\label{App:evaluation-details:ma}
We depict additional information regarding our experiments on multi-objective MAs.

\paragraph{Job scheduling.}
The job scheduling case study originates from~\cite{DBLP:journals/jacm/BrunoDF81} and was already discussed in Section~\ref{sec:introduction}.
We consider $N$ jobs that are executed on $K$ identical processors.
Each of the $N$ jobs gets a different rate between 1 and 3.
 We consider the following objectives.
\begin{itemize}
	\item[$\mathbb{E}_1$:] Minimize the expected time until all jobs are completed.
	\item[$\mathbb{E}_2$:] Minimize the expected time until $\lceil \nicefrac{N}{2}\rceil$ jobs are completed.
	\item[$\mathbb{E}_3$:] Minimize the expected waiting time of the jobs.
	\item[$\mathbb{P}$:] Minimize the probability that the job with the lowest rate is completed before the job with the highest rate.
	\item[$\mathbb{P}_1^\le$:] Maximize the probability that all jobs are completed within $\nicefrac{N}{2K}$ time units.
	\item[$\mathbb{P}_2^\le$:] Maximize the probability that  $\lceil \nicefrac{N}{2}\rceil$ jobs are completed within $\nicefrac{N}{4K}$ time units.
\end{itemize}
The objectives have been combined as follows:
($\obj^i$ refers to the objectives considered in Column $i$ of Table~\ref{tab:maRes}):
\begin{align*}
	\obj^1 = (\mathbb{E}_1, \mathbb{E}_2, \mathbb{E}_3) \quad
	\obj^2 = (\mathbb{E}_1, \mathbb{P}^\le_2) \quad
	\obj^3 = (\mathbb{P}, \mathbb{E}_1, \mathbb{E}_2, \mathbb{E}_3) \quad
	\obj^4 = (\mathbb{P}, \mathbb{E}_3, \mathbb{P}^\le_1, \mathbb{P}^\le_2)
\end{align*}

\paragraph{Polling.}
The polling system is based on \cite{Srinivasan:1991,DBLP:conf/concur/TimmerKPS12}.
It considers two stations, each having a separate queue storing up to $K$ jobs of $N$ different types.
The jobs arrive at Station~$i$ (for $i \in \{1,2\}$) with some rate $\rate_i$ as long as the queue of the station is not full.
A server polls the two stations and processes the jobs by (nondeterministically) taking a job from a non-empty queue.
The time for processing a job is given by a rate which depends on the type of the job.
Erasing a  job from a queue is unreliable, i.e., there is a $10\,\%$ chance that an already processed job stays in the queue.
For $i \in \{1,2\}$ we assume the following objectives:
\begin{itemize}
	\item[$\mathbb{E}_i$:] Maximize the expected number of processed jobs of Station~$i$ until its queue is full.
	\item[$\mathbb{E}_{2+i}$:] Minimize the expected sum of all waiting times of the jobs arriving at Station~$i$ until the queue of Station~$i$ is full.
	\item[$\mathbb{P}^\le_i$:] Minimize the probability that the queue of Station~$i$ is full within two time units.
\end{itemize}
The objectives have been combined as follows:
($\obj^i$ refers to the objectives considered in Column $i$ of Table~\ref{tab:maRes}):
\begin{align*}
	\obj^1 = (\mathbb{E}_1, \mathbb{E}_2) \quad
	\obj^2 = (\mathbb{E}_1, \mathbb{E}_2, \mathbb{E}_3, \mathbb{E}_4) \quad
	\obj^3 = (\mathbb{P}^\le_1, \mathbb{P}^\le_2) \quad
	\obj^4 = (\mathbb{E}_1, \mathbb{E}_2, \mathbb{P}^\le_1, \mathbb{P}^\le_2)
\end{align*}

\paragraph{Stream.}
This case study considers a client of a video streaming platform. 
The client consecutively receives $N$ data packages and stores them into a buffer.
The buffered packages are processed during the playback of the video.
The time it takes to receive (or to process) a single package is modeled by an exponentially distributed delay. % with rate $\rate$ (rate $\altrate$).
Whenever a package is received and the video is not playing, the client nondeterministically chooses whether it starts the playback  or whether it keeps on buffering.
The latter choice is not reliable, i.e., there is a $1\,\%$ chance that the playback is started anyway.
In case of a buffer underrun\footnote{A buffer underrun occurs when the next package needs to be processed while the buffer is empty.}, the playback is paused and the client waits for new packages to arrive.
We analyzed the following objectives:
\begin{itemize}
	\item[$\mathbb{E}_1$:] Minimize the expected buffering time until the playback is finished.
	\item[$\mathbb{E}_2$:] Minimize the expected number of buffer underruns during the playback.
	\item[$\mathbb{E}_3$:] Minimize the expected time to start the playback.
	\item[$\mathbb{P}^\le_1$:] Minimize the probability for a buffer underrun within  2 time units.
	\item[$\mathbb{P}^\le_2$:] Maximize the probability that the playback starts within 0.5 time units.
\end{itemize}
The objectives have been combined as follows:
($\obj^i$ refers to the objectives considered in Column $i$ of Table~\ref{tab:maRes}):
\begin{align*}
	\obj^1 = (\mathbb{E}_1, \mathbb{E}_2) \quad
	\obj^2 = (\mathbb{E}_3,  \mathbb{P}_1^\le) \quad
	\obj^3 = (\mathbb{P}^\le_1, \mathbb{P}^\le_2) \quad
	\obj^4 = (\mathbb{E}_1, \mathbb{E}_3, \mathbb{P}^\le_1)
\end{align*}

\paragraph{Mutex.}
This case study regards a randomized mutual exclusion protocol based on  \cite{Pnueli1986,DBLP:conf/concur/TimmerKPS12}.
Three processes nondeterministically choose a job for which they need to  enter the critical section.
The amount of time a process   spends in its critical section is given by a rate which depends on  the chosen job.
There are $N$ different types of jobs.
For each $i \in \{1,2,3\}$ the following objective are considered:
\begin{itemize}
	\item[$\mathbb{P}^\le_i$:] Maximize the probability that Process~$i$ enters its critical section within 0.5 time units.
	\item[$\mathbb{P}^\le_{3+i}$:] Maximize the probability that Process~$i$ enters its critical section within 1 time unit.
\end{itemize}
The objectives have been combined as follows:
($\obj^i$ refers to the objectives considered in Column $i$ of Table~\ref{tab:maRes}):
\begin{align*}
	\obj^1 = (\mathbb{P}^\le_1, \mathbb{P}^\le_2,\mathbb{P}^\le_3) \quad
	\obj^2 = (\mathbb{P}^\le_4, \mathbb{P}^\le_5,\mathbb{P}^\le_6) \quad
\end{align*}

\begin{table}[t]
	\centering
	\caption{Additional model details.}
	\label{tab:maDetails}
	\input{tables/maDetails}
\end{table}

\subsection{Comparison with \prism}
\label{App:evaluation-details:prism}
We considered \prism 4.3.1 obtained from its website \url{www.prismmodelchecker.org}.
We conducted our experiments on \prism with both variants of the value iteration-based implementation (standard and Gauss-Seidel) and chose the faster variant for each benchmark instance. For all experiments the approximation precision $\eta = 0.001$ was considered.

The detailed results are given in Table~\ref{tab:mdp}.
We depict the different benchmark instances with the number of states of the  MDP (Column \emph{\#states}) and the considered combination of objectives ($\mathbb{P}$ represents an (untimed) probabilistic objective, $\mathbb{E}$ an expected reward objective, and $\mathbb{C}^\le$ a step-bounded reward objective).
Column \emph{iter} lists the time required for the iterative exploration of the set of achievable points as described in~\cite{ForejtKPatva12}.
In Column \emph{verif} we depict the verification time -- including the time for the iterations as well as the conducted preprocessing steps.
Column \emph{total} indicates the total runtime of the tool which includes model building time and verification time.
For our implementation, we also list the number of vertices of the obtained under-approximation (Column \emph{pts}). 

During our experiments we observed some issues  considering the implementation in \prism.
For example \prism does not detect that both objectives considered for the \emph{sched.}-instances yield infinite rewards under every possible resolution of non-determinism. Instead of that, \prism gives an incorrect answer.

\begin{table}[t]
	\centering
	\caption{Results for our implementation (\storm) and \prism on the  multi-objective MDP benchmarks from~\cite{ForejtKPatva12}. All run-times are in seconds.}
	\label{tab:mdp}
	\input{tables/mdpresults}
\end{table}

\subsection{Comparison with \imca}
\label{App:evaluation-details:imca}
We consider \imca 1.6 obtained from \url{https://github.com/buschko/imca}. 
The experiments on \imca have been conducted with and without enabling value-iteration and we chose the faster variant for each benchmark instance.
For timed reachability objectives, the precision $\eta = 0.01$ was considered in all experiments.

The resulting verification times are given in Table~\ref{tab:single}.
We depict the different benchmark instances with the number of states of the  MA (Column \emph{\#states}) and the considered objective (as discussed in App.~\ref{App:evaluation-details:ma}).
Besides the run-times of \imca, we depict the run-times of our implementation (effectively performing multi-objective model checking with only one objective) in Column \storm (multi).
Column \storm (single) shows the run-times obtained when \storm is invoked with standard (single-objective)  model checking methods.
 
\begin{table}[t]
	\centering
	\caption{Results for our implementation (\storm) and \imca for single-objective MAs. All run-times are in seconds.}
	\label{tab:single}
	\input{tables/singleobjresults}
\end{table}

%% file: tables/maDetails.tex
\setlength{\tabcolsep}{3pt}
\begin{tabular}{c|crrrrr}
& N(-K)   &   \multicolumn{1}{c}{\#states} & \multicolumn{1}{c}{\#choices} &  \multicolumn{1}{c}{\#transitions} & \multicolumn{1}{c}{\#MS} & \multicolumn{1}{c}{$\rate{\max}$}  \\ \hline
\multirow{3}{*}{\rotatebox[origin=c]{90}{\emph{jobs}}} &
 10-2    & 12\,554                        &      23\,061                  &       34\,581                      &        11\,531                &          5.7                 \\ 
& 12-3    & 116\,814                      &      225\,437                 &       450\,783                     &        112\,719               &          8.5                 \\ 
& 17-2    & 4\,587\,537                   &      8\,912\,931              &       13\,369\,379                 &        4\,456\,466            &          5.9                 \\ \hline
\multirow{3}{*}{\rotatebox[origin=c]{90}{\emph{polling}}}  &
 3-2     & 1\,020                         &      1\,852                   &       2\,477                       &          508                  &          14                  \\ 
& 3-3     & 9\,858                        &      18\,295                  &       24\,536                      &        4\,801                 &          14                  \\ 
& 4-4     & 827\,735                      &      1\,682\,325              &       2\,146\,086                  &        465\,125               &          16                  \\ \hline
\multirow{3}{*}{\rotatebox[origin=c]{90}{\emph{stream}}} &
 30      & 1\,426                         &      1\,861                   &       2\,731                       &          931                  &          8                   \\ 
& 250     & 94\,376                       &      125\,501                 &       187\,751                     &        62\,751                &          8                   \\ 
& 1000    & 1\,502\,501                   &      2\,002\,001              &       3\,001\,001                  &        1\,001\,001            &          8                   \\ \hline
\multirow{2}{*}{\rotatebox[origin=c]{90}{\emph{mutex}}}  &
 2       & 13\,476                        &      31\,752                  &       36\,120                      &          216                  &          2                   \\ 
& 3       & 38\,453                       &      99\,132                  &       111\,687                     &        8\,487                 &          3                   \\   
\end{tabular}

%% file: tables/mdpresults.tex
\begin{tabular}{l|crc|rrr|rrrr}
	                            & \multicolumn{3}{c|}{benchmark}                                   &                   \multicolumn{3}{c|}{\prism}                    &                               \multicolumn{4}{c}{\storm}                               \\
	                             &  instance &    \multicolumn{1}{c}{\#states} &   $\obj$       & \multicolumn{1}{c}{iter} & \multicolumn{1}{c}{verif} & \multicolumn{1}{c|}{total} & \multicolumn{1}{c}{pts} & \multicolumn{1}{c}{iter} & \multicolumn{1}{c}{verif} & \multicolumn{1}{c}{total} \\ \hline
	\multirow{6}{*}{\rotatebox[origin=c]{90}{\emph{consensus}}}   
								 & $2\_3\_2$ &                        691 &     $\mathbb{P,P}$  &  0.019 &      0.183                                                 &     0.285   &       3 &              0.007             &      0.010    &            0.474       \\
	                             & $2\_4\_2$ &                     1\,517 &     $\mathbb{P,P}$  &  0.038 &      0.329                                                 &     0.501   &       2 &              0.012             &      0.017    &            0.497       \\
	                             & $2\_5\_2$ &                    3\,169 &     $\mathbb{P,P}$   &  0.053 &      0.528                                                 &     0.740   &       2 &              0.018             &      0.028    &            0.518     \\
	                             & $3\_3\_2$ &                  17\,455 &     $\mathbb{P,P}$    &  0.232 &      1.416                                                 &     1.771   &       2 &              0.135             &      0.193    &            1.169      \\
	                             & $3\_4\_2$ &                   61\,017 &     $\mathbb{P,P}$   &  0.854 &      4.267                                                 &     4.998   &       2 &              0.499             &      0.806    &            3.421       \\
	                             & $3\_5\_2$ &                 181\,129 &     $\mathbb{P,P}$    &  2.835 &      9.735                                                 &    10.813   &       2 &              1.734             &      3.639    &            10.675        \\ \hline
	\multirow{6}{*}{\rotatebox[origin=c]{90}{\emph{zeroconf(-tb)}}}                                                                                                                                                                                     
	                             &       $4$       &                5\,449 &     $\mathbb{P,P}$ &  0.130 &      6.157                                                 &     6.423   &       2 &              0.077             &      0.146    &            0.830    \\
	                             &       $6$       &              10\,543 &     $\mathbb{P,P}$  &  0.235 &     12.093                                                 &    12.428   &       2 &              0.213             &      0.368    &            1.178    \\
	                             &       $8$       &               17\,221 &     $\mathbb{P,P}$ &  0.408 &     22.143                                                 &    22.596   &       2 &              0.467             &      0.819    &            1.454   \\ 
                                &   $2\_14$    &                29\,572 &     $\mathbb{P,P}$    &  0.285 &     45.715                                                 &    46.311   &       2 &              0.615             &      1.926    &            2.924     \\
	                             &   $4\_10$    &               19\,670 &     $\mathbb{P,P}$    &  0.262 &     40.259                                                 &    40.780   &       2 &              0.568             &      1.256    &            2.052    \\
	                             &   $4\_14$    &              42\,968 &     $\mathbb{P,P}$     &  0.363 &     96.813                                                 &    97.631   &       1 &              2.706             &      6.216    &            7.469    \\ \hline
	\multirow{6}{*}{\rotatebox[origin=c]{90}{\emph{team-form.}}}                                                                                                                                                                                           
%	                            &       $3$       &              12\,475 &     $\mathbb{P,E}$   &  0.162 &     14.486\tablefootnote{The returned approximation of the Pareto curve is too coarse.\label{foot:prismCoarse}}
%				 					   															      		  		                                                  &    14.788   &       6 &              0.160             &      0.257    &            0.877   \\
%	                             &       $4$       &            96\,665 &     $\mathbb{P,E}$    &  1.696 &     351.520\footnotemark[\getrefnumber{foot:prismCoarse}]  &   353.139   &       6 &              1.360             &      6.637    &            9.325   \\
%	                             &       $5$       &          907\,993 &     $\mathbb{P,E}$     &  8.748 & 14\,642.685\footnotemark[\getrefnumber{foot:prismCoarse}]  & 14\,649.290 &       6 &              22.197            &    866.151    &          889.889   \\
	                            &       $3$       &              12\,475 &     $\mathbb{P,E}$   & \multicolumn{3}{c|}{incorrect}                                                   &       5 &              0.160             &      0.257    &            0.877   \\
	                             &       $4$       &            96\,665 &     $\mathbb{P,E}$    &   \multicolumn{3}{c|}{incorrect}                                                 &       3 &              1.360             &      6.637    &            9.325   \\
	                             &       $5$       &          907\,993 &     $\mathbb{P,E}$     &   \multicolumn{3}{c|}{incorrect}                                                 &       3 &              22.197            &    866.151    &          889.889   \\
	                             &       $3$       &            12\,475 &    $\mathbb{P,E,P}$     & \multicolumn{3}{c|}{not supported}
																																				                                    &      10 &               4.060            &     1.432     &       2.020  \\
	                             &       $4$       &           96\,665 &    $\mathbb{P,E,P}$     &    \multicolumn{3}{c|}{not supported}  
																																										            &      13 &               1.327            &     9.447     &       12.256  \\
	                             &       $5$       &         907\,993 &    $\mathbb{P,E,P}$     &     \multicolumn{3}{c|}{not supported} 
																																								                    &      8  &              48.873            &     894.525   &        918.858 \\ \hline
	\multirow{3}{*}{\rotatebox[origin=c]{90}{\emph{sched.}}} 
	                             &       $5$       &           31\,965 &     $\mathbb{E,E}$     &    \multicolumn{3}{c|}{error}  
																																					                                &       &                 ---              &               &        1.214   \\
	                             &      $25$       &        633\,735 &     $\mathbb{E,E}$     &      \multicolumn{3}{c|}{incorrect}
																																											         &      &                 ---              &                &       13.907   \\
	                             &      $50$       &    2\,457\,510 &     $\mathbb{E,E}$     &      \multicolumn{3}{c|}{incorrect}
																																											        &       &                 ---              &                &       53.119    \\ \hline
	\multirow{3}{*}{\rotatebox[origin=c]{90}{\emph{dpm}}}      
	                              &      $100$      &             636 &     $\mathbb{C}^\le,\mathbb{C}^\le$     &    0.187  &           0.228 &       0.298                        &        6 &                0.143              &  0.145    &        0.355 \\
	                             &      $200$      &              636 &     $\mathbb{C}^\le,\mathbb{C}^\le$     &    0.213  &           0.247 &       0.312                        &        4 &                0.210              &  0.213    &        0.433 \\
	                             &      $300$      &              636 &     $\mathbb{C}^\le,\mathbb{C}^\le$     &    0.239  &           0.285 &       0.360                        &        3 &                0.205              &  0.207    &        0.433
\end{tabular}

%% file: tables/singleobjresults.tex
\begin{tabular}{l|crc|r|r|r}
	                                                         & \multicolumn{3}{c|}{benchmark}                                  &      \multicolumn{1}{c|}{\imca}            & \multicolumn{1}{c|}{\storm (multi)}& \multicolumn{1}{c}{\storm (single)}   \\
	                                                         &  instance & \multicolumn{1}{c}{\#states} &   $\obj$             &      \multicolumn{1}{c|}{verif. time}      &   \multicolumn{1}{c|}{verif. time}  & \multicolumn{1}{c}{verif. time}       \\ \hline
\multirow{4}{*}{\rotatebox[origin=c]{90}{\emph{jobs}}}       & 10\_2      &       12\,554               & $\mathbb{E}_1$       &          0.009                             &                0.047                &    0.021                              \\
                                                             & 10\_2     &        12\,554               & $\mathbb{P}^{\le}_2$ &          1.054                             &                2.977                &    1.702                              \\
                                                             & 12\_3     &        116\,814              & $\mathbb{E}_1$       &          0.136                             &                0.556                &    0.279                              \\
                                                             & 12\_3     &        116\,814              & $\mathbb{P}^{\le}_2$ &         19.938                             &               56.242                &    31.682                             \\ \hline
\multirow{4}{*}{\rotatebox[origin=c]{90}{\emph{polling}}}    & 3\_3      &        9\,858                & $\mathbb{E}_1$       &          6.254                             &                0.102                &    0.095                              \\
                                                             & 3\_3      &        9\,858                & $\mathbb{P}^{\le}_1$ &         21.948                             &               54.350                &    14.163                             \\
                                                             & 4\_4      &        827\,735              & $\mathbb{E}_1$       &     3\,630.283                             &               52.162                &    47.746                             \\
                                                             & 4\_4      &        827\,735              & $\mathbb{P}^{\le}_1$ &     3\,424.730                             &           8\,615.390                &    1\,597.095                         \\ \hline
\multirow{4}{*}{\rotatebox[origin=c]{90}{\emph{stream}}}     & 30        &        1\,426                & $\mathbb{E}_1$       &          0.005                             &                0.009                &    0.004                              \\
                                                             & 30        &        1\,426                & $\mathbb{P}^{\le}_1$ &          0.481                             &                1.578                &    0.509                              \\
                                                             & 250       &        94\,376               & $\mathbb{E}_1$       &          2.972                             &                1.462                &    1.261                              \\
                                                             & 250       &        94\,376               & $\mathbb{P}^{\le}_1$ &         36.663                             &              111.450                &    33.527                             \\ \hline
\multirow{2}{*}{\rotatebox[origin=c]{90}{\emph{mutex}}}      & 2         &        13\,476               & $\mathbb{P}^{\le}_1$ &          1.785                             &                1.217                &    0.4                                \\
                                                             & 2         &        13\,476               & $\mathbb{P}^{\le}_4$ &          6.922                             &                4.118                &    1.008                              \\
\end{tabular}

%% file: main.bbl
\begin{thebibliography}{10}

\bibitem{DBLP:conf/lics/EisentrautHZ10}
Eisentraut, C., Hermanns, H., Zhang, L.:
\newblock On probabilistic automata in continuous time.
\newblock In: Proc. of LICS, IEEE CS (2010)  342--351

\bibitem{DBLP:journals/iandc/DengH13}
Deng, Y., Hennessy, M.:
\newblock On the semantics of {M}arkov automata.
\newblock Inf. Comput. \textbf{222} (2013)  139--168

\bibitem{DBLP:journals/tdsc/BoudaliCS10}
Boudali, H., Crouzen, P., Stoelinga, M.:
\newblock A rigorous, compositional, and extensible framework for dynamic fault
  tree analysis.
\newblock {IEEE} Trans. Dependable Sec. Comput. \textbf{7}(2) (2010)  128--143

\bibitem{DBLP:conf/cav/CosteHLS09}
Coste, N., Hermanns, H., Lantreibecq, E., Serwe, W.:
\newblock Towards performance prediction of compositional models in industrial
  {GALS} designs.
\newblock In: Proc.\ of CAV. Vol. 5643 LNCS, Springer (2009)  204--218

\bibitem{KatoenWu16}
Katoen, J.P., Wu, H.:
\newblock Probabilistic model checking for uncertain scenario-aware data flow.
\newblock {ACM} Trans. Embedded Comput. Sys. \textbf{22}(1) (2016)  15:1--15:27

\bibitem{DBLP:journals/cj/BozzanoCKNNR11}
Bozzano, M., Cimatti, A., Katoen, J.P., Nguyen, V.Y., Noll, T., Roveri, M.:
\newblock Safety, dependability and performance analysis of extended {AADL}
  models.
\newblock Comput. J. \textbf{54}(5) (2011)  754--775

\bibitem{DBLP:conf/apn/EisentrautHK013}
Eisentraut, C., Hermanns, H., Katoen, J.P., Zhang, L.:
\newblock A semantics for every {GSPN}.
\newblock In: Petri Nets. Vol. 7927 LNCS, Springer (2013)  90--109

\bibitem{HatefiH12}
Hatefi, H., Hermanns, H.:
\newblock Model checking algorithms for {M}arkov automata.
\newblock {ECEASST} \textbf{53} (2012)

\bibitem{DBLP:journals/corr/GuckHHKT14}
Guck, D., Hatefi, H., Hermanns, H., Katoen, J.P., Timmer, M.:
\newblock Analysis of timed and long-run objectives for {M}arkov automata.
\newblock LMCS \textbf{10}(3) (2014)

\bibitem{DBLP:conf/atva/GuckTHRS14}
Guck, D., Timmer, M., Hatefi, H., Ruijters, E., Stoelinga, M.:
\newblock Modelling and analysis of {M}arkov reward automata.
\newblock In: Proc.\ of ATVA. Vol. 8837 LNCS, Springer (2014)  168--184

\bibitem{DBLP:conf/setta/HatefiBWFHB15}
Hatefi, H., Braitling, B., Wimmer, R., Fioriti, L.M.F., Hermanns, H., Becker,
  B.:
\newblock Cost vs. time in stochastic games and {M}arkov automata.
\newblock In: {Proc. of SETTA}. Vol. 9409 LNCS, Springer (2015)  19--34

\bibitem{Wimmeretal17}
Butkova, Y., Wimmer, R., Hermanns, H.:
\newblock Long-run rewards for {M}arkov automata.
\newblock In: Proc.\ of TACAS. LNCS, Springer (2017) To appear.

\bibitem{DBLP:journals/jacm/BrunoDF81}
Bruno, J.L., Downey, P.J., Frederickson, G.N.:
\newblock Sequencing tasks with exponential service times to minimize the
  expected flow time or makespan.
\newblock J. {ACM} \textbf{28}(1) (1981)  100--113

\bibitem{KNP11}
Kwiatkowska, M., Norman, G., Parker, D.:
\newblock \textsc{Prism} 4.0: Verification of probabilistic real-time systems.
\newblock In: Proc.\ of CAV. Vol. 6806 LNCS, Springer (2011)  585--591

\bibitem{ForejtKPatva12}
Forejt, V., Kwiatkowska, M., Parker, D.:
\newblock {P}areto curves for probabilistic model checking.
\newblock In: Proc.\ of ATVA. Vol. 7561 LNCS, Springer (2012)  317--332

\bibitem{DBLP:journals/jair/RoijersVWD13}
Roijers, D.M., Vamplew, P., Whiteson, S., Dazeley, R.:
\newblock A survey of multi-objective sequential decision-making.
\newblock J. Artif. Intell. Res. \textbf{48} (2013)  67--113

\bibitem{EtessamiKVY08}
Etessami, K., Kwiatkowska, M.Z., Vardi, M.Y., Yannakakis, M.:
\newblock Multi-objective model checking of {M}arkov decision processes.
\newblock LMCS \textbf{4}(4) (2008)

\bibitem{ForejtKNPQtacas11}
For\v{e}jt, V., Kwiatkowska, M.Z., Norman, G., Parker, D., Qu, H.:
\newblock Quantitative multi-objective verification for probabilistic systems.
\newblock In: Proc.\ of TACAS. Vol. 6605 LNCS, Springer (2011)  112--127

\bibitem{DBLP:conf/stacs/BruyereFRR14}
Bruy{\`{e}}re, V., Filiot, E., Randour, M., Raskin, J.F.:
\newblock Meet your expectations with guarantees: Beyond worst-case synthesis
  in quantitative games.
\newblock In: Proc.\ of STACS. Vol.~25 LIPIcs, Schloss Dagstuhl -
  Leibniz-Zentrum fuer Informatik (2014)  199--213

\bibitem{DBLP:conf/csl/BaierDK14}
Baier, C., Dubslaff, C., Kl{\"{u}}ppelholz, S.:
\newblock Trade-off analysis meets probabilistic model checking.
\newblock In: {CSL-LICS}, {ACM} (2014)  1:1--1:10

\bibitem{DBLP:journals/jcss/BrazdilCFK17}
Br{\'{a}}zdil, T., Chatterjee, K., Forejt, V., Kucera, A.:
\newblock Trading performance for stability in {M}arkov decision processes.
\newblock J. Comput. Syst. Sci. \textbf{84} (2017)  144--170

\bibitem{DBLP:journals/corr/abs-1104-3489}
Br{\'{a}}zdil, T., Brozek, V., Chatterjee, K., Forejt, V., Kucera, A.:
\newblock Markov decision processes with multiple long-run average objectives.
\newblock LMCS \textbf{10}(1) (2014)

\bibitem{DBLP:conf/tacas/BassetKTW15}
Basset, N., Kwiatkowska, M.Z., Topcu, U., Wiltsche, C.:
\newblock Strategy synthesis for stochastic games with multiple long-run
  objectives.
\newblock In: Proc.\ of TACAS. Vol. 9035 LNCS, Springer (2015)  256--271

\bibitem{DBLP:conf/ecai/Teichteil-Konigsbuch12}
Teichteil{-}K{\"{o}}nigsbuch, F.:
\newblock Path-constrained {M}arkov decision processes: bridging the gap
  between probabilistic model-checking and decision-theoretic planning.
\newblock In: Proc.\ of ECAI. Vol. 242 Frontiers in AI and Applications, {IOS}
  Press (2012)  744--749

\bibitem{DBLP:conf/vmcai/RandourRS15}
Randour, M., Raskin, J.F., Sankur, O.:
\newblock Variations on the stochastic shortest path problem.
\newblock In: Proc.\ of VMCAI. Vol. 8931 LNCS, Springer (2015)  1--18

\bibitem{DBLP:conf/tacas/Junges0DTK16}
Junges, S., Jansen, N., Dehnert, C., Topcu, U., Katoen, J.P.:
\newblock Safety-constrained reinforcement learning for mdps.
\newblock In: Proc.\ of TACAS. Vol. 9636 LNCS, Springer (2016)  130--146

\bibitem{DBLP:conf/atva/DavidJLLLST14}
David, A., Jensen, P.G., Larsen, K.G., Legay, A., Lime, D., S{\o}rensen, M.G.,
  Taankvist, J.H.:
\newblock On time with minimal expected cost!
\newblock In: Proc.\ of ATVA. Vol. 8837 LNCS, Springer (2014)  129--145

\bibitem{DBLP:conf/mfcs/ChenFKSW13}
Chen, T., Forejt, V., Kwiatkowska, M.Z., Simaitis, A., Wiltsche, C.:
\newblock On stochastic games with multiple objectives.
\newblock In: Proc.\ of MFCS. Vol. 8087 LNCS, Springer (2013)  266--277

\bibitem{Put94}
Puterman, M.L.:
\newblock {{M}arkov} Decision Processes: Discrete Stochastic Dynamic
  Programming.
\newblock John Wiley and Sons (1994)

\bibitem{DBLP:conf/fossacs/NeuhausserSK09}
Neuh{\"{a}}u{\ss}er, M.R., Stoelinga, M., Katoen, J.P.:
\newblock Delayed nondeterminism in continuous-time {M}arkov decision
  processes.
\newblock In: Proc.\ of FOSSACS. Vol. 5504 LNCS, Springer (2009)  364--379

\bibitem{DBLP:journals/corr/DehnertJK016}
Dehnert, C., Junges, S., Katoen, J.P., Volk, M.:
\newblock A {Storm} is coming: A modern probabilistic model checker.
\newblock In: Proc.\ of CAV. (2017)

\bibitem{quatmannMastersThesisMultiObjMA}
{Tim Quatmann}:
\newblock {Multi-objective model checking of Markov Automata}.
\newblock Master's thesis, RWTH Aachen University (2016)

\bibitem{HaddadM14}
Haddad, S., Monmege, B.:
\newblock Reachability in {MDPs}: Refining convergence of value iteration.
\newblock In: {RP}. Vol. 8762 LNCS, Springer (2014)  125--137

\bibitem{Srinivasan:1991}
Srinivasan, M.M.:
\newblock Nondeterministic polling systems.
\newblock Management Science \textbf{37}(6) (1991)  667--681

\bibitem{DBLP:conf/concur/TimmerKPS12}
Timmer, M., Katoen, J.P., van~de Pol, J., Stoelinga, M.:
\newblock Efficient modelling and generation of {M}arkov automata.
\newblock In: Proc.\ of CONCUR. Vol. 7454 LNCS, Springer (2012)  364--379

\bibitem{Pnueli1986}
Pnueli, A., Zuck, L.:
\newblock Verification of multiprocess probabilistic protocols.
\newblock Distributed Computing \textbf{1}(1) (1986)  53--72

\bibitem{Neuhausser2010}
Neuh{\"{a}}u{\ss}er, M.R.:
\newblock Model checking Nondeterministic and Randomly Timed Systems.
\newblock PhD thesis, {RWTH} Aachen University (2010)

\bibitem{ash2000probability}
Ash, R.B., Dol{\'e}ans-Dade, C.:
\newblock Probability and Measure Theory.
\newblock Harcourt/Academic Press (2000)

\end{thebibliography}
